\documentclass[a4paper]{report}
\usepackage[utf8]{inputenc}
\usepackage[T1]{fontenc}
\usepackage{MyRJ}
\usepackage{amsmath,amssymb,array}
\usepackage{booktabs}
\usepackage{xspace}
\usepackage{cleveref}

\usepackage{amsmath, amsthm, amssymb}
\newtheorem{proposition}{Proposition}
\newtheorem{remark}{Remark}
\newtheorem{Example}{Example}

\usepackage{listings} 

\lstdefinestyle{mystyle}
{
    language = R,
    basicstyle=\footnotesize,
   numbers=left,
    backgroundcolor = {\color{backcolour}},
    stringstyle = {\color{codegreen}},
     identifierstyle={\color{green!40!black}},
    keywordstyle = {\color{magenta}},
    keywordstyle = [2]{\color{blue}},
    keywordstyle = [3]{\color{yellow}},
    keywordstyle = [4]{\color{teal}},
    otherkeywords = {in,=,@,$, subsample, spli, net, seq2,gather,li,ncomb,gramian, project, colMeans},
    morekeywords = [2]{in,=,@,$},
    morekeywords = [3]{},
    literate={Bsplines}{{\color{green!40!black} Bsplines \color{black}}}{7}
   }
\definecolor{codegreen}{rgb}{0,0.6,0}
\definecolor{backcolour}{rgb}{0.95,0.95,0.92}

\definecolor{darkblue}{HTML}{2C326E}
\definecolor{mygreen}{rgb}{0,0.6,0}
\definecolor{mygray}{rgb}{0.5,0.5,0.5}
\definecolor{mymauve}{rgb}{0.58,0,0.82}

\begin{document}

\begin{article}
\title{\sc{\tt Splinets} -- splines through the Taylor expansion, their support sets and 
orthogonal bases\vspace{1cm}}
\author{Krzysztof Podg\'orski \\
Department of Statistics\\
Lund University
}

\maketitle

\abstract{
A new representation of splines that targets efficiency in the analysis of functional data is implemented. 
The efficiency is achieved through two novel features:  using the recently introduced orthonormal spline bases, the so-called {\it splinets} and accounting for the spline support sets in the proposed spline object representation. 
The recently-introduced orthogonal splinets are evaluated by {\it dyadic orthogonalization} of the $B$-splines.
The package is built around the {\it Splinets}-object that represents a collection of splines.
It treats splines as mathematical functions and contains information about the support sets and the values of the derivatives at the knots that uniquely define these functions.
Algebra and calculus of splines utilize the local Taylor expansions at the knots within the support sets. 
Several orthonormalization procedures of the $B$-splines are implemented including the recommended dyadic method leading to the splinets. 
The method bases on a dyadic algorithm that can be also viewed as the efficient method of diagonalizing a band matrix. 
The locality of the $B$-splines in terms of the support sets is, to a great extend, preserved in the corresponding splinet. 
This together with implemented algorithms utilizing locality of the supports provides a valuable computational tool for the functional data analysis. 
The benefits are particularly evident when the sparsity in the data plays important role. 
Various diagnostic tools are provided allowing to maintain stability of the computations. 
Finally, the projection operation to the space of splines is implemented that facilitates functional data analysis. 
An example of simple functional analysis of the data using the tools in the package is presented.
The functionality of the package extends beyond the splines to piecewise polynomial functions, although the splines are its focus.  
}

\section{Introduction}
\label{sec:intro}
In functional analysis, it is often desired that functions considered are continuous or even differentiable up to a certain order. 
From this perspective, spline functions given over a set of knots form convenient finite-dimensional functional spaces.  
There exist many R-packages that handle splines, see \cite{Perperoglou:2019aa}. 
However, none of them treats consistently as elements of functional spaces with explicitly evaluated orthogonal bases. 
The proposed package, {\tt Splinets},  approaches to splines exactly like this by using an object that represents a set of functional splines and providing efficient orthogonal bases as such objects.
This focus on functional form of splines and functional analysis approach to them makes {\tt Splinets} particularly suitable for functional data analysis.  
The care was taken to make this functional treatment efficient through first, accounting for support sets of spline, then by using efficient and novel orthogonal bases. 

For a given set of knots and a smoothness order, a spline between two subsequent knots is equal to a polynomial that is smoothly connected at the knots to the polynomials over the neighboring intervals. 
The order of smoothness at the knots is equal to the number of derivatives that are continuous at these knots, including the zero order derivative, i.e. the function itself.
A selection of an order that is higher than zero makes splines a smoother alternative to the piecewise constant functions.
Given a smoothness order and knots, the splines form a finite dimensional functional space and one can consider a suitable basis of functions that spans it.  
The $B$-splines are the most popular bases of splines \cite{Boor1978APG, schumaker2007spline}. 
Their locality expressed by their support, i.e. the sets over which they are non-zero is one of their main features. 
In order to utilize this feature in numerical implementations, one would have to keep track of the support sets of splines. 
In the presented implementation of splines, the object stores the information about the support sets.
More specifically, the spline object contains a sequence of disjoint intervals such that their union is equal to the support set. 
The object stores also values of the derivatives at the knots which allows to use the Taylor expansion at the knots to evaluate splines at any point. 
Since the Taylor representation is local with respect to knots it is characterized by numerical stability in the computations.

More specifically, the main object {\tt Splinets} implemented through the S4 system for the OOP in {\sf R} is
 defined through {\tt setClass} function with the following slots
{\footnotesize \begin{verbatim}
representation(
           knots="vector", smorder="numeric", equid="logical", type = "character",
           supp="list", der="list", taylor = "matrix", type = "character", epsilon="numeric"
                    ),
         \end{verbatim}
         }
\noindent which represents a collection of splines,  all built over the same knots given in {\tt knots}, of the smoothness order $k$ given in {\tt smorder}.
Further {\tt supp}  is the list of matrices having row-wise pairs of the endpoints of the intervals, the union of which constitutes the support set of a particular spline,  and the flag {\tt equid} informs about the equally placed knots, for which the computation can be significantly accelerated. 
The matrices of the derivatives at the knots inside the support sets are given in the list {\tt der} of 
 matrices, where an element in the list is referring to a particular spline in our collection, the length of the list corresponds to the number of splines in the object.  
Descriptions of other fields are given in the {\tt Splinets}-package but are not crucial for this presentation.

While the locality of the representation expressed by the support sets are making the package suitable for analysis of sparse data, its most important feature  is efficiently using the spline bases. 
The standard $B$-bases are evaluated and the efficiency is achieved by explicitly using the support in implemented algebra and calculus of splines. 
There is, however, one problem with the $B$-splines bases that adds a computational burden when they are used to decompose a function and for functional data analysis. 
Namely, the $B$-splines are not orthogonal. 
Since any basis in a Hilbert space can be orthogonalized, it is to the point to consider orthonormalization of the $B$-splines,  \cite{nguyen2015construction,goodman2003class,cho2005class}.
In \cite{Liu2019SplinetsE}, a new natural orthogonalization method for the $B$-splines has been introduced. 
The orthogonalization was argued to be the most appropriate since it, firstly, preserves most from the original structure of the $B$-splines and, secondly, obtains computational efficiency both in the basis element evaluations and in the spectral decomposition of a functional signal.

Although fundamentally different, the orthogonalization method was inspired by one-sided and two-sided orthogonalization discussed in \cite{mason1993orthogonal}.
Since the constructed basis spreads a net of splines rather  than a sequence of them, we coin the term {\it splinet} when referring to such a base.  
It is formally shown that the splinets, similarly to the $B$-splines, feature locality with a small size of the total support. 
If the number of knots over which the splines of a given order are considered is $n$, then
the total support size of the $B$-splines is on the order $O(1)$ with respect to $n$, while the corresponding splinet has the total support size on the order of $\log n$ which is only slightly bigger.
On the other hand, the previously discussed orthogonalized bases have the total support size of the order $O(n)$, where $n$ stands for the number of knots which is also the number of basis functions.
Moreover, if one allows for negligible errors in the orthonormalization, then the total support size no longer will depend on $n$, i.e. becomes constant and thus achieving the rate of the original $B$-splines.

The functionality of the package extends beyond the splines since the {\tt Splinets}-object can be used to represent any piecewise polynomial function of a given order. 
All the relevant functions work properly even if the smoothness at the knots is not preserved. 
Various methods of correcting a piecewise polynomial functions to make them splines are implemented, including the orthogonal projection to the space of splines. 
Nevertheless, the focus of both the package and this presentation is on the splines. 

The organization of the material is as follows. We start with the functional representation of splines used in the package implementation. 
Then we describe the generic functions used to build spline objects from the piecewise polynomial functions. 
This is accomplished by correcting the values of the derivatives to assure the continuity at the knots.
Using these correction methods, a simple random spline generator is obtained. 
In the next section, we turn to a discussion of algebraic and functional operations on splines including integral, derivatives, and inner products. 
The main section is about the bases of splines, where the orthonormal spline bases are at the center of the discussion. 
Here, we promote the orthonormal bases that are obtained from the $B$-splines by efficient dyadic algorithms and  referred to as {\it splinets}. 
The presentation of the important implementation of the projection of functions to the space of splines by the means of the developed spline bases follows. 
The paper is concluded with two examples that utilized the package. 
In the first, we show how the package can be used to obtain an alternative efficient orthonormalization of the $B$-spline basis in the non-dyadic case.
In the second, the tools in the package are used to perform a functional data analysis by means of the splines using a functional dataset that is also a part of the package. 

\section{Mathematical foundations of a functional representation of splines}
Splines are piecewise polynomials of a given smoothness order with continuous derivatives, up to this order (exclusive) at the points they interconnect which are called {\it internal knots}. 
The domain of a spline will be referred to as its {\it range} and  is assumed across the paper to be a closed finite interval. 
Additional knots called the {\it terminal knots} are the endpoints of the spline range. 
For a given set of knots, the space of splines is finite-dimensional with the inner product of the Hilbert space of the square-integrable functions. 

Due to the continuity requirements, the behavior of a spline between two given knots is necessarily affected by the form of polynomials at the neighboring between-knots subintervals. 
Since the between-knots intervals with terminal knots have only one neighboring between-knots interval, the influence of values over other intervals is not the same.   
To mitigate this biased terminal knots effect, it is natural and, as it will be seen, also mathematically elegant to introduce the zero boundary conditions at the terminal knots for all the derivatives except the highest order one. 
It can be shown that to remove the zero boundary effect, one has to consider splines over the knots obtained by extending a certain number of the knots from both the ends of the complete set of the knots. 
More specifically, the number of knots that has to be added at each end is equal to the order of the splines. 
Often the knots are added by replicating the terminal knots although there are some serious disadvantages of such an approach.
All this is elaborated in full detail below in Proposition~\ref{prop:dim} and in the remark that follows it. 

The most natural and computationally stable way of evaluating values of a given spline is through the Taylor expansion at the closest knot. 
For this, it is convenient to have all derivatives at a knot directly accessible. 
To achieve this, the proposed implementation of a functional spline uses an object that holds the matrix of the derivatives at the knots. 
Such a matrix needs to satisfy certain conditions to guarantee that the derivatives at the knots are continuous. 
In this section, we recap mathematical properties of splines that are fundamental for the implementation.

\subsection{Splines with zero-boundary conditions at the endpoints}
The splines involve knots at which the polynomials smoothly connect. 
A set of such knots is represented as a vector $\boldsymbol \xi$ of $n+2$ ordered values. 
There are two alternative but in a certain sense equivalent requirements on the behavior of a spline at the endpoints of its range.
In the first one, no boundary conditions are imposed. 
The main problem in this unrestricted approach is that the polynomials at both the ends of the considered range do not `sense' the same restrictions from the neighbors as the polynomial residing further from both the endpoints. 
This is due to the fact that at the ends the `neighbors'  are only present from one side.
Another approach that is favored in this work is putting the zero boundary restrictions on the derivatives at the endpoints.
 It is  mathematically equivalent to the first one in a limiting sense when the $k$ initial knots and the $k$ terminal knots converge to the beginning and the end of the range, respectively.  
 Moreover and most importantly, the approach is structurally elegant and because of it easier to implement.
For all these reasons, we used it in our package. In what follows, we elaborate some mathematical detail and notational conventions.

Through the rest of the paper, we impose on a spline and all its derivatives of the order smaller than the spline smoothness order  the value of zero  at both the  endpoints of the range. 
In this case, if we consider knots $\boldsymbol \xi=(\xi_{0},\dots, \xi_{n+1})$
it is important to assume that  $n\ge k$, in order to have enough knots to define at least one non-zero spline with the $2k$ zero-boundary conditions at the endpoints.
Indeed, if $n=k-1$, then we have $k+1$-knots yielding $k$ between knot intervals. 
On each such interval  a spline is equal to a polynomial of order $k$.
The dimension of the space of such piecewise polynomial functions is $k(k+1)$. However, at each internal knot there are $k$ equations to make derivative up to order $k$ (excluding) continuous. 
This reduces the initial unrestricted polynomials dimension by $(k-1)k$ dimensions to $2k$, but there are $2k$ equations for the derivatives to be zero at the endpoints.
We conclude that the dimension of the spline space is eventually reduced to zero, meaning that there is only a function trivially equal to zero in this space.

From now on $\mathcal S^{\boldsymbol \xi}_{ k}$ stands for the $n-k+1$ dimensional space of the $k$-smoothed splines with the zero boundary conditions at the terminal knots of in the ordered knots given in vector $\boldsymbol \xi$.
Whenever, showing the dependence on either $k$ or $\boldsymbol \xi$ or both is not important, they will be dropped from the notation and thus, for example, $\mathcal S$ stands for $\mathcal S^{\boldsymbol \xi}_{ k}$ if both $k$ and $\boldsymbol \xi$ are clear from the context. 

The requirement that $n\ge k$ allows to obtain non-trivial linear space, however, in order to retrieve the flexibility of the unrestricted splines one needs even more knots. 
It is explained in the next result whose practical and theoretical consequences are discussed in the remark following the result. 
\begin{proposition}
\label{prop:dim}
Let $\boldsymbol \zeta=(\zeta_0,\dots,\zeta_{n'+1})$ be the knots considered for the $k+n'+1$ dimensional linear space of the splines of the $k$th order without zero boundary conditions. 
Further, let $n=n'+2k$ and 
$$
\boldsymbol \xi=(\xi_0,\dots,\xi_{n+1})=(\xi_0,\dots,\xi_{k-1},\zeta_0,\dots,\zeta_{n'+1},\xi_{k+n'+2},\dots,\xi_{n'+2k+1}).
$$
The dimension of the space of the splines of the $k$th order with the zero boundary conditions at the endpoints over $\boldsymbol \xi$ is equal to $k+n'+1$, i.e. the same as that of the unrestricted splines over $\boldsymbol \zeta$. 
Moreover, the space of the splines with the boundary conditions when restricted to $\boldsymbol \zeta$ coincides with the space of the unrestricted splines over $\boldsymbol \zeta$.
\end{proposition}
\begin{proof}
The result easily follows from the following count of the dimension of the splines with the imposed boundary conditions: there are $n+1$ intervals with polynomials having $k+1$ coefficients for the total $(n+1)(k+1)$ coefficients.
To count free coefficients,  one subtracts the $2k$ initial conditions and the $nk$  continuity conditions at the $n$ internal knots. This yields the dimension $n+1-k=n'+k+1$. 
\end{proof}
\begin{remark}
From a practical point of view, considering functional analysis using unrestricted splines of the $k$-th order over an interval  $[a,b]$ is equivalent to using the splines with the zero boundary conditions at the ends by adding $k$ knots below $a$ and $k$ knots above $b$. 
This is because by the above result, when the splines with the boundary conditions and defined over the extended set of knots are restricted in the argument to $[a,b]$, they produces the same linear space as the space of unrestricted splines over this interval.

It should be noted, however,  that the topological properties of the two spaces differ because the inner products will differ at the extended boundaries. 
One can however obtain the topological equivalence either by converging the limit of the added knots to $a$ for the lower knots and to $b$ for the upper knots or by modifying the inner product of the zero-boundary condition splines by taking the integral only over $[a,b]$. 

We conclude that using the splines with the zero-boundary conditions is by no means restrictive and from now on we consider only this case. 
\end{remark}

The splines can be represented in a variety ways. In \cite{Qin}, a general matrix representation was proposed that allows for efficient numerical processing of the operations on the splines. 
This was utilized in Zhou et al. (2012) \cite{Zhou} and Reed \cite{Redd} to represent and orthogonalize $B$-splines that was implemented in the {\sf R}-package \href{https://CRAN.R-project.org/package=orthogonalsplinebasis}{\it Orthogonal B-Spline Basis Functions}. 
In our approach we propose to represent a spline in a different way.
Namely, we focus on the values of the derivatives at knots and the support of a spline.  
The goal is to achieve a better numerical stability as well as to utilize the discussed efficiency of base splines having support only on  small portion of the considered domain.
Our approach constitutes the basis for spline treatment in the package {\tt Splinets} that accompanies this work.

The fundamental fact that we use here is that for a given order, say $k$, and a vector of knot points $\boldsymbol \xi=\left(\xi_0,\dots, \xi_{n+1}\right)$, a spline $S\in \mathcal S$ is uniquely defined by  the values of the first $0,\dots, k$ derivatives at the knots.
Here, by a natural convention that we use across the paper, the $0$ derivative is the function itself. 
The values of derivatives at the knots allow for the Taylor expansions at the knots but they cannot be taken arbitrarily due to the smoothness at the knots. 
Since our computational implementation of the spline algebra fundamentally depends on the relation between the matrix of derivatives values and the splines, the restrictions for the derivative matrix needs to be addressed in full detail.
\subsection{Matrix representation of splines}
The values of the derivatives at the knots are at the center of the Taylor expansion representation of a spline.  
Therefore the matrix of the derivative values has became the main component of an object belonging to our main class in the package.
However, due to discontinuity at knots of the highest order derivative (the order of the smoothness of the considered splines), there is some ambiguity in the representation of such a matrix. 
There are two natural representations, one is the symmetric with respect to the endpoints of the spline support and the other is the one-sided, for example the right-hand-sided (RHS) limits. 
There are some benefits each of these conventions and therefore we use both as described below. 

From now on, we assume that we consider a spline with the full support range, i.e. not vanishing on any subinterval of the entire range of knots. 
For any spline function $S\in \mathcal S_k^{\boldsymbol \xi}$, we consider $\mathbf s_j=(s_{0j},\dots, s_{n+1j})$ is an $n+2$-dimensional  vector (column) of values of the $j^{\rm th}$-derivative of $S$, $j=0,\dots, k$, at the knots given in vector $\boldsymbol \xi= \left(\xi_0,\dots, \xi_{n+1}\right)$ that, as a general convention for all vectors (also the convention in {\sf R}),  is treated as a $(n+2)\times 1$ column matrix. 
These columns are kept in a $(n+2)\times (k+1)$ matrix 
\begin{equation}
\label{eq:spmat}
\mathbf S\stackrel{def}{=}\left [ \mathbf s_0 \mathbf s_1 \dots  \mathbf s_k \right]
\end{equation} 

Since the derivative of the $k^{\rm th}$ order is not continuous at the knots while constant between them, one needs some convention how to define the values in $\mathbf s_k$.
In  the symmetric with respect to the endpoints of the support approach, one  considers the RHS limits for the LHS half of the knots and the LHS for the RHS half of the knots, where knots are considered only within the spline support range. 
This is the main convention that is adopted in the object and in the package since the symmetry is more natural for the splines with the zero boundary conditions on both the ends of the range.
Consequently, we assume that  the value of $s_{ik}$  is the RHS $k$th derivative at the knots $\xi_i$, $i \le n/2$ and  the LHS derivative at the knots $\xi_i$, $i \ge n/2 +1$.
If $n$ is odd, $n=2l+1$ for some integer $l$,  the undefined yet value of $s_{l+1k}$ is set to zero while the LHS and RHS values of the $k$-derivatives at $\xi_{l+1k}$ coincide with $s_{lk}$ and $s_{l+2k}$, respectively. 
We note also that if $n$ is even and $n=2l$, then $s_{lk}=s_{l+1k}$. 
\begin{remark}
We would like to point out that the one-sided convention, another way of treating the highest order derivative, is often more convenient computationally since it does not depend on the support range of a particular spline.
The value of the $k$th derivative at a knot is considered as the right
hand side limit except for the last knot $\xi_{n+1}$ where it is assumed to be equal to zero, as there are no values on the right hand side of $\xi_{n+1}$. 
The advantage of the one-sided approach is the isomorphic relation between functional space and the space of the derivative matrices can be easily retrieved independently of the spline support range. Therefore, this approach will be applied for defining different algebraic and analytic operations on splines, see the section on the algebra and calculus of splines.
The two forms of the matrices, symmetric and one-sided, being equivalent can be easily transformed forth and back and a convenient function
 {\tt sym2one()} is available in the package. The concepts are illustrated in Example~\ref{matrices}.
\end{remark}

In general,  $\mathbf S$ lies in the $(n+2)(k+1)$ dimensional linear space of the $(n+2)\times (k+1)$ matrices.
However, the matrices corresponding to legitimate splines occupy only a proper subspace that correspond to the $n-k+1$ dimensional space of splines. 
This restricted subspace can be expressed by certain linear relations that the entries of $\mathbf S$ need to satisfy.
To obtain them explicitly,  we note that $S\in \mathcal S$  over interval $(\xi_i,\xi_{i+1}]$, $i=0,\dots, n$, is given through its Taylor expansions
$$
S(t)
=\sum_{j=0}^{k-1} s_{i j} \frac{(t-\xi_{i})^j}{j!}+s_{i+\delta_i\,k}\frac{(t-\xi_{i})^k}{k!}
=\sum_{j=0}^{k-1} s_{i+1\, j} \frac{(t-\xi_{i+1})^j}{j!}+s_{i+\bar \delta_i\,k}\frac{(t-\xi_{i+1})^k}{k!},
$$
where $\delta_i=\mathbb I_{(n/2,n+1]}(i)$ and $\bar \delta_i=1-\delta_i$. Here and in what follows $\mathbb I_A$ is the indicator function of a set $A$. 
Similarly, the derivatives, $r=1,\dots,k$, are given through
\begin{align*}
S^{(r)}(t)&=\sum_{j=0}^{k-r-1} s_{i\, j+r} \frac{(t-\xi_{i})^j}{j!} +s_{i+\delta_i\,k}\frac{(t-\xi_{i})^{k-r}}{(k-r)!}
 =
\sum_{j=0}^{k-r-1} s_{i+1\, j+r} \frac{(t-\xi_{i+1})^j}{j!}+s_{i+\bar\delta_i\,k}\frac{(t-\xi_{i+1})^{k-r}}{(k-r)!}. 
\end{align*}

To express the symmetry, we divide the knots into the left ones $\boldsymbol \xi^{L}=(\xi_0,\dots, \xi_{l+1})$ and  introduce the following notation for the right half ones in the reverse order
\begin{equation}
\label{eq:RHSk}
\left(\xi_0^{R},\dots, \xi_{l+1}^{R}\right)=\left(\xi_{n+1},\dots, \xi_{n-l}\right),
\end{equation}
where $l=[n/2]$. We observe that for $n$ even we have $\xi_{l}=\xi^{R}_{l+1}$ and $\xi_{l+1}=\xi^{R}_{l}$, and for $n$ odd  $\xi_{l+1}=\xi^{R}_{l+1}$.
Let us also define $\mathbf S^{R}$ as the bottom half of the matrix $\mathbf S$ in the reverse order through
$$
s_{ij}^{R}=s_{n+1-i\,j}, ~~ i=0,\dots, l+1, ~~ j=0,\dots, k.
$$
The restrictive relations following from the above Taylor expansions can be split into  the LHS and the RHS knots, for  $i=0,\dots, l, ~~r=0,\dots, k-1$ as follows
\begin{equation}
\label{eq:addk-1nL}
\begin{split}
s_{i+1r}=\sum_{j=0}^{k-r} \frac{(\xi_{i+1}-\xi_{i})^{j}}{j!}s_{i\, j+r},& ~~~s^R_{i+1r}=\sum_{j=0}^{k-r} \frac{(\xi^R_{i+1}-\xi^R_{i})^{j}}{j!}s^R_{i~j+r}.
\end{split}
\end{equation}
We observe that for even $n$, the relations for the knots $\xi_{l+1}$ and $\xi^R_{l+1}$ are equivalent due to  $\xi_{l}=\xi^{R}_{l+1}$ and $\xi_{l+1}=\xi^{R}_{l}$. 
Consequently, the number of the above equations is equal to $2(l+1)k-k= nk + k$ so that the dimension of the space of the matrices satisfying them is equal to $(n+2)(k+1)-1-nk - k=n+k+1$ (we also account in this count for  $s_{l k}=s_{l+1\, k}$). From this number, one has to subtract $2k$ for the zero boundary conditions at both the ends yielding $n-k+1$ for the dimension of the space of the matrices, as required. 
Similar count can be made for the odd number $n=2(l+1)$.  
Namely, from (\ref{eq:addk-1nL}), we reduce the full dimension $(n+2)(k+1)$ by $2(l+1)k+1=nk+k+1$ (accounting here also for $s_{l+1 k}=0$) to obtain the reduced dimension equal to $n+k+1$ which then reduce by $2k$ due to the zero boundary conditions at the endpoints. 
From now on, the restricted space of matrices is identified with $\mathcal S$ and a matrix $\mathbf S \in \mathcal S$ is identified with a spline $S$.

The importance of the derived relations for numerical implementation of the spline objects is two-fold. 
Firstly, they can be used to define splines through specifying their derivatives values.
Secondly, in applications that are computationally intensive, where a lot of algebraic operations are performed on the entries of the matrix $\mathbf S$ often, due to round-off errors, the matrix entries may cease to satisfy (\ref{eq:addk-1nL}). In such a case, it maybe required to verify if the equations hold up to a required accuracy. 
If the accuracy is not achieved some corrections of computational results need to be addressed and the derived equations can be also utilized for the purpose. 

The equations can be written in the matrix notation using a $ (k+1)\times (k+1)$ lower triangular Toeplitz matrix  that for $\alpha>0$ is defined    through 
\begin{equation}
\label{eq:A}
\mathbf A_{\alpha}=\begin{bmatrix} 
1 & 0 &0 & 0&  \dots &0& 0 &0&0\,\,\,\,\\
\alpha  & 1 &0 & 0& \cdots&0 & 0 & 0&0\,\,\,\,\\
\frac{\alpha^2}2 & \alpha & 1 & 0& \dots & 0&  0 & 0&0 \,\,\,\,\\
\frac{\alpha^3}{3!} & \frac{\alpha^2}2 &\alpha & 1 &  \dots &0&0 & 0&0 \,\,\,\,\\
\vdots & \vdots  & \vdots & \vdots & \ddots &  \vdots & \vdots & \vdots\,\,\,\,\\
 \frac{\alpha^{k-3}}{(k-3)!}& \frac{\alpha^{k-4}}{(k-4)!} & \frac{\alpha^{k-5}}{(k-5)!}& \frac{\alpha^{k-6}}{(k-6)!} & \cdots  &1&0 & 0  &0\,\,\,\,\\
 \frac{\alpha^{k-2}}{(k-2)!}& \frac{\alpha^{k-3}}{(k-3)!} & \frac{\alpha^{k-4}}{(k-4)!}&  \frac{\alpha^{k-5}}{(k-5)!}  & \cdots&\alpha & 1 & 0 &0\,\,\,\,\\
 \frac{\alpha^{k-1}}{(k-1)!} & \frac{\alpha^{k-2}}{(k-2)!} &  \frac{\alpha^{k-3}}{(k-3)!} & \frac{\alpha^{k-4}}{(k-4)!} & \cdots   &\frac{\alpha^2}2 & \alpha &1 &0\,\,\,\,\\
\frac{\alpha^k}{k!} & \frac{\alpha^{k-1}}{(k-1)!} & \frac{\alpha^{k-2}}{(k-2)!} &  \frac{\alpha^{k-3}}{(k-3)!} & \cdots & \frac{\alpha^3}{3!}&\frac{\alpha^2}2 &\alpha &  1\,\,\,\,\\
\end{bmatrix}.
\end{equation}
 The matrix $\mathbf S$ can be split to $\mathbf S^L$ and  $\mathbf S^R$, where $\mathbf S^L=[\mathbf s_0 \dots \mathbf s_{l+1}]$ is made of the first $l+2$ rows of $\mathbf S$ and thus corresponding to $\boldsymbol \xi^L$  while  $\mathbf S^R=[s_{ij}^R]_{i=0,j=0}^{l+1,k}$ corresponds to $\boldsymbol \xi^{R}$, as defined before.
 Then (\ref{eq:addk-1nL}) can be written as 
 \begin{align}
 \label{eq:SLR}
 \mathbf S^L=\begin{bmatrix}
 \mathbf s_{0\cdot}\\
 \mathbf s_{0\cdot} \mathbf A_{\xi_1 - \xi_{0}}\\
  \mathbf s_{1\cdot}\mathbf A_{\xi_2 - \xi_{1}}\\
  \vdots \\
   \mathbf s_{i-1\, \cdot}\mathbf A_{\xi_i - \xi_{i-1}}\\
    \vdots \\
    \mathbf s_{l \cdot}\mathbf A_{\xi_{l+1}- \xi_{l}} 
 \end{bmatrix}
 +
  \left[\mathbf 0\, \mathbf 0 \, \dots \, \mathbf 0 \, \boldsymbol \Delta \mathbf s_k^L \right],&
 \hspace{3mm}
  \mathbf S^R=\begin{bmatrix}
 \mathbf S^R_{0\cdot}\\
 \mathbf S^R_{0\cdot} \mathbf A_{\xi^R_1 - \xi^R_{0}}\\
  \mathbf S^R_{1\cdot}\mathbf A_{\xi^R_2 - \xi^R_{1}}\\
  \vdots \\
   \mathbf S^R_{i-1\, \cdot}\mathbf A_{\xi^R_i - \xi^R_{i-1}}\\
    \vdots \\
    \mathbf S^R_{l \cdot}\mathbf A_{\xi^R_{l+1}- \xi^R_{l}} 
 \end{bmatrix}
 +
 \left[\mathbf 0\, \mathbf 0 \, \dots \, \mathbf 0 \,\, \boldsymbol \Delta \mathbf s_k^R \right],
 \end{align}
 where $(l+2)\times(l+2)$ matrix $\boldsymbol \Delta$ is given in
\begin{align*}
\boldsymbol \Delta&= 
\begin{bmatrix} 
0 & 0 & 0 &\dots &0\\
-1 &       1 &           0 &  \dots    & 0 \\ 
0  &     -1  &          1  &\ddots    &\vdots  \\
\vdots & \vdots  & \ddots & \ddots & 0 \\ 
0  &    0 &    \dots  &    -1        & 1 
\end{bmatrix}.
\end{align*}
For the equally spaced $\boldsymbol \xi$, we have even simpler relations
 \begin{align}
 \label{eq:SLRS}
 \mathbf S^L=\begin{bmatrix}
 \mathbf s_{0\cdot}\\
 \mathbf S_{0..l\,\cdot\,} \mathbf A_{\alpha}
 \end{bmatrix}
 +
 \left[\mathbf 0\, \mathbf 0 \, \dots \, \mathbf 0 \, \boldsymbol \Delta \mathbf s_k^L \right],&
 \hspace{3mm}
  \mathbf S^R=\begin{bmatrix}
 \mathbf S^R_{0\cdot}\\
 \mathbf S^R_{0..l\,\cdot} \mathbf A_{\alpha} 
 \end{bmatrix}
 +
 \left[\mathbf 0\, \mathbf 0 \, \dots \, \mathbf 0 \, \boldsymbol \Delta \mathbf s_k^R \right],
 \end{align}
for $\alpha=(\xi_{n+1}-\xi_0)/(n+1)$.

To understand better the introduced notation and the symmetric and one-sided matrix representations we present a simple example. 
\begin{Example}
\label{matrices}
We consider the cubic splines, i.e. $k=3$ and two examples of the matrices for equally spaced knots with $n=10$ and $n=11$. 
The matrices of the derivatives at the knots that are adopted in the {\tt Splinets}-object have the following forms
{\scriptsize
$$
\mathbf S_{10}=
 \begin{bmatrix}
     \mathbf {\color{red} 0.00} &   \mathbf {\color{red}  0.00}&   \mathbf {\color{red}  0.00} & -24.50 \\ 
  -0.00 & -0.10 & -2.23 & -75.91 \\ 
  -0.03 & -0.62 & -9.13 & 99.79 \\ 
  -0.11 & -1.03 & -0.06 & 2.52 \\ 
  -0.21 & -1.03 & 0.17 & 1.16 \\ 
  -0.30 & -1.01 & 0.28 &\mathbf{ \color{red}  0.30} \\ 
  -0.39 & -0.98 & 0.31 &\mathbf{\color{red}  0.30} \\ 
  -0.48 & -0.96 & 0.15 & -1.75 \\ 
  -0.56 & -0.95 & 0.20 & 0.53 \\ 
  -0.53 & 3.11 & 88.98 & 976.66 \\ 
  -0.11 & 3.58 & -78.67 & -1844.17 \\ 
     \mathbf {\color{red} 0.00} &   \mathbf {\color{red} 0.00}&   \mathbf {\color{red} 0.00}  & 865.37 \\ 
  \end{bmatrix},
  \,\,\,
  \mathbf S_{11}=
 \begin{bmatrix}
    \mathbf {\color{red} 0.00} &   \mathbf {\color{red} 0.00}&   \mathbf {\color{red} 0.00} & 833.39 \\ 
  0.08 & 2.89 & 69.45 & -1799.55 \\ 
  0.39 & 2.43 & -80.51 & 966.97 \\ 
  0.41 & -0.92 & 0.07 & 0.62 \\ 
  0.33 & -0.91 & 0.12 & -5.29 \\ 
  0.25 & -0.92 & -0.32 & -0.38 \\ 
  0.18 & -0.95 & -0.35  &\mathbf{\color{red}   0.00} \\ 
  0.10 & -0.98 & -0.43 & -0.98 \\ 
  0.01 & -0.99 & 0.29 & 8.73 \\ 
  -0.07 & -0.96 & 0.40 & 1.33 \\ 
  -0.11 & 0.41 & 32.49 & 385.06 \\ 
  -0.02 & 0.88 & -21.20 & -644.33 \\ 
       \mathbf {\color{red} 0.00} &   \mathbf {\color{red} 0.00}&   \mathbf {\color{red} 0.00}  & 254.43 \\ 
  \end{bmatrix}.
$$
}
In the first and the last rows, we observe the zero boundary conditions.  In the even number of knots case, the two highest derivatives at the center are equal to each other, while for the odd number of knots, the highest derivative column has zero at the center.    
The transformation by means of {\tt sym2one()} leads to the one-sided representations 
{\scriptsize$$
\mathbf S'_{10}=
 \begin{bmatrix}
     \mathbf {\color{red} 0.00} &   \mathbf {\color{red}  0.00}&   \mathbf {\color{red}  0.00} & -24.50 \\ 
  -0.00 & -0.10 & -2.23 & -75.91 \\ 
  -0.03 & -0.62 & -9.13 & 99.79 \\ 
  -0.11 & -1.03 & -0.06 & 2.52 \\ 
  -0.21 & -1.03 & 0.17 & 1.16 \\ 
  -0.30 & -1.01 & 0.28 & 0.30 \\ 
  -0.39 & -0.98 & 0.31 & -1.75 \\ 
  -0.48 & -0.96 & 0.15 & 0.53 \\ 
  -0.56 & -0.95 & 0.20 & 976.66 \\ 
  -0.53 & 3.11 & 88.98 & -1844.17 \\ 
  -0.11 & 3.58 & -78.67 & 865.37 \\ 
     \mathbf {\color{red} 0.00} &      \mathbf {\color{red} 0.00} &   \mathbf {\color{red}  0.00}&   \mathbf {\color{red}  0.00} \\ 
  \end{bmatrix},
  \,\,\,
  \mathbf S'_{11}=
 \begin{bmatrix}
    \mathbf {\color{red} 0.00} &   \mathbf {\color{red} 0.00}&   \mathbf {\color{red} 0.00} &833.39 \\ 
  0.08 & 2.89 & 69.45 & -1799.55 \\ 
  0.39 & 2.43 & -80.51 & 966.97 \\ 
  0.41 & -0.92 & 0.07 & 0.62 \\ 
  0.33 & -0.91 & 0.12 & -5.29 \\ 
  0.25 & -0.92 & -0.32 & -0.38 \\ 
  0.18 & -0.95 & -0.35 & -0.98 \\ 
  0.10 & -0.98 & -0.43 & 8.73 \\ 
  0.01 & -0.99 & 0.29 & 1.33 \\ 
  -0.07 & -0.96 & 0.40 & 385.06 \\ 
  -0.11 & 0.41 & 32.49 & -644.33 \\ 
  -0.02 & 0.88 & -21.20 & 254.43 \\  
        \mathbf {\color{red} 0.00} &      \mathbf {\color{red} 0.00} &   \mathbf {\color{red}  0.00}&   \mathbf {\color{red}  0.00} \\
  \end{bmatrix}.
$$
}
The left- and right-hand-side halves are in this example as follows
{\scriptsize
\begin{align*}
\mathbf S^L_{10}=
 \begin{bmatrix}
     \mathbf {\color{red} 0.00} &   \mathbf {\color{red}  0.00}&   \mathbf {\color{red}  0.00} & -24.50 \\ 
  -0.00 & -0.10 & -2.23 & -75.91 \\ 
  -0.03 & -0.62 & -9.13 & 99.79 \\ 
  -0.11 & -1.03 & -0.06 & 2.52 \\ 
  -0.21 & -1.03 & 0.17 & 1.16 \\ 
  -0.30 & -1.01 & 0.28 &\mathbf{ \color{red}  0.30} \\ 
  -0.39 & -0.98 & 0.31 &\mathbf{\color{red}  0.30} \\ 
  \end{bmatrix},
  \,\,\,&
  \mathbf S^L_{11}=
 \begin{bmatrix}
     \mathbf {\color{red} 0.00} &   \mathbf {\color{red} 0.00}&   \mathbf {\color{red} 0.00} & 833.39 \\ 
  0.08 & 2.89 & 69.45 & -1799.55 \\ 
  0.39 & 2.43 & -80.51 & 966.97 \\ 
  0.41 & -0.92 & 0.07 & 0.62 \\ 
  0.33 & -0.91 & 0.12 & -5.29 \\ 
  0.25 & -0.92 & -0.32 & -0.38 \\ 
  0.18 & -0.95 & -0.35  &\mathbf{\color{red}   0.00} \\ 
  \end{bmatrix},
 \\
 \mathbf S^R_{11}=
 \begin{bmatrix}
     \mathbf {\color{red} 0.00} &   \mathbf {\color{red}  0.00}&   \mathbf {\color{red}  0.00} & 865.37 \\ 
  -0.11 & 3.58 & -78.67 & -1844.17 \\ 
  -0.53 & 3.11 & 88.98 & 976.66 \\ 
  -0.56 & -0.95 & 0.20 & 0.53 \\ 
  -0.48 & -0.96 & 0.15 & -1.75 \\ 
  -0.39 & -0.98 & 0.31 &\mathbf{\color{red}  0.30} \\ 
  -0.30 & -1.01 & 0.28  &\mathbf{\color{red}  0.30} \\ 
  \end{bmatrix},
  \,\,\,&
 \mathbf S^R_{11}=
 \begin{bmatrix}
     \mathbf {\color{red} 0.00} &   \mathbf {\color{red} 0.00}&   \mathbf {\color{red} 0.00} & 254.43 \\ 
  -0.02 & 0.88 & -21.20 & -644.33 \\ 
  -0.11 & 0.41 & 32.49 & 385.06 \\ 
  -0.07 & -0.96 & 0.40 & 1.33 \\ 
  0.01 & -0.99 & 0.29 & 8.73 \\ 
  0.10 & -0.98 & -0.43 & -0.98 \\ 
  0.18 & -0.95 & -0.35   &\mathbf{\color{red}   0.00} \\ 
  \end{bmatrix}.
\end{align*}
}
\end{Example}

\subsection{Implementation of class {\it 'Splinets'} in the object oriented programming}
The efficiency gains by using our methods can be fully utilized only if a proper approach to the spline calculus is taken. 
In particular, one must utilize small supports of the considered spline bases. 
If two splines are having in common only a small portion of their supports, then the inner product between them needs to be evaluated only over this common support. 
Thus for splines  with small support  evaluations become computationally less demanding even if the number of knots and thus of basis splines is large. 
For this reason, we have implemented a spline object that contains information about the spline support and we use this information in computational procedures. 

Outside of the support the values of the derivatives are zero and to utilize this in efficient computations, any spline function $S$ of order $k$ over knots in $\boldsymbol \xi=(\xi_0,\dots, \xi_{n+1})$ is identified as 
$$
S=\left\{k, \boldsymbol \xi,\mathcal I, \mathbf s_0,\mathbf s_1, \dots, \mathbf s_k\right\},
$$
where $\mathcal I=\{(i_1, i_1+m_1+1), \dots,(i_ N,i_N+m_N+1) \}$ is a sequence of ordered pairs of indexes in $\{1,\dots,n+2\}$ representing the intervals, the union of which is the support of a spline, i.e. the minimial closed set outside  of which spline vanishes. It implies that the knots outside the support would have the corresponding entries in $\mathbf S$ equal to zero and thus there is no need to include them in the matrix of derivatives. 
The matrix $\mathbf S$ is thus divided into  $N$ blocks $\mathbf S_r$ row-wise of the sizes $(m_r+2)\times k$, $r=1,\dots, N$.
Thus the columns in the $r$-block $\mathbf s_0^{(r)},\mathbf s_1^{(r)}, \dots, \mathbf s_k^{(r)}$ are 
$
m_r+2
$ dimensional column vectors of values of the derivatives of $S$ at the knots of support, $\xi_{i_r+l}$, $l=0,\dots, m_r+1$.
It is important to point that the block  $\mathbf S_r$ is kept in the symmetric around center format as described in the previous section, except there $n$ has to be replaced by $m_r$ and $(\xi_0,\dots, \xi_{n+1})$ by $(\xi_{i_r},\dots,\xi_{i_r+m_r+1})$. 

For the so defined object the most elementary function is verification if the object satisfy condition of a spline with zero-boundary conditions for which (\ref{eq:SLR}) hold. 
It is implemented in {\tt is.spline()} method for this class. 
Two functions {\tt gather()} and {\tt subsample()} are grouping two {\tt Splinets} objects into one and subsampling from an {\tt Splinets} object to obtain another object, respectively. 
Another basic function is  {\tt evaluate()}, which evaluates splines in a {\tt Splinets}-object at a given vector of argument values within the range of knots through the Taylor expansion. 
Finally, the generic function {\tt plot()} utilizes spline evaluation to plot graphs through $k\cdot N$ points in-between knots, where $k$ is the order of a spline and $N$ is the parameter passed by the method. 
  
\subsection{Support sets}
One of the most important of our implementation of the spline algebra is controlling the locality of a given spline by keeping information about its support sets. 
We envision that this feature will be particularly useful for extension of the package to higher dimensions that we plan to do in the future, but even in the current version it plays its role.  
The main reason why it is important is that the orthonormalization we implemented to obtain the splinets is optimal in respect to the total support size. 
Thus our goal is to utilize this feature to improve efficiency in computations -- the operations  outside of the support set simply are not needed and thus not carried out.  
 In the background of the all analytical and algebraical operations on {\tt Splinet}-objects, the package carries out the evaluation of the support sets of the outputs. 
 This evaluation is based on the set algebra operations on the support sets which are kept in the field {\tt supp} of the object. 
 
 All operations for maintaining the correct information about the support sets are carried in the background but it may happen that the actual support set is not matching the evaluated one. 
 A typical example, when such a problem may occur is addition of two splines. 
 Generally,  the sum of two splines has the support sitting on the union of the individual supports. 
 However, if we add $S$ to $-S$ the support is an empty set and thus the general rule does not apply. 
 For this reason, the package has a function {\tt exsupp} that extracts the support from a {\tt Spline}-object.
 This function should be used whenever the accurate support sets are needed and there is a suspicion of deviation in the object from the actual support set representation. 
\section{Building functional splines}
The class of {\tt spline} allows for quite arbitrarily general {\tt fields} and the method {\tt is.spline()} checks if the defined object is indeed a spline. 
Since, it is not trivial to correctly define a proper matrix of the derivative at the knots, the function that `corrects' a given matrix so it follows (\ref{eq:SLR}) is of interest and thus implemented in the package. 
Using such a function provides the simplest way of defining a proper spline object. 
It can be also used to correct the values of the derivatives at the knots due to a roundup error for a large number of spline evaluations. 
Another way to build a proper spline, is to randomly select a spline object. In the package, we actually have a flexible method of generating random spline objects. 
Finally, one can obtain a set of splines by utilizing a basis of splines, this method is presented later on in the paper where a number of different spline base are discussed. 
In what follows, the first two methods of obtaining splines are discussed in further detail.

\subsection{Generating an individual non-random spline}
For given knots and an order of smoothness,  $n-k+1$ is the dimension of the spline space.
The problem of defining a spline reduces to  obtaining matrices $\mathbf S$ such that  (\ref{eq:SLR}) is satisfied, by, for example, setting $n-k+1$ entries and the rest  evaluated using the relation (\ref{eq:addk-1nL})  or, equivalently, (\ref{eq:SLR}). 
Although, the relations are linear, the objects in question are matrices and solving the linear equation for matrices is not as straightforward as solving linear equations for vectors. 
Here we discuss three special important cases when the evaluation can be efficiently and transparently handled. 

\begin{description}
\item[The first row and the last column case (frlc):]
The first case is obtained by fixing values of the highest order derivative, i.e. the last column in $\mathbf S$. 
Consider a $(m+2) \times (k+1)$ submatrix   $\mathbf U=[u_{ij}]_{i,j=0}^{m+1,k}$ of $\mathbf S$ made of $m+2$ subsequent rows in the top half of $\mathbf S$.
Assume that the first row $\mathbf u_{0\,\cdot }$ and the last column $\mathbf u_{\cdot\, k}$  of $\mathbf U$ are known. 
Thus, in total, $m+k+1$ values of are known. 
In this case, all the remaining values are the coordinates of vectors $\mathbf u_{i\,0..k-1}$, $i=1,\dots, m+1$, which  can be recursively computed using, for $i=1,\dots, m+1$, the relations
\begin{align}
\mathbf u_{i\,0.. k-1}&=\mathbf u_{i-1\,\cdot}~\left[\mathbf A_{\xi_i-\xi_{i-1}}\right ]_{ \cdot\, 0.. k-1}=\left[\mathbf u_{i-1\cdot}\mathbf A_{\xi_i - \xi_{i-1}}\right]_{0.. k-1},
\end{align}
We note that in these relations $u_{m,k}$ is not present, so its value does not affect the other computed values in the matrix $\mathbf U$.
\item[The first row and the first column case (frfc):]
Another important special case is evaluation of the entries of $\mathbf U$ is when, instead of the last column entries (the values of the $k$-th derivative), the first column ones (the values of the spline) are known, i.e. $\mathbf u_{\cdot 0}$ is available together with   the first row $\mathbf u_{0\,\, 0..k-1 }$ (except the $k$th derivative) and we want to evaluate the rest of the entries of $\mathbf U$ except $u_{m,k}$, which is not involved in the Taylor expansion formula based on the considered knots. 
We can assume that we also know the value $u_{0 k}$ since 
$$
u_{10}=\mathbf u_{0\, \cdot}\left[\mathbf A_{\xi_{1}- \xi_{0}}\right]_{0..k\, 0}
$$
from which we can evaluate $u_{0 k}$ through the explicit formula 
$$
u_{0k}=\frac{1}{\left[\mathbf A_{\xi_{1} - \xi_{0}}\right]_{k~0}}
\left(u_{10}-\mathbf u_{0\,0..k-1}\left[\mathbf A_{\xi_{1}- \xi_{0}}\right]_{0..k-1\, 0}\right).
$$
However, by having an entry of the last column one can follow the evaluation of the next row as in the previous case. 
In other words,  to evaluate  the coordinates of vectors $\mathbf u_{i\,1..k}$, $i=1,\dots, {m}$ and $\mathbf u_{m+1,1..k-1}$ we follow recurrent relations
\begin{equation}
\begin{split}
\mathbf u_{i\,1..k-1}&=\mathbf u_{i-1~\cdot}\left[\mathbf A_{\xi_i - \xi_{i-1}}\right]_{\cdot\, 1..k-1}\, ,\, i=1,\dots,m+1 ,\\
u_{ik} &=\frac{1}{\left[\mathbf A_{\xi_{i+1} - \xi_{i}}\right]_{k~0}}\left(u_{i+1\,0}-\mathbf u_{i\,0..k-1}\left[\mathbf A_{\xi_{i+1}- \xi_{i}}\right]_{0..k-1\, 0}\right)
\, ,\, i=1,\dots,m.
\end{split}
\end{equation}
\item[The first row and the last row case (frlr):] 
The final case we consider is when the first and the last rows in $\mathbf U=[u_{ij}]_{i,j=0}^{m+1,k}$ are set. 
We note that $u_{m+1\,k}$ correspond to the $k$th derivative on the interval $[\xi_{m+1},\xi_{m+2})$ and thus on can set it to an arbitrary number without having any effect on the behavior of the spline on $[\xi_0,\xi_{m+1})$. 
Thus disregarding this irrelevant entry, the dimension of the space of the matrices $\mathbf U$ that corresponds to the space of splines is $k+m+1$. 
It implies that by specifying the first and the last row (with $u_{m+1\,k}$ not contributing to the dimension of the splines) we have an overspecified matrix if $m<k$, the uniquely specified matrix if $m=k$ and an underspecified matrix if $m>k$.
To assure that other entries are defined, we consider only  $m=0,\dots, k$ and note that if $m<k$, then $k-m$ entries in the last (or the first) row do not need to be used. 
As we will see, it is convenient to assume that $\mathbf u_{m+1\, 0..k-m-1}$ is not used and these entries can be evaluated based on the other values.

Consider first $m=0$, in which the case the entire $2 \times k+1$ matrix $\mathbf U$ is defined but  the entries must satisfy additional $k$ conditions
$$
\mathbf u_{1\,0..k-1}=\mathbf u_{0\cdot}\left[\mathbf A_{\xi_1 - \xi_{0}}\right]_{\cdot\,0..k-1},
$$
or, otherwise, it does not correspond to a valid spline. One can correct this by changing the values on the left-hand-side (LHS) of the above to the one obtained from the right-hand-side (RHS). 
This corrects the underlying spline object to match the matrix restrictions.

For $m=1$, the middle row $u_{1\cdot}$ needs to be evaluated, which is done in two steps. 
First, as before, we evaluate 
\begin{equation}
\label{EqOne}
\mathbf u_{1\,0..k-1}=\mathbf u_{0\cdot}\left[\mathbf A_{\xi_1 - \xi_{0}}\right]_{\cdot\,0..k-1},
\end{equation}
then the final value $u_{1 k}$ is obtained by solving
\begin{align*}
\label{EqTwo}
u_{2\,\, k-1}&=
\mathbf u_{1\cdot}\left[\mathbf A_{\xi_2 - \xi_{1}}\right]_{\cdot\, \,k-1}=
\begin{bmatrix} u_{1\, k-1}
u_{1\, k}
\end{bmatrix} 
\left[\mathbf A_{\xi_2 - \xi_{1}}\right]_{k-1..k\,\, k-1}\\
&=
u_{1\, k-1}
\left[\mathbf A_{\xi_2 - \xi_{1}}\right]_{k-1\,\, k-1}
+
u_{1\, k}
\left[\mathbf A_{\xi_2 - \xi_{1}}\right]_{k\,\, k-1}\\
&=
\mathbf u_{0\cdot}\left[\mathbf A_{\xi_1 - \xi_{0}}\right]_{\cdot\,k-1}
\left[\mathbf A_{\xi_2 - \xi_{1}}\right]_{k-1\,\, k-1}
+
u_{1\, k}
\left[\mathbf A_{\xi_2 - \xi_{1}}\right]_{k\,\, k-1}\\
&=
\mathbf u_{0\cdot}\left[\mathbf A_{\xi_1 - \xi_{0}}\right]_{\cdot\,k-1}
+
u_{1\, k}
\left[\mathbf A_{\xi_2 - \xi_{1}}\right]_{k\,\, k-1},
\end{align*}
where the last equation follows from the fact that $\mathbf A _{\xi_2 - \xi_{1}}$ has ones on the diagonal. 
The explicit solution is 
\begin{align*}
u_{1 k}=\left(u_{2\,\, k-1}-\mathbf u_{0\cdot}\left[\mathbf A_{\xi_1 - \xi_{0}}\right]_{\cdot\,k-1}
\right)/\left[\mathbf A_{\xi_2 - \xi_{1}}\right]_{k\,\, k-1}
.
\end{align*}
Additional $k-1$  conditions need to be satisfied in order for the entries of $\mathbf U$ to represent a valid spline
$$
\mathbf u_{2\, 0..k-2}=\mathbf u_{1\cdot}\left[\mathbf A_{\xi_2 - \xi_{1}}\right]_{\cdot\, \,0\dots k-2}.
$$
Again, if these equations are not satisfied, one can correct them by changing the LHS values to the RHS ones.  
This is permissible since the values on the LHS have not been used for the evaluation of other entries in the matrix. 

For a general $m\le k-1$, we have to evaluate the $m$ rows $\mathbf u_{1..m\,\, \cdot}$ given that the first and the last, i.e. $\mathbf u_{0 \cdot}$ and $\mathbf u_{m+1 \cdot }$, are given. 
It is sufficient to provide equations needed to evaluate $\mathbf u_{1..m\, k}$ and  use the first-row-and-the-last-column case described above to evaluate all other entries. 
We use (\ref{EqOne}) to obtain $\mathbf u_{1\,0..k-1}$, then we consider
\begin{align*}
\mathbf u_{2\, \,k-m.. k-1 }&= \mathbf u_{1\cdot}\left[\mathbf A_{\xi_2 - \xi_{1}}\right]_{\cdot\, \,k-m.. k-1}
=\mathbf u_{1\, \,k-m.. k}\left[\mathbf A_{\xi_2 - \xi_{1}}\right]_{k-m.. k\, \,\,k-m.. k-1}\\
&=\mathbf u_{1\, \,k-m.. k-1}\left[\mathbf A_{\xi_2 - \xi_{1}}\right]_{k-m.. k-1\, \,\,k-m.. k-1} 
+
u_{1 k}\left[\mathbf A_{\xi_2 - \xi_{1}}\right]_{k\, \,\,k-m.. k-1}
\\
&=
\mathbf u_{0\, \,k-m.. k-1}
\left[\mathbf A_{\xi_1 - \xi_{0}}\right]_{k-m.. k-1\, \,\,k-m..k-1}
\left[\mathbf A_{\xi_2 - \xi_{1}}\right]_{k-m.. k-1\, \,\,k-m.. k-1} \\
&+
u_{0k}
\left[\mathbf A_{\xi_1 - \xi_0}\right]_{k \, \,\,k-m.. k-1} 
\left[\mathbf A_{\xi_2 - \xi_{1}}\right]_{k-m.. k-1\, \,\,k-m.. k-1}
+
u_{1 k}\left[\mathbf A_{\xi_2 - \xi_{1}}\right]_{k\, \,\,k-m.. k-1}.
\end{align*}
Using  matrices $\mathbf A^{(j)}$ of sizes $m\times m$, and vectors $\mathbf c^{(j)}$ of  sizes $1 \times m$, defined by
\begin{align*}
\mathbf A^{(j)}&=\left[\mathbf A_{\xi_j - \xi_{j-1}}\right]_{k-m.. k-1 \, \,\,k-m.. k-1},\\
\mathbf c^{(j)}&=\left[\mathbf A_{\xi_j - \xi_{j-1}}\right]_{k \, \,\,k-m.. k-1},
\end{align*}
 we write in a compact manner
\begin{align*}
\mathbf u_{1\, \,k-m.. k-1 }&=
\mathbf u_{0\, \,k-m.. k-1}\mathbf A^{(1)}
+
u_{0 k}\mathbf c^{(1)},\\
\mathbf u_{2\, \,k-m.. k-1 }&=
\mathbf u_{0\, \,k-m.. k-1}\mathbf A^{(1)}\mathbf A^{(2)}
+
u_{0 k}\mathbf c^{(1)}\mathbf A^{(2)}
+
u_{1 k}\mathbf c^{(2)}.
\end{align*}

Similarly, for $j=2,\dots, m+1$:
\begin{align*}
\mathbf u_{j\, \,k-m.. k-1 }&= \mathbf u_{j-1\cdot}\left[\mathbf A_{\xi_j - \xi_{j-1}}\right]_{\cdot\, \,k-m.. k-1}
=\mathbf u_{j-1\, \,k-m.. k}\left[\mathbf A_{\xi_j - \xi_{j-1}}\right]_{k-m.. k\, \,\,k-m.. k-1}\\
&=\mathbf u_{j-1\, \,k-m.. k-1}\mathbf A^{(j)} 
+
u_{j-1 k}\mathbf c^{(j)}
.
\end{align*}
Using $m \times m$ matrices $\mathbf B^{(r\,j)}=\prod_{l=r}^j\mathbf A^{(l)}$, $r=1,\dots,j$, $\mathbf B^{(j+1\,j)}=\mathbf I_m$, we obtain non-recurrent relations
\begin{equation}
\label{nonrec}
\mathbf u_{j\, \,k-m.. k-1 }
=
\mathbf u_{0\, \,k-m.. k-1}\mathbf B^{(1j)}
+\sum_{r=1}^{j}
u_{r-1k}\mathbf c^{(r)}\mathbf B^{(r+1\,j)}.
\end{equation}
For $j=m+1$ they can be written as
\begin{align*}
\mathbf u_{m+1\, \,k-m.. k-1 }&=\mathbf u_{0\, \,k-m.. k-1}\mathbf B^{(1\,m+1)}
+\sum_{r=1}^{m+1}
u_{r-1k}\mathbf c^{(r)}\mathbf B^{(r+1\,m+1)}
\\
&=\mathbf u_{0\, \,k-m.. k-1}\mathbf B^{(1\, m+1)}
+u_{0 k} \mathbf c^{(1)}\mathbf B^{(2\, m+1)}
+\mathbf 
u_{1..m\,\, k}^{\top}\mathbf C^{(m)}\\
&=\left(
\mathbf u_{0\, \,k-m.. k-1}\mathbf A^{(1)}
+u_{0 k} \mathbf c^{(1)}
\right)
\mathbf B^{(2\, m+1)}
+
\mathbf 
u_{1..m\,\, k}^{\top}\mathbf C^{(m)}\\
&=
\mathbf u_{0\, \,k-m.. k}\mathbf D^{(m)}
+
\mathbf 
u_{1..m\,\, k}^{\top}\mathbf C^{(m)},
\end{align*}
where  
$$
\mathbf D^{(m)}=\left[\mathbf A_{\xi_1 - \xi_{0}}\right]_{k-m.. k\, \,\,k-m.. k-1}\mathbf B^{(2\, m+1)}
$$ 
 is an $(m+1)\times m$ matrix, while $m \times m$  matrix $\mathbf C^{(m)}$ is the row-wise concatenation of $m$ row $1\times m$ matrices $\mathbf c^{(r)}\mathbf B^{(r+1\,m+1)}$, $r=2,\dots,m+1$.
In these equations only $\mathbf u_{1..m\, k}$ is unknown thus the problem reduces to finding the inverse $\mathbf E^{(m)}$ of $\mathbf C^{(m)}$ and evaluating
\begin{equation}
\label{final}
\mathbf 
u_{1..m\,\, k}^{\top}=
\left( 
\mathbf u_{m+1\, \,k-m.. k-1} -
\mathbf u_{0\, \,k-m.. k}\mathbf D^{(m)}
\right) \mathbf E^{(m)}.
\end{equation}
Whenever $m\le k$, the vector $\mathbf u_{1..m\,\, k}$ is computable from the above. 
When it is combined with the first row $\mathbf u_{0\, \cdot}$ (or the last row $\mathbf u_{m+1\, \cdot}$) all the remaining values can be computed as described above in the first-row-last-column case. 
Since $\mathbf u_{m+1\, 0..k-m-1}$ is not used in the above equations their value have to agree with the values that has been evaluated. 
If they are not they are corrected by setting them to the computed values.

The special case of equally spaced knots simplifies significantly the computations. 
Indeed, in this case we can set $\mathbf A$ to the common value of $\mathbf A^{(j)}$, $j=1,\dots,m+1$, and $\mathbf c$ to $\mathbf c^{(j)}$'s, i.e. 
$$
\mathbf A =\left[\mathbf A_{\xi_1 - \xi_0}\right]_{k-m.. k-1 \, \,\,k-m.. k-1} ,~~ \mathbf c = \left[\mathbf A_{\xi_1 - \xi_0}\right]_{k \, \,\,k-m.. k-1}.
$$  
Then (\ref{nonrec}) becomes
$$
\mathbf u_{j\, \,k-m.. k-1 }
=
\mathbf u_{0\, \,k-m.. k-1}\mathbf A^{j}
+\sum_{r=1}^{j}
u_{r-1k}\mathbf c\mathbf A^{j-r},
$$
and the matrices in (\ref{final}) become
\begin{align*}
\mathbf D^{(m)}&=
\left[\mathbf A_{\xi_1 - \xi_{0}}\right]_{k-m.. k\, \,\,k-m.. k-1}\mathbf A^{ m},\\
 \mathbf E^{(m)}&=
\begin{bmatrix}
\mathbf c {\mathbf A}^{m-1}\\
\mathbf c {\mathbf A}^{m-2}\\
\vdots\\
\mathbf c {\mathbf A}\\
\mathbf c\\
\end{bmatrix}^{-1}.
\end{align*}

The special case of $m=k$ is particularly important, since it leads to the unique matrix specification. 
We summarize our findings for this case in the result below. 
 \end{description}
 
 \begin{proposition}
 Let $\mathbf U=[u_{ij}]_{i,j=0}^{k+1,k}$ be a $(k+2)\times (k+1)$ matrix that corresponds to the matrix of derivatives at knots $\xi_0,\dots, \xi_{k+1}$ of a $k$th order of spline defined over these knots, where the right-hand-side derivative is considered as  the $k$th order derivative.
 
  If its first row $\mathbf u_{0 \cdot}$ and the last row $\mathbf u_{k+1 \cdot}$ are both given, then all the remaining entries are uniquely defined. In particular, the last column $\mathbf u_{\cdot k+1}$ is given through 
\begin{equation}
\label{final2}
\mathbf 
u_{1..k\,\, k}^{\top}=
\left(
\mathbf u_{k+1\, \,0.. k-1 }-\mathbf u_{0\, \,0.. k}\mathbf D
\right)
\mathbf E,
\end{equation}
where 
\begin{align*}
\mathbf E&=\begin{bmatrix}
\mathbf c^{(2)} \prod_{l=3}^{k+1}{\mathbf A}^{(l)}\\
\mathbf c^{(3)} \prod_{l=4}^{k+1}{\mathbf A}^{(l)}\\
\vdots\\
\mathbf c^{(k)} {\mathbf A}^{(k+1)}\\
\mathbf c^{(k+1)}
\end{bmatrix}^{-1},\\
\mathbf D&=\left[\mathbf A_{\xi_1 - \xi_{0}}\right]_{0.. k\, \,\,0.. k-1}\prod_{l=2}^{k+1}{\mathbf A}^{(l)},
\end{align*}
are $ k \times k $ and $ (k+1) \times k $ matrices, respectively, in which, for $j=1,\dots, k+1$, 
\begin{align*}
\mathbf A^{(j)}&=\left[\mathbf A_{\xi_j - \xi_{j-1}}\right]_{0.. k-1 \, \,\,0.. k-1},\\
\mathbf c^{(j)}&=\left[\mathbf A_{\xi_j - \xi_{j-1}}\right]_{k \, \,\,0.. k-1}.
\end{align*}
All the remaining entries are recursively defined, for $i=1,\dots, k$ through
\begin{align*}
\mathbf u_{i\,0.. k-1}&=\mathbf u_{i-1\,\cdot}~\left[\mathbf A_{\xi_i-\xi_{i-1}}\right ]_{ \cdot\, 0.. k-1}=\left[\mathbf u_{i-1\cdot}\mathbf A_{\xi_i - \xi_{i-1}}\right]_{0.. k-1}.
\end{align*}
 \end{proposition}

To illustrate different features of the matrix derivative adjustments, we put all three methods to the task of `correcting' a randomly chosen piecewise cubic polynomial function shown in Figure~\ref{SpOb}~{(\it Top)}.
We observe that the first two methods, while providing quite accurate approximation close to zero, destabilize quickly -- the method based on matching the highest order derivative biases away due to over-smoothing, while deviation of the one that matches the values of the function at the knots  increases variability.  
On the other hand, the last method obtained by matching the derivatives at the two endpoints produces quite stable smoother `correction' of the initial function. 

It is clear that in the above consideration, the knots do not need start from $\xi_0$ and the order can be reversed from the right-to-left instead the left-to-right with all necessary but simple adjustments. All these properties can be used in many different ways to build the splines with a matrix of the derivatives  satisfying  (\ref{eq:SLR}). 
In our package, we implemented three basic approaches. 
Namely, the center-row-last-column (CR-LC), the center-row-first-column (CR-FC), and the regular row match (RRM).
All three methods are design for the case when the number of internal knots in the support is at least $2*k+2$. 
The cases with a smaller number of the knots in the support can be treated through the functional basis approach and function {\tt project()} that is discussed later in this work.
In what follows, we describe these constructions and the method for the spline objects that is implemented in the package.

\begin{widefigure}
\vspace{-1cm}
\begin{center}
\includegraphics[width=0.42\textwidth]{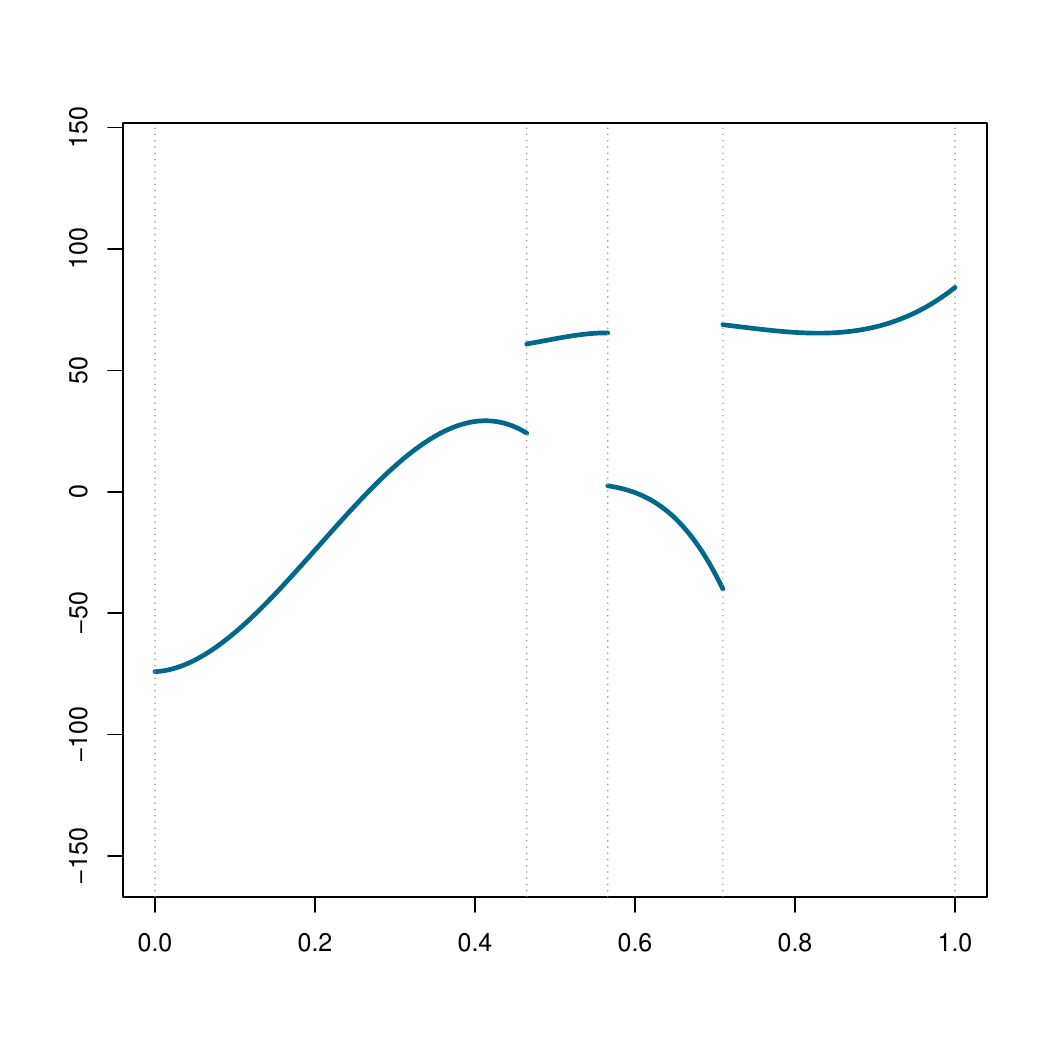}
\includegraphics[width=0.42\textwidth]{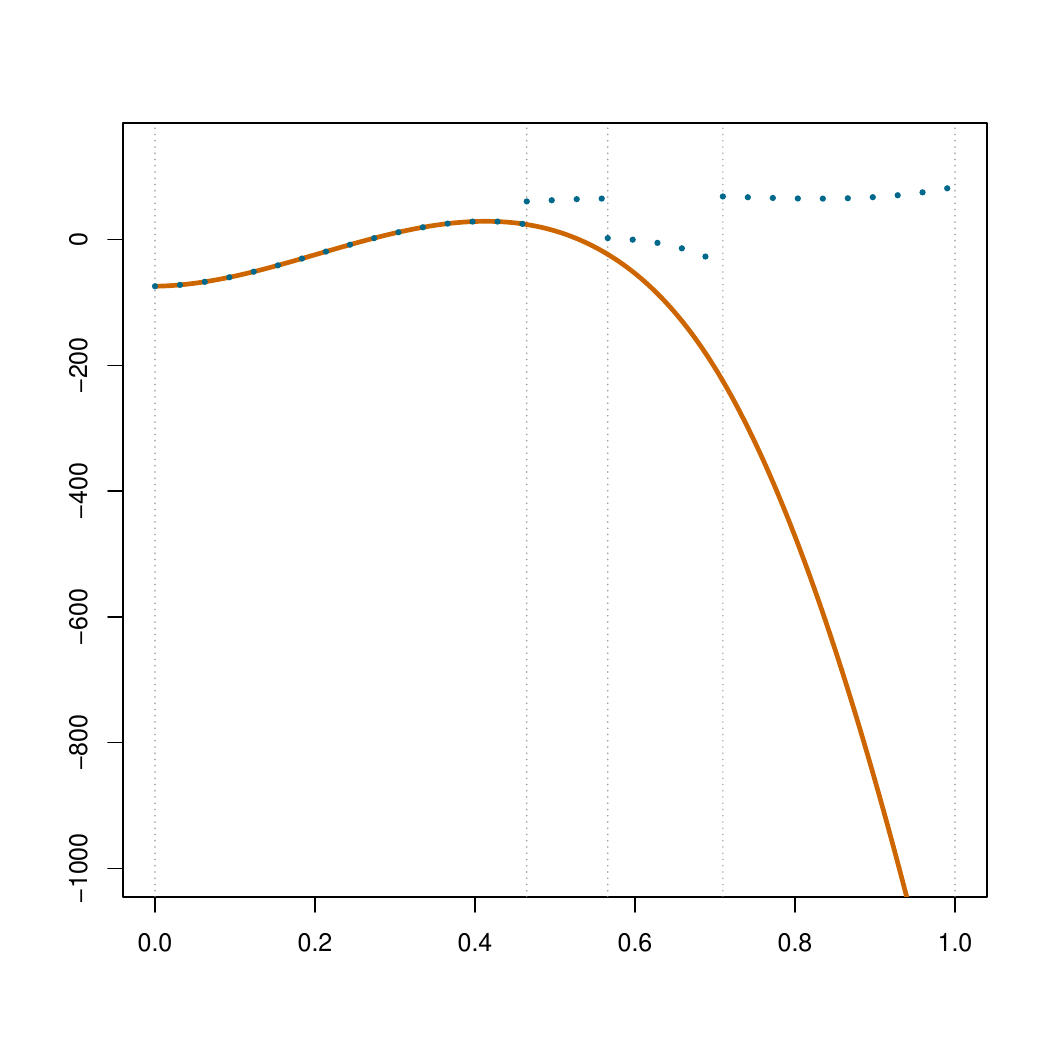}
\vspace{-1cm}
\\
\includegraphics[width=0.42\textwidth]{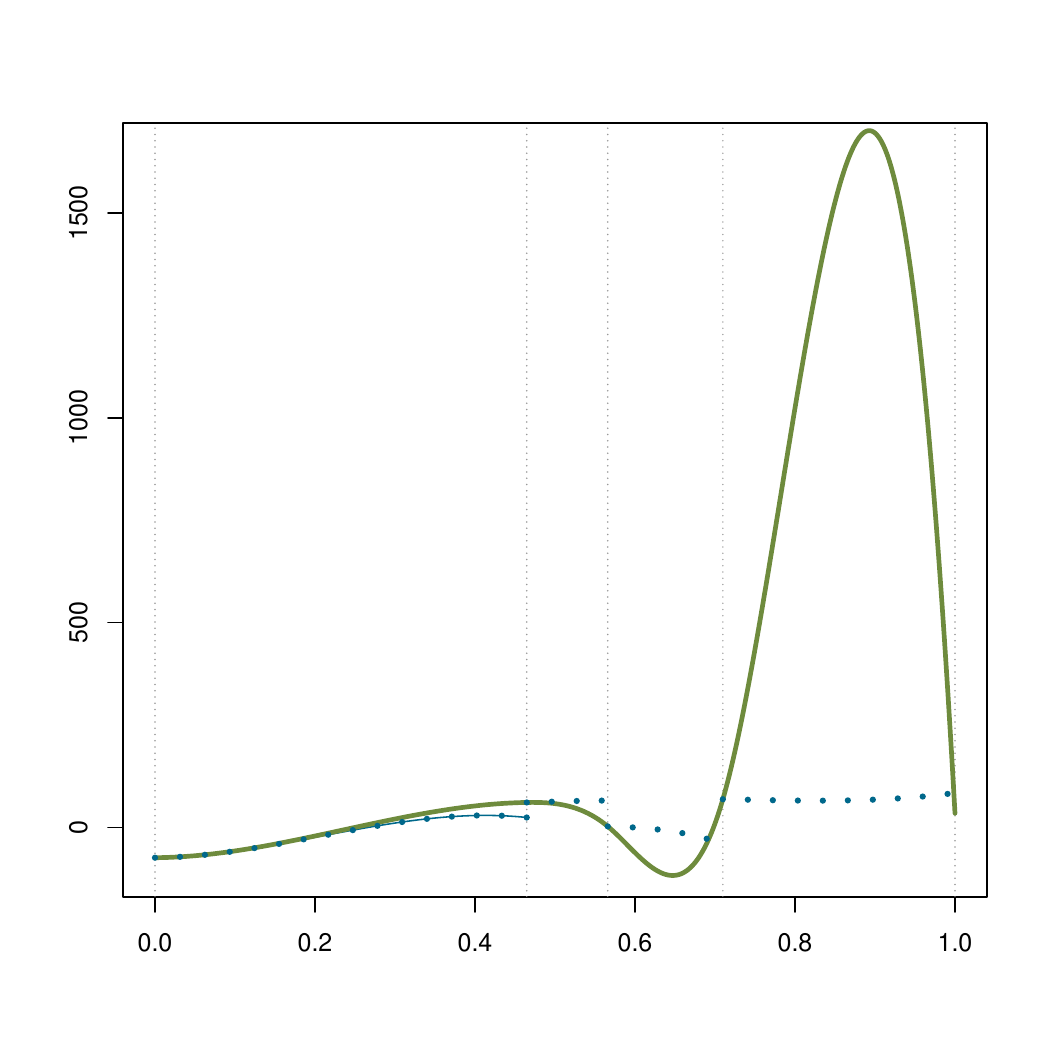}
\includegraphics[width=0.42\textwidth]{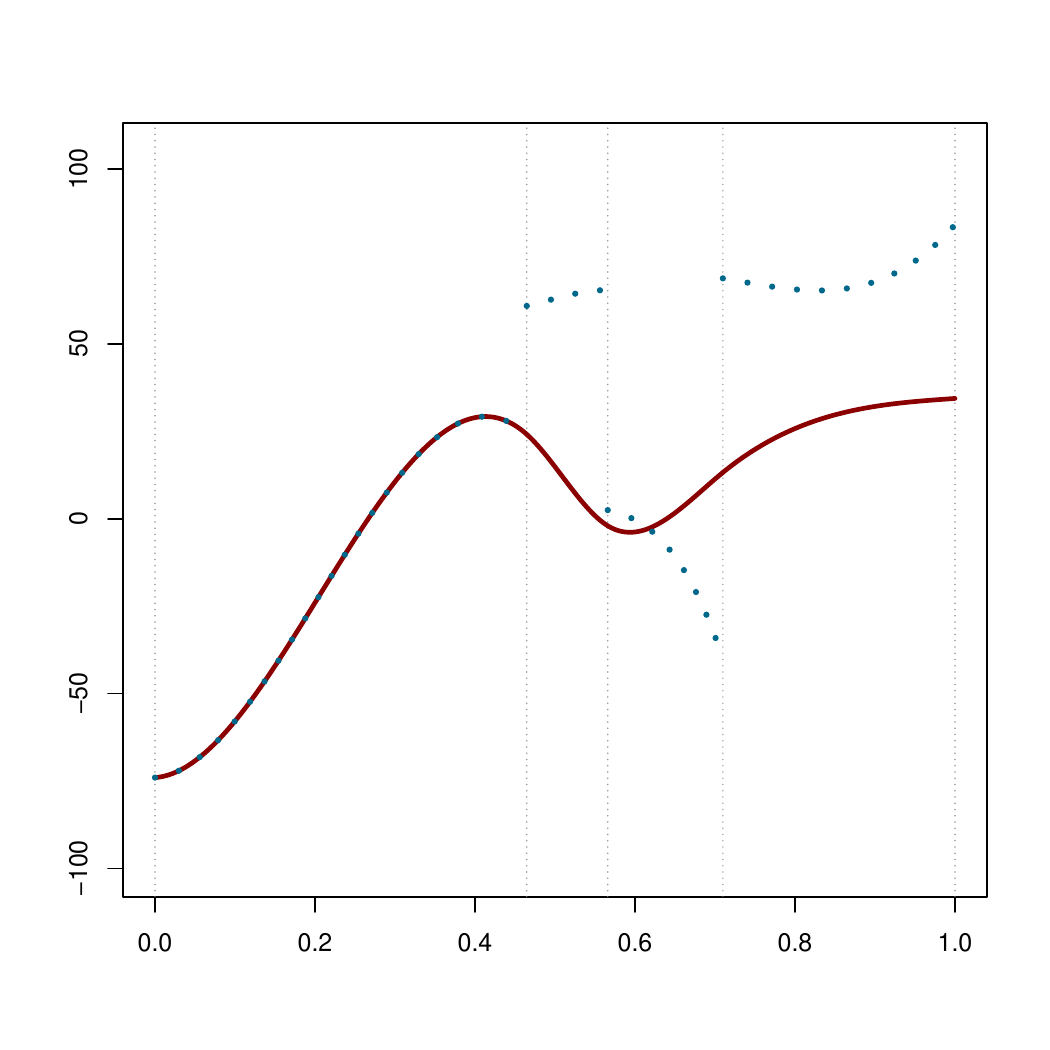} 
\vspace{-1cm}
\\
\includegraphics[width=0.42\textwidth]{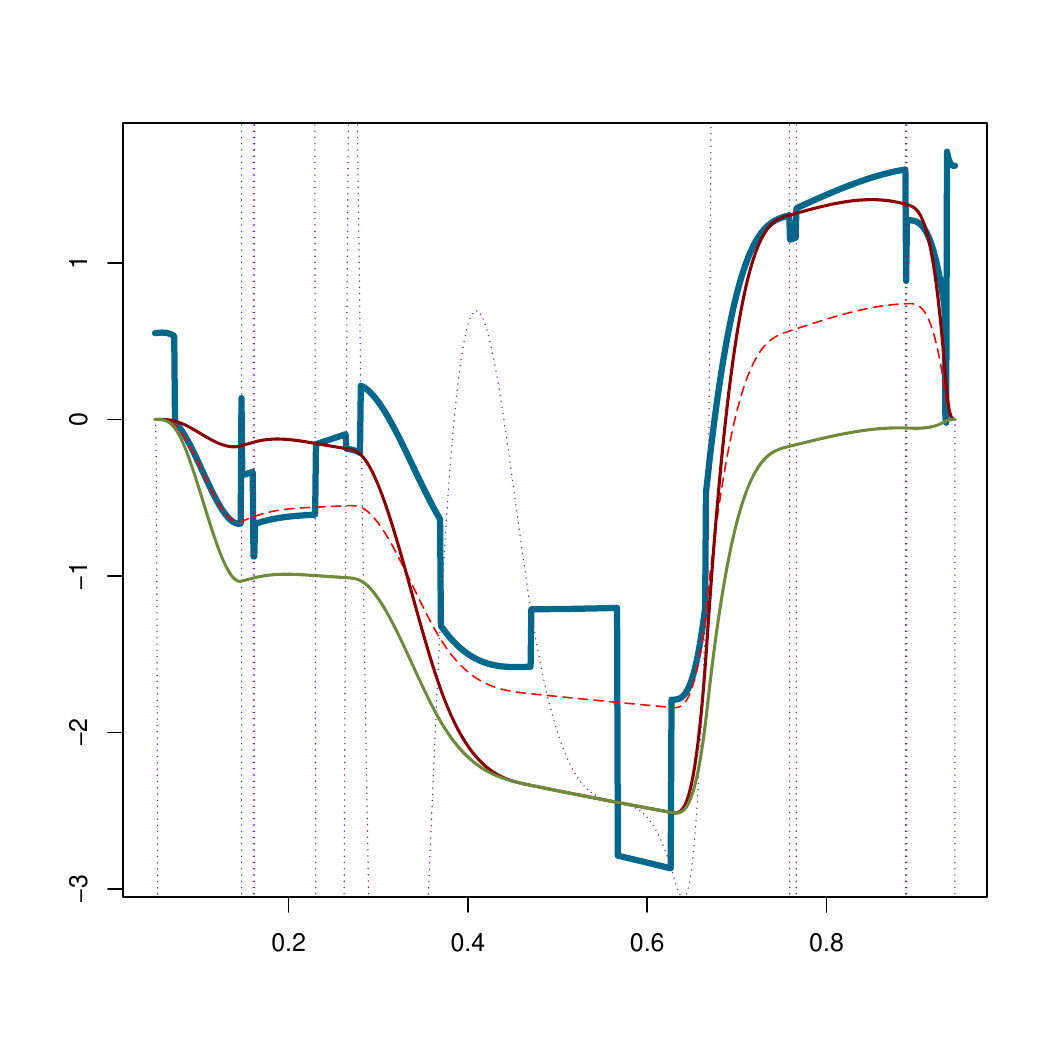}
\includegraphics[width=0.42\textwidth]{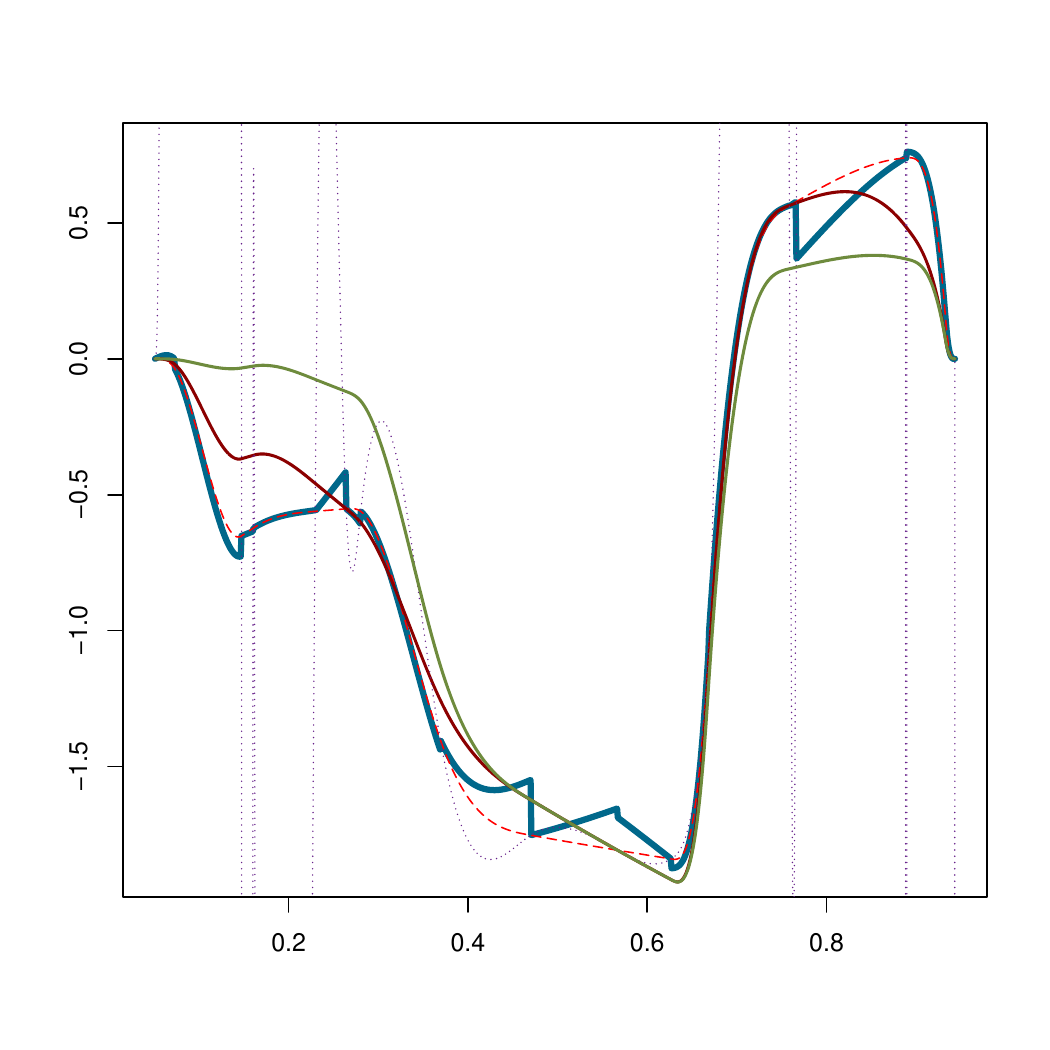}
\vspace{-1cm}
\end{center}
\caption{\small {\bf Top-Middle} -- The methods used in the direct spline building. {\it Top-Left:} The spline object that is not a spline but only a piecewise cubic polynomial without continuity at knots. {\it Top-Right:} A valid spline object obtained by matching the highest derivative at knots. {\it Middle-Left:} A spline obtained by matching the values of the function at knots. {\it  Middle-Right:} The most stable `correction' of the derivative matrix by matching derivatives at the endpoints (the RHS match is using the values that correspond to a RHS continuation of the function, so they do not match the function on the graph). 
{\bf Bottom} -- A spline given by a red dashed line is distorted and then reconstructed. In the LHS graph, all values at the knots are distorted by a random noise, while in the RHS graph only values of the derivatives are distorted.   The thick broken deep-sky-blue  line on both the graphs represents the resulting distorted piecewise polynomial function. The thin dotted lines going out of the bounds illustrate the CR-FC method, which tries to follow the values of the function at the knots. The RRM method represented by the dark-red lines seems to be the most accurate. The third line, olive-green, represents the CR-LC method.  }
\label{SpOb}
\end{widefigure}

\begin{description}
\item[The CR-LC method] In this method, we assume first that $k+1$ knots at both the initial and the terminal end point are excluded from considerations, so that we reduce to $n-2k$ knots. 
The reason for excluding these knots are the zero boundary conditions that uniquely determine the spline once the values of $\mathbf S$ at the $k+2^{\rm nd}$ 
knot from each endpoint are determined together with the value of the highest derivative over the interval between  the $k+1^{\rm st}$ knot and the $k+2^{\rm nd}$ 
knot.  It  will be discussed later in further detail. 

For these remaining knots, there are $n-2k-1$ values of the $k$th order derivative over the intervals between the knots.
In this approach, we set them to given values in the input. 
These leaves $k$ more independent values to be set as the dimension of the spline space is $n-k-1$. 
We set them through the derivatives at the center $\mathbf s^L_{l+1\,\, 0..k-1}$, where $l=[n/2]$ as before. 
It clearly also defines $\mathbf s^R_{l+1\,\, 0..k-1}$ (we use the Taylor expansion if $n$ is even, which is equivalent to applying `frlr' with $m=0$ to these two columns). 

It is convenient to reverse the order in the matrices  $\mathbf S^L_{k+2..l+1\,\,\cdot}$ and $\mathbf S^R_{k+2..l+1\,\,\cdot}$ and to consider `mirror' matrices $\widetilde{\mathbf S}^L=\widetilde{\mathbf S}^L_{0..l-k-1\,\,\cdot}$ and $\widetilde{\mathbf S}^R=\widetilde{\mathbf S}^R_{0..l-k-1\,\,\cdot}$, where also the values of the $k$th derivatives, i.e. $\tilde s^L_{i,k+1}$ and $\tilde s^R_{i,k+1}$,  are assigned so that they are over the interval corresponding to pairs  $(\widetilde{\mathbf S}^L_{i\,\,\cdot}, \widetilde{\mathbf S}^L_{i+1\,\,\cdot})$ and $(\widetilde{\mathbf S}^L_{i\,\,\cdot}, \widetilde{\mathbf S}^L_{i+1\,\,\cdot})$, respectively,  $i=0,\dots,l-k-1$. 
It is clear that by setting the values at the center knots as described above we have given the first rows in $\widetilde{\mathbf S}^L$ and $\widetilde{\mathbf S}^R$. 
Moreover, setting the $k$th derivatives over all the $n-2k$ knots (including over the interval between  the $k+1^{\rm st}$ knot and the $k+2^{\rm nd}$ 
knot from each end) is equivalent to setting the last columns in each of the two matrices $\widetilde{\mathbf S}^L$ and $\widetilde{\mathbf S}^R$.
By applying the `frlc' approach to these two matrices, we correct all the entries in these two matrices so that the matrix correspond to a spline. 
We emphasize that we also obtain the $k$th derivatives $s^L_{k,k+1}$ and $s^R_{k,k+1}$ at the $k+1^{\rm st}$ knot counting from each endpoint, respectively. 

Now, the values at these  knots  at the endpoints that were initially excluded, i.e.  $\mathbf S^L_{0..k+1\,\,\cdot}$ and $\mathbf S^R_{0..k+1\,\,\cdot}$, are handled through the `frlr' approach. Since our method is symmetric, it is enough to consider only one end-point, say the LHS. 
At this end, we have already the $k+1$ values of the derivatives at the $k+2$ knot (including the $k$th derivative over the interval between  the $k+1^{\rm st}$ knot and the $k+2^{\rm nd}$ 
knot)  and the $k$ values $\mathbf S_{0, 0..k-1}$, which are zero because of the zero boundary condition. 
Starting from the RHS and continuing to the LHS, one can obtain uniquely all the remaining entries in $\mathbf S^L_{0..k+1\,\,\cdot}$ using `frlr' approach for $\mathbf U = \widetilde{\mathbf S}_{0..k+1,\cdot}$, which is the mirror of $\mathbf S_{0..k+1,\cdot}$ with reassigned the $k$th derivatives. In fact, it guarantees the unique solution to the problem (the case $m=k$).

\item[The CR-FC method] The method is very similar to the previous one except this time we set {\it values} of the spline at $n-2k+2$ knots. 
This leaves $k-1$ more values to be set, which we do at the center except for the $k$th derivative. All other entries can be evaluated as described in the `frfc' approach for the matrix $\mathbf U$ set to  $\widetilde{\mathbf S}^L_{0..l-k\,\,\cdot}$ and $\widetilde{\mathbf S}^R_{0..l-k\,\,\cdot}$, respectively, and then evaluating the values at the $k+2$ knots at each endpoint  as described in the second part of the CR-LC method.

\item[The RRM method] This methods uses only the `frlr' approach which produces most stable results for large number of knots relatively to the order of the spline as illustrated in Figure~\ref{SpOb}.
The method starts the same way as the two previous methods by setting $k+1$ values at the center, dividing matrix $\mathbf S$ into $\mathbf S^L$ and $\mathbf S^R$ and, in the even $n$ case, settling first the values at the two central knots using `frlr'  with $m=0$, as described before. 
Having then the central values settled, it uses `frlr'  with $m=k$, starting from these central knots  through the successive groups of $k+1$ knots as long as the number of the evaluated rows in  $\widetilde{\mathbf S}^L$ and $\widetilde{\mathbf S}^R$, respectively,  does not exceed $l-k+1$, i.e. obtaining values $\widetilde{\mathbf S}^L_{0..j(k+1)\,\,\cdot}$ and $\widetilde{\mathbf S}^R_{0..j(k+1)\,\,\cdot}$, for $j$ being integer part of $(l-k)/(k+1)$. 
Then the `frlr' approach is used for $\widetilde{\mathbf S}^L_{j(k+1)..l-k+1\,\,\cdot}$ and $\widetilde{\mathbf S}^R_{j(k+1)..l-k+1\,\,\cdot}$, unless $j(k+1)=l-k+1$.
The last step for the final $k+2$ knots at each end is the same as in the two previous methods. 
\end{description}

In the package, these three methods are implemented in the method {\tt construct()}.
In Figure~\ref{SpOb}~{\it(Bottom)}, we present an illustration of  {\tt Splinets} objects as results of application of the functions {\tt construct()}.
There a spline has been distorted by addition of some random noise to its matrix of the derivatives. 
Then the three above described methods has been used to correct the matrix of derivative and construct valid splines. 
The methods RRM and CR-LC perform similarly with the RRM doing a better job when only the derivatives are distorted (right). 
The CR-FC that tries to follows the values of the distorted function performs badly at the boundaries but improves its performance at the center when the values of the function are not distorted (right). 

\subsection{Random spline generator}
In many applications, it is convenient to have an effective simulations of random functions. 
The package has implemented a simple approach through simulation of the error noise around the matrix of the derivatives at knots and then using our construction methods to obtain a valid spline.
Formally, the model from which simulations are made can be written in the terms of the matrices of the derivatives
$$
{\mathbf T}=\mathbf S+\mathbf C\left(\boldsymbol\epsilon(\boldsymbol \Sigma, \boldsymbol \Theta)\right),
$$
where ${\mathbf T}$ is the $(n+2)\times (k+1)$ matrix of derivatives for a generated random spline, $\mathbf S$ is the analogous matrix for its mean value spline, $\boldsymbol\epsilon(\boldsymbol \Sigma, \boldsymbol \Theta)$ is a $(n+2)\times (k+1)$ random matrix generated from the mean-zero matrix valued normal distribution with the covariances $\boldsymbol \Sigma$  and $\boldsymbol \Theta$, $(n+2)\times(n+2)$ and  $(k+1)\times(k+1)$, respectively, non-negatively defined matrices. Namely, for a matrix $\mathbf Z$ of iid standard normal variables
$$
\boldsymbol\epsilon(\boldsymbol \Sigma, \boldsymbol \Theta)=\boldsymbol \Sigma^{1/2}\mathbf Z \boldsymbol \Theta^{1/2}.
$$
Finally, $C(\cdot)$ is a chosen correction method  of an arbitrary $(n+2)\times (k+1)$ matrix so it satisfies the conditions required for it to be a matrix of derivative for a spline as described in the previous sections. 
In the package implementation, this is one of the three methods: RRM, CR-LC, CR-FC. 
In Figure~\ref{RanSp}, we see samples random functions generated using the package implementation {\tt rspline()}, using different variances and different spline correction methods.
\begin{widefigure}[t!]
\vspace{-1cm}
\begin{center}
\includegraphics[width=0.42\textwidth]{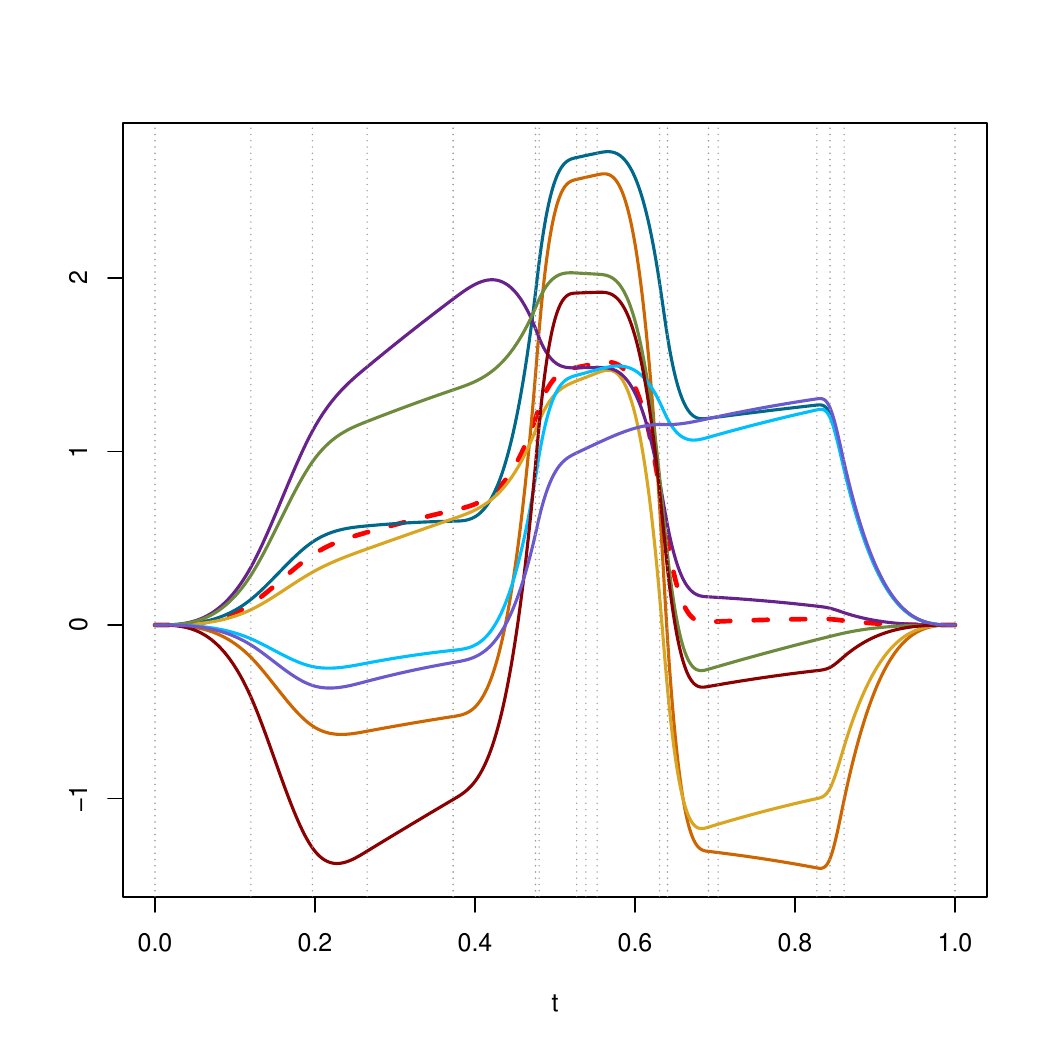}
\includegraphics[width=0.42\textwidth]{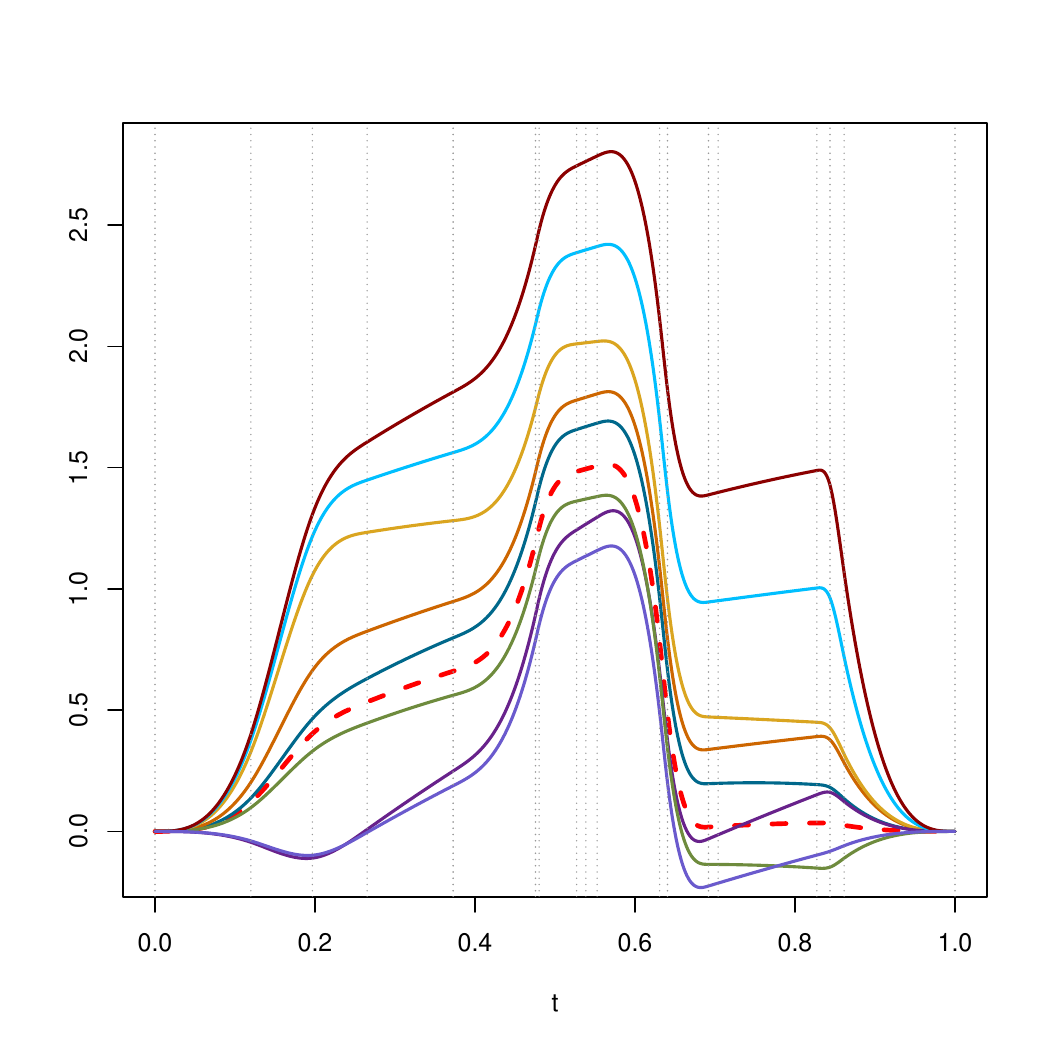}
\vspace{-1cm} \\
\includegraphics[width=0.42\textwidth]{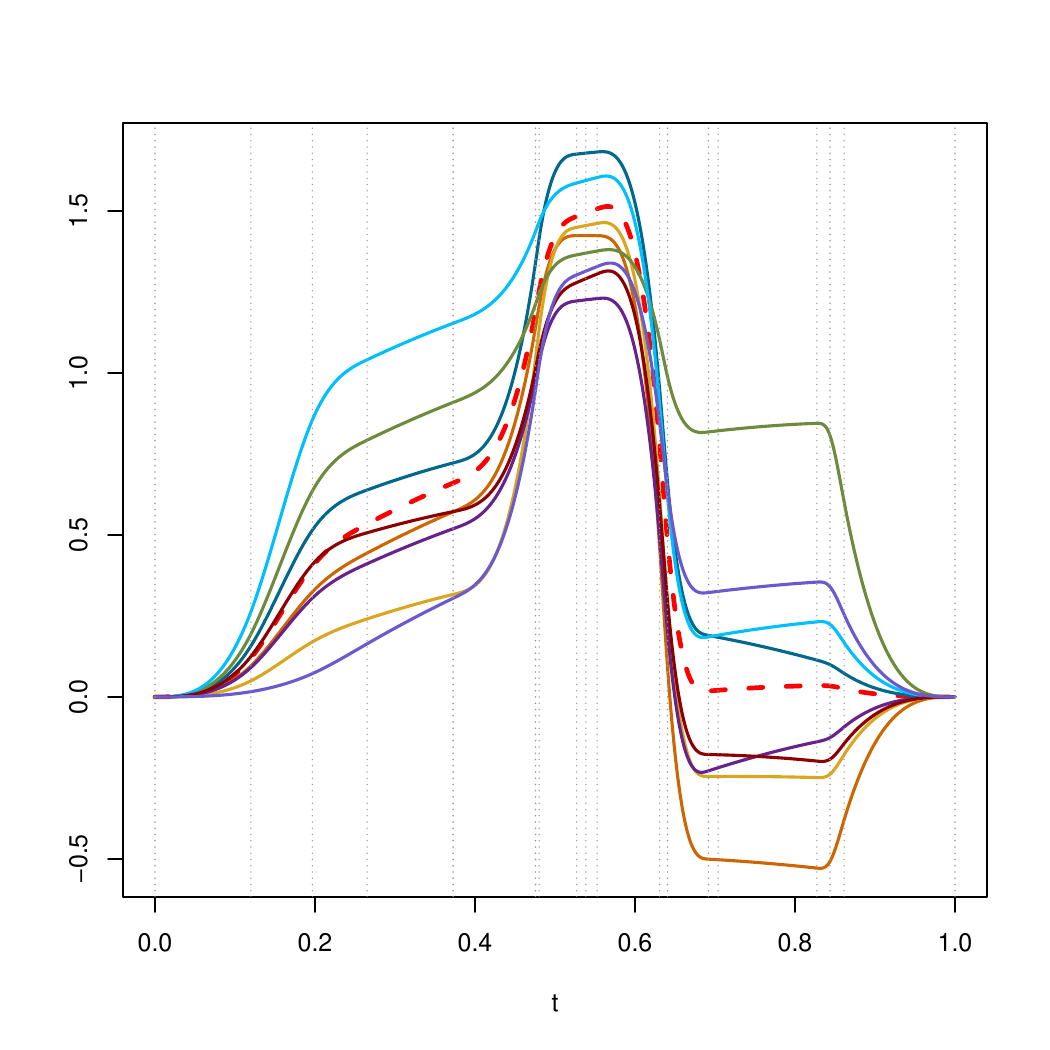}
\includegraphics[width=0.42\textwidth]{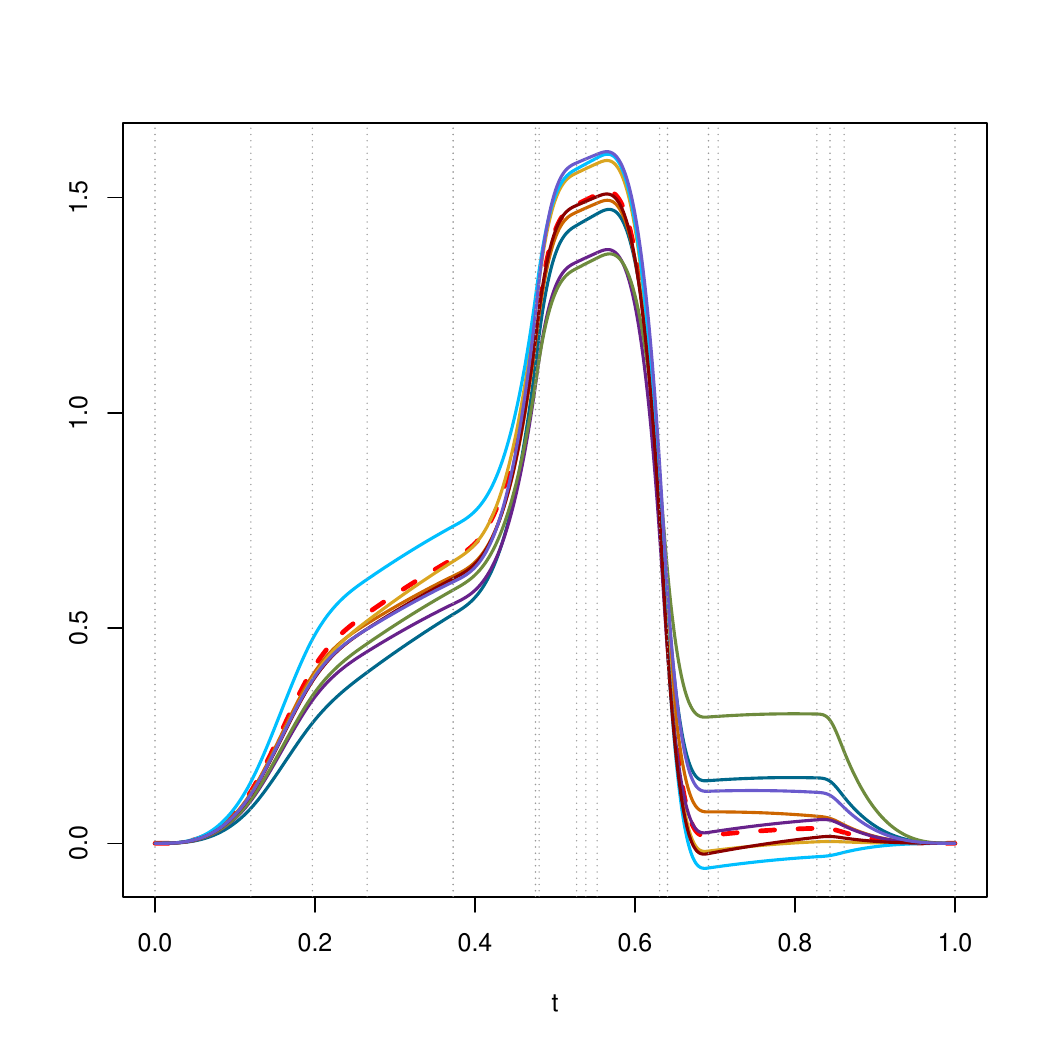}
\vspace{-1cm}
\end{center}
\caption{\small Random samples of splines. The dashed lines in the middle of the packs of random splines represent the mean value spline which is the same for all the cases. 
The vertical dashed lines are passing through the knot locations. {\it Left-Top}: Independent splines with standard normal variables in the matrix of the derivatives, the RRM correction method.  
{\it Left-Bottom}: The same but through the CR-LC correction method.
\ {\it Right}: Using smaller variances for knots closer to zero.  }
\label{RanSp}
\end{widefigure}
\section{Basic operations on splines}\label{sec:oper}
Splines are elements of a Hilbert space and thus their linear combination and the inner products are readily defined. 
In the representation of the package splines are kept in the main object  $ \widetilde{\mathcal{S}}$  as a sequence $\left\{ \mathcal{S}_1,\dots, \mathcal{S}_d \right\}$ of quadruples
$$
\mathcal{S}_j = \left\{k, \boldsymbol \xi, \mathcal{I}_j, \mathbf{S}_j\right\},
$$
where, $k$ is the common smoothness order, the common knots are $\boldsymbol \xi = (\xi_0, \xi_1,\dots, \xi_{n+1})$.
The derivative matrix $\mathbf{S}_j$, $j=1,\dots,d$ for each spline is a sequence of  matrices defining the derivative over components of the support set given through pairs of the indices in $\mathcal I_j$.
Here, for simplicity, we assume the one-sided representation of each matrix in $\mathbf{S}_j$ and we identify $\mathcal I_j$ with a support set in the range of the splines.
\subsection{Embedding a spline into higher dimension spaces of splines}
One of the most interesting aspects of the spline spaces are their agility following from different choices of the knots. 
However, this is not that often utilized with most of the focus typically being on the splines under the fixed set of knots. 
In the proposed package, we  provide tools to fully explore the properties of the splines under different choices of knots. 
 Any spline of a given order remains a spline of the same order if one considers it on a larger set of knots than the original one.
However, this changes {\tt Splinets} representation of the so refined spline. 
It is thus important to have a function that embeds a given spline into the bigger space of splines residing on a refined set of knots. 
In the package the function \code{refine()} does the task allowing conveniently represent splines from smaller spaced in the bigger more refined spaces.  
This function will be utilized in the final section of this work, where projections to spline spaces are discussed.  
\subsection{Linear combination}
The linear combination of splines could be easily implemented if the splines have full supports.
In this case, a linear combination of splines correspond simply to taking the same linear combination of the matrices of the values of the derivatives at the knots.
Thus for two splines, 
$\mathcal{S} = \left\{k, \boldsymbol \xi, (1, n+2), \mathbf{S}\right\}$ and $\widetilde{\mathcal{S}} = \left\{k, \boldsymbol \xi, (1, n+2), \widetilde{\mathbf{S}}\right\}$, 
their linear combination with coefficients $\alpha$ and $\widetilde{\alpha}$  has the matrix of the derivatives given by 
$$
\mathbf W=\alpha \mathbf S+\widetilde{\alpha}\widetilde{\mathbf S}.
$$
One of the important features of the package, is utilization of the support sets.  
Thus for general  $\mathcal{I}$ and $\widetilde{\mathcal{I}}$, we utilize them as follows.
First, we note that the support of the linear combination is at most $\mathcal J=\mathcal I \cup \widetilde{\mathcal I}$.
Then the matrix of the derivatives over this support for the linear combination is given by $\mathbf{W}$ which is obtained from
\begin{eqnarray*}
{\mathbf{W}}_{\mathcal{I}_1 \cdot} &=& \alpha \mathbf{S}_{\mathcal{I}_1 \cdot} \\
{\mathbf{W}}_{\mathcal{I}_2 \cdot} &=& \alpha \mathbf{S}_{\mathcal{I}_2 \cdot}+\widetilde{\alpha}\widetilde{\mathbf{S}}_{\mathcal{I}_2 \cdot}\\
{\mathbf{W}}_{\mathcal{I}_3 \cdot} &=& \widetilde{\alpha}\widetilde{\mathbf{S}}_{\mathcal{I}_3\cdot}
\end{eqnarray*}
where $\mathcal I_1=\mathcal I \setminus \mathcal I_2$, $\mathcal I_2=\mathcal I \cap \widetilde{\mathcal I}$, and $\mathcal I_3= \widetilde{\mathcal I} \setminus \mathcal I_2$
and subindexing the matrices by the support components stands for considering these parts of them that correspond to the respective components of the support. 

More general, denote the operation of linear transformation as $\ell(\widetilde{\mathcal{S}}, \mathbf{P})$, where $\mathbf{P} = \{p_{ij}\}$ is a $m \times d$ transformation matrix. 
The linear transformation operator define a map from an input {\tt Splinets}-object of size $d$ to the same order {\tt Splinets}-object of size $m$. 
The derivative matrix of the $i$th spline in the output can be calculated in the full support case as
$$
\sum_{j=1}^d p_{ij} \mathbf{S}_{j}.
$$ 
Similarly to the case of a linear combination of two splines,  the support sets can be utilized to improve the computational efficiency.
The operation of the linear combination of splines is implemented in {\tt lincomb()}.

\subsection{Derivative and integral}
If $t\in [\xi_{i},\xi_{i+1}]$, $i=0,\dots, n$, then 
$$
S(t)=\sum_{l=0}^k  \frac{(t-\xi_{i})^l}{l!} s_{i,l}.
$$
The derivative at this point is 
$$
S'(t)=\sum_{l=1}^k  \frac{(t-\xi_{i})^l}{l!} s_{i,l}.
$$
Thus the derivative matrix of the first derivative of a spline is given by
$$
\mathcal{S}(S') = \left\{k-1, \boldsymbol \xi, \mathcal{I}, \mathbf{s}_1,..., \mathbf{s}_k \right\}.
$$
This simple evaluation of the spline derivative is implemented in {\tt deriva()}.

The indefinite integral of a spline within each interval $[\xi_i, \xi_{i+1}]$, is given by
\begin{equation}
	\label{eq:indefInt}
	\int S(t) dt = \sum_{l=1}^k  \frac{(t-\xi_{i})^{l}}{l!} s_{i,l-1} + c_i,
\end{equation}
where $i = 0,\dots,n$ and $c_i$ is an arbitrary real number. 
There are many ways to determine a specific inverse derivative function by specifying $c_i$'s. 
We choose the one that preserve the spline smoothness but will not necessary yield zero boundary condition on the RHS terminal knot. 

Set $c_0 = 0$, for $i = 1,\dots,n+1$. 
To achieve the continuity at the knots $c_i$ is iteratively calculated by
$$
c_i = [\mathbf{S}]_{i-1,.} \mathbf{A}^*_{\xi_{i+1}-\xi_{i}},
$$
where $\mathbf{A}^*_{\xi_{i+1}-\xi_{i}}$ is slightly different from \cref{eq:A}.
Now, it is a column vector of the Taylor coefficients from order $1$ to order $k+1$. 
$$
A_{\alpha}^* = \left(\alpha, \frac{\alpha^2}{2!},\dots, \frac{\alpha^{(k+1)}}{(k+1)!}\right)^T.
$$
One can notice that the resulting splines do not necessarily satisfy the boundary condition.
The boundary condition of the output is only guaranteed when the definite integral of the input spline over the entire range of knots vanishes.
For example, the derivative of a spline that satisfies the boundary condition has this property.
In this sense the operation constitute the inverse of the derivative. 
More specifically, if $\mathcal I$ is the integration operator and $\mathcal D$ is the differentiation operator, then 
$$
\mathcal I \mathcal D = \mathcal D \mathcal I=\mathbf I, 
$$
where $\mathbf I$ is the identity operator on $\mathcal S_k^{\boldsymbol \xi}$, and 
$$
\mathcal D \mathcal S_k^{\boldsymbol \xi}= \mathcal S_{k-1}^{\boldsymbol \xi},\, \, \mathcal I \mathcal S_k^{\boldsymbol \xi} \varsupsetneq  \mathcal S_{k+1}^{\boldsymbol \xi}.
$$
This operation is implemented as a function {\tt integra()}. 

The definite integral of a spline can be calculated through \cref{eq:indefInt}. The integral of spline $S(t)$ within interval $[\xi_i, \xi_{i+1}]$ is
$$
\int_{\xi_{i}}^{\xi_{i+1}} S(t)dt = [\mathbf{S}]_{i,.} [\mathbf{A}^*_{\xi_{i+1}-\xi_{i}}]_{.,0}  
$$
The definite integral of a spline is implemented in {\tt dintegra()}. 
\subsection{Inner product}
Finally, we demonstrate how the topology induced by the inner product in the space of splines can be expressed in the terms of the coefficients of the matrices. 
The calculation of mutual inner products of a set of splines is implemented in {\tt gramian()}. 
The algorithm for calculating inner product is established in the following proposition.  
\begin{proposition}
\label{prop:topo}
Let $\tilde{\mathcal S}$ be the $ n + k + 1$ dimensional space  of $(n+2)\times (k+1)$ matrices $\mathbf S$ as in \cref{eq:spmat} satisfying \cref{eq:addk-1nL}. 
For $\mathbf S,\tilde{\mathbf S} \in \tilde{\mathcal S}$, let us define the inner product 
$$
 \langle \mathbf S, \tilde{\mathbf S}  \rangle 
 =
\begin{pmatrix} 1 & \frac 12 & \dots & \frac 1{2k+1} \end{pmatrix} 
\sum_{i=0}^n  
(\mathbf A_{\xi_{i+1}-\xi_{i}}^{*T} \cdot \mathbf S_{i\cdot})*({\mathbf A_{\xi_{i+1}-\xi_{i}}^{*T} \cdot \tilde{\mathbf S}}_{i\cdot}),
$$
Here we use the following notations and conventions:  for two  $r\times 1$ vectors $\mathbf v$ and $\mathbf w$,  their convolution is  a  $(2r-1)\times 1$ vector defined by 
$$
\mathbf v * \mathbf w=
\left(\sum_{m=(p-r+1)\vee 1}^{p\wedge r} v_{p-m+1}w_{m}
\right)_{p=1}^{2r-1}
,
$$ 
 while $\mathbf v \cdot \mathbf w$ is coordinate-wise multiplication of vectors. Moreover, for a matrix $\mathbf X$,  its  $i^{ th}$ row is denoted by $\mathbf X_{i\cdot}$. 
 
Then $\tilde{\mathcal S}$ equipped with this inner product is isomorphic with the space of splines of the $k$th order spanned over the knots $\xi_0,\dots, \xi_{n+1}$ equipped with the standard inner product of the square integrable functions.  
\end{proposition}

\begin{proof}
For a given set of knots $\boldsymbol \xi$ let $S$ and $\tilde S$ be the (unique) splines such that $\mathcal S_0 (S)= \mathbf S$ and $\mathcal S_0 (\tilde S)=\tilde{\mathbf S}$.
Then 
\begin{align*}
\langle S, \tilde{S}\rangle 
&= 
\sum_{i=0}^n \int_{\xi_{i}}^{\xi{i+1}} S(t) \tilde S(t)~ dt\\
&=
\sum_{i=0}^n \int_{\xi_{i}}^{\xi{i+1}} 
\sum_{j=0}^k s_{ij} \frac{(t-\xi_i)^j}{j!}
\sum_{j=0}^k \tilde s_{ij} \frac{(t-\xi_i)^j}{j!}
~ dt
\\
&=
\sum_{i=0}^n  
\sum_{j,r=0}^k \frac{s_{ij}\tilde s_{ir}}{j!r!} \int_{\xi_{i}}^{\xi{i+1}} (t-\xi_i)^{j+r} ~dt
\\
&=
\sum_{i=0}^n  
\sum_{j,r=0}^k \frac{s_{ij}\tilde s_{ir}}{j!r!} \frac{(\xi_{i+1}-\xi_i)^{j+r+1}}{j+r+1}
\\
&=
\sum_{i=0}^n  
\sum_{l=0}^{2k}\frac{(\xi_{i+1}-\xi_i)^{l+1}}{l+1}\sum_{m=(l-k)\vee 0}^{l\wedge k} \frac{s_{il-m}\tilde s_{im}}{(l-m)!m!}
\\
&=
\sum_{i=0}^n  
\sum_{l=0}^{2k}
\frac{1}{l+1}
\sum_{m=(l-k)\vee 0}^{l\wedge k} 
\frac{
 s_{il-m}(\xi_{i+1}-\xi_i)^{l-m+1/2}
}{
(l-m)!
}
\frac{
\tilde s_{im}(\xi_{i+1}-\xi_i)^{m+1/2}
}{
m!
},
\end{align*}
which shows the isometry property of the mapping $\mathcal S_0$. 
\end{proof}

\section{Bases of splines and their orthogonalizations}
As we have seen in the previous sections, the direct approach to building splines requires a lot of care and often can be cumbersome. 
A much better way to build splines is through functional bases of splines. 
There are many possible choices of such bases but the most popular are the $B$-splines. 
Despite having many advantages, the $B$-splines do not constitute an orthogonal basis. 
Our main contribution is to implement an optimal orthogonalization of the $B$-splines introduced in \cite{Liu2019SplinetsE}.
The presentation of this spline basis benefits from organizing them in the form of a net. 
In this framework, the derived orthogonal bases of splines are referred to as the splinets. 


\begin{figure}[t!]
  \centering
  \raisebox{1cm}{
  \includegraphics[width=0.43\textwidth]{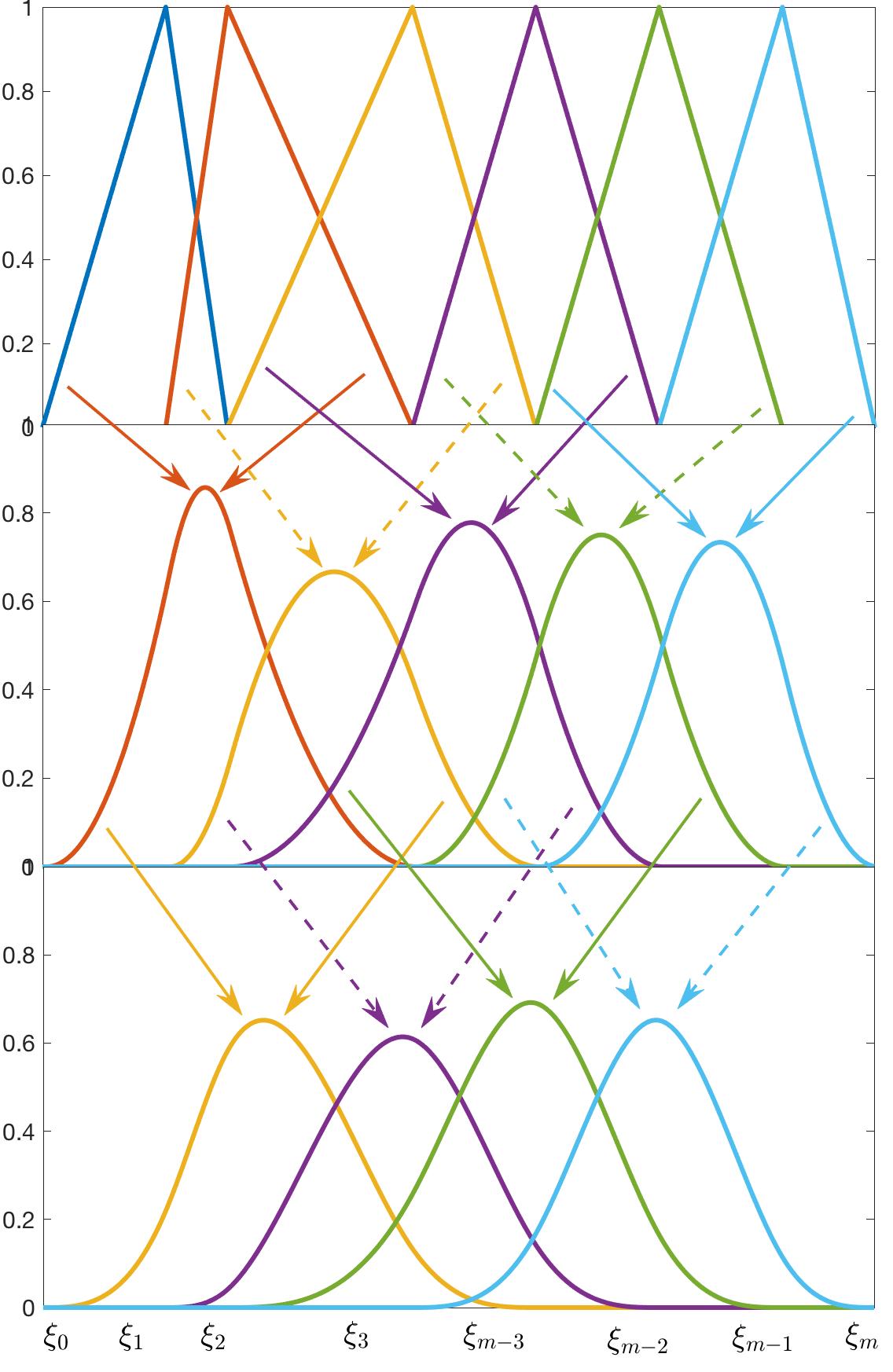}\hspace{-1mm}
  }
  \parbox[b]{0.55\textwidth}{ 
  \includegraphics[width=0.6\textwidth]{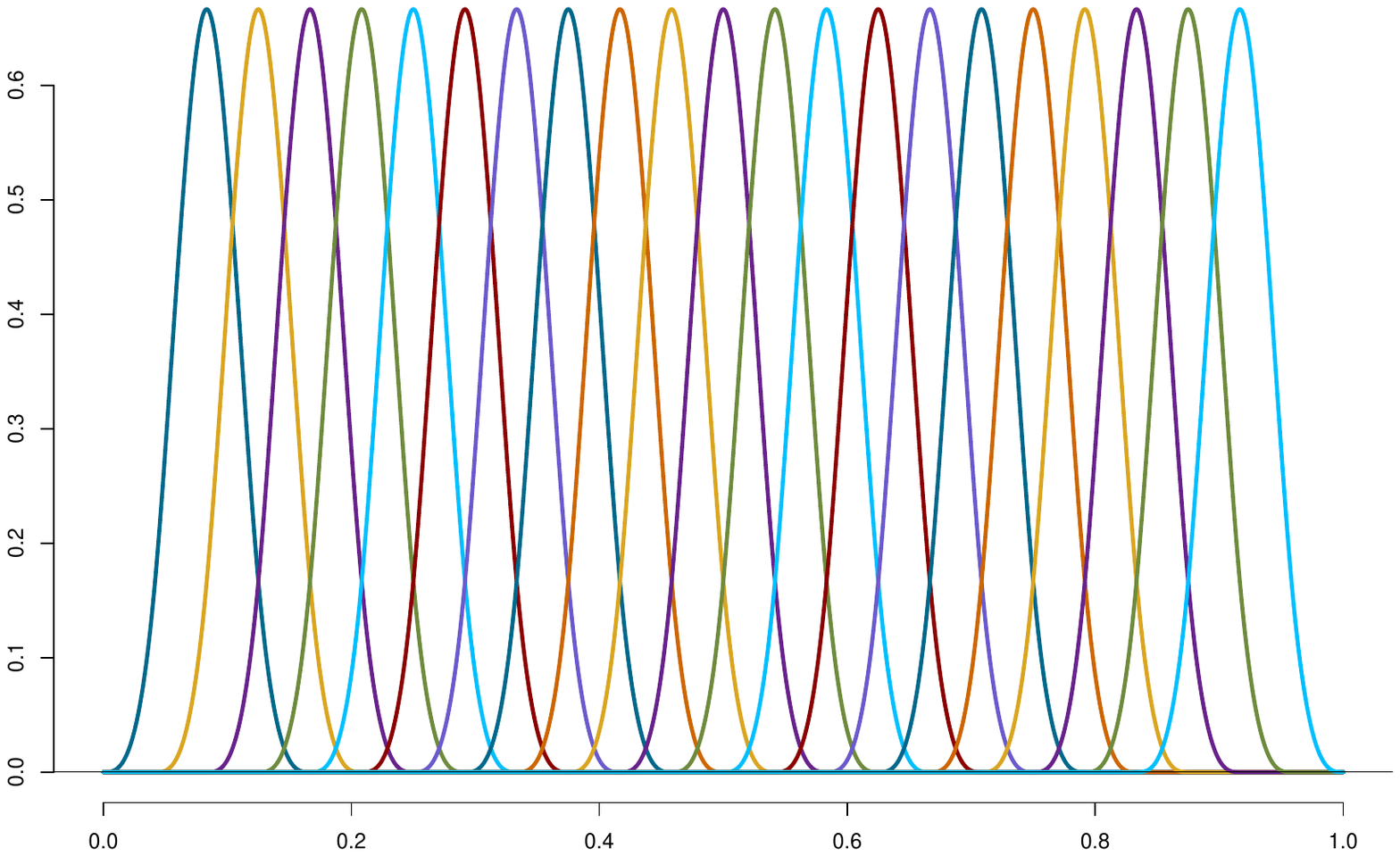}\vspace{-7mm}
  \includegraphics[width=0.6\textwidth]{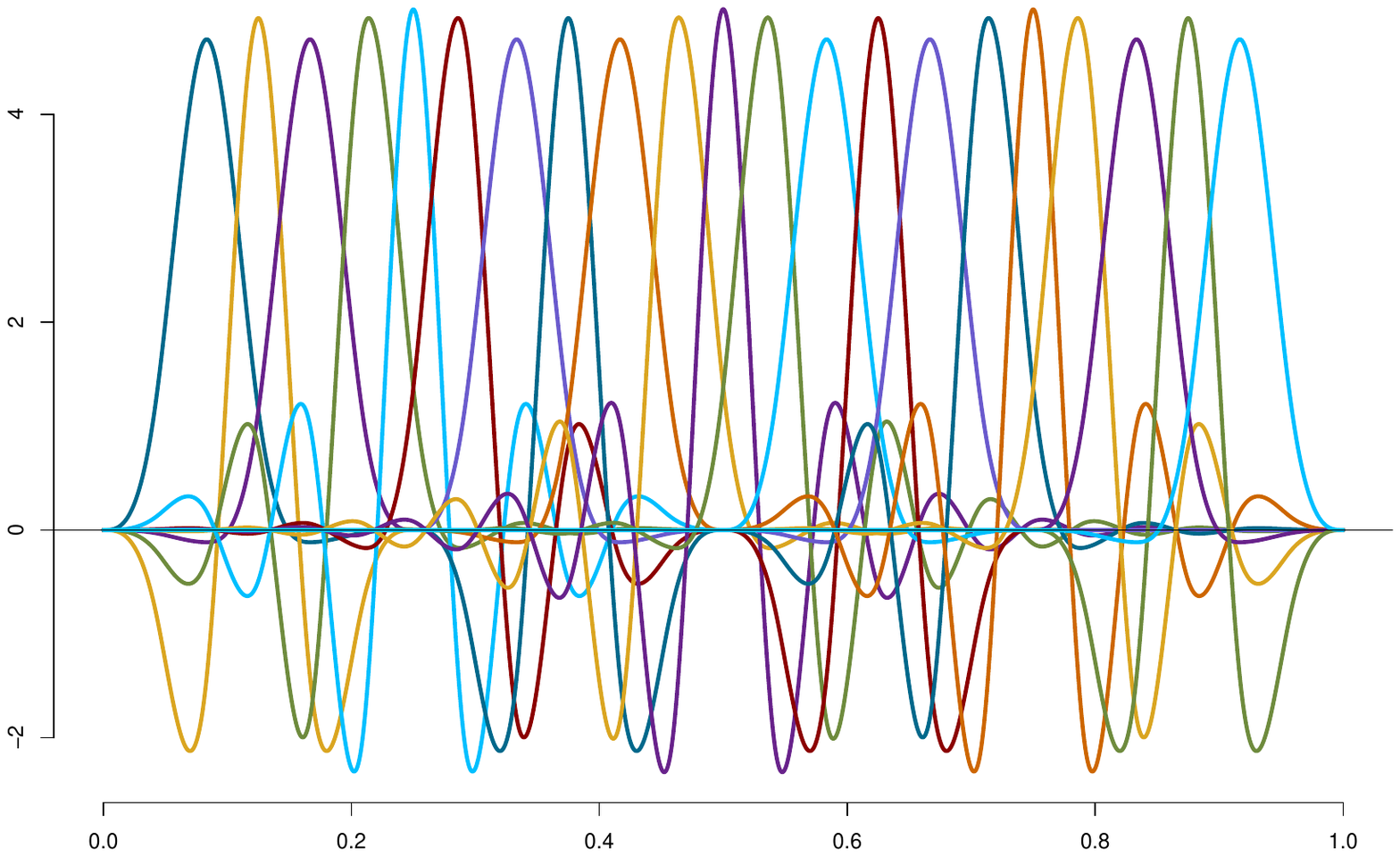}
  }
  \caption{{\it Left:}\, The recursion in definition of the $B$-splines, the first order splines {\it (top)}, the second order spline {\it (middle)},  and the third order spline {\it (bottom)}; 
  {\it Right:}\, The third order $B$-Splines on equidistant grid {\it (top)} and the splinet build of them {\it (bottom)}.}
  \label{fig:RecSp}
\end{figure}

\subsection{$B$-splines}
\label{sec:basics}
This section unifies the notation and provides the most fundamental facts about the $B$-splines. 
The most convenient way to define the $B$-splines on the knots $\boldsymbol \xi =\left(\xi_0,\dots, \xi_{n+1} \right)$, $n=0,1,\dots$ is through the splines with the boundary conditions and using the recurrence on the order. 
Namely, once the $B$-splines of a certain order are defined, then the $B$-splines of the next order are easily expressed by their `less one' order counterparts. 
In the process, the number of the splines decreases by one and the number of the initial conditions (derivatives equal to zero) increases by one at each endpoint. 
We keep the notation $B^{\boldsymbol \xi}_{l,k}$, for the $l$th $B$-spline of the order $k$, $l=0,\dots ,n-k$. 
For the zero order splines, the $B$-spline basis is made of indicator functions 
\begin{equation}
\label{eq:indi}
B^{\boldsymbol \xi}_{l,0}=\mathbb I_{(\xi_{l},\xi_{l+1}]}, ~~~l=0,\dots n,
\end{equation}
 for the total of $n+1$-elements and zero initial conditions. 
Clearly, the space of zero order splines (piecewise constant functions) is $n+1$ dimensional so the so-defined zero order $B$-splines constitute an orthogonal basis.

The following recursion relation leads to the $B$-splines of an arbitrary order $k \le n$.
Suppose now that we have defined $B^{\boldsymbol \xi}_{l,k-1}$, $l=0,\dots,n-k+1$. 
The $B$-splines of order $k$ are defined, for $l=0,\dots,n-k$, by
\begin{equation}
\label{eq:recspline}
B_{l,k}^{\boldsymbol \xi }(x)
=
 \frac{ x- {\xi_{l}}
  }{
  {\xi_{l+k}}-{\xi_{l}} 
  } 
 B_{l,k-1}^{\boldsymbol \xi}(x)+
  \frac{{\xi_{l+1+k}}-x}{
 {\xi_{l+1+k}}-{\xi_{l+1}} 
 } 
 B_{l+1,k-1}^{\boldsymbol \xi}(x).
\end{equation}

It is also important to notice that the above evaluations need to be performed only over the joint support of the splines involved in the recurrence relation.
The recurrent structure of the support is as follows. 
For zero order splines, the support of $B_{l,0}^{\boldsymbol \xi}$ is clearly $[\xi_l,\xi_{l+1}]$, $l=0,\dots, n$. 
If the supports of $B_{l,k-1}^{\boldsymbol \xi}$'s are $[\xi_l,\xi_{l+k}]$, $l=0,\dots,n-k-1$, then the support of $B_{l,k}^{\boldsymbol \xi}$ is the joint support of $B_{l,k-1}^{\boldsymbol \xi}$ and $B_{l+1,k-1}^{\boldsymbol \xi}$, which is $[\xi_l,\xi_{l+1+k}]$,  $l=0,\dots,n-k$.
In order to translate these recursive relations to the relations between the matrices of the derivatives at the knots, we need the following result on the derivatives of the $B$-splines. 
This result follows from \ref{eq:recspline} and  the graphical illustration of the recurrence is presented in Fig.~\ref{fig:RecSp}-{\it (Left)}. 
 
\begin{proposition}
\label{prop:dersp}
For $i= 0,\dots, k$ and $l=0,\dots,n-k+1$:
\begin{multline}
\label{eq:recder}
\frac{d^iB_{l,k}^{\boldsymbol \xi}}{dx^i}(x)
=
\frac{i}{\xi_{l+k}-\xi_l}\frac{d^{i-1}B_{l,k-1}^{\boldsymbol \xi}}{dx^{i-1}}(x)+\frac{ x- {\xi_{l}}
  }{
  {\xi_{l+k}}-{\xi_{l}} 
  } 
\frac{d^{i}B_{l,k-1}^{\boldsymbol \xi}}{dx^{i}}(x)
+\\
+
\frac{i}{\xi_{l+1}-\xi_{l+k+1}}\frac{d^{i-1}B_{l+1,k-1}^{\boldsymbol \xi}}{dx^{i-1}}(x)+\frac{ {\xi_{l+k+1}-x}
  }{
  {\xi_{l+k+1}}-{\xi_{l+1}} 
  } 
\frac{d^{i}B_{l+1,k-1}^{\boldsymbol \xi}}{dx^{i}}(x).
\end{multline}
The support of ${d^iB_{l,k}^{\boldsymbol \xi}}/{dx^i}$ is $[\xi_{l},\xi_{l+k+1}]$ and
if $i=k$, then $d^{i}B_{l+1,k-1}^{\boldsymbol \xi}/dx^{i}\equiv 0$.
\end{proposition}
\begin{proof}[Proof of \ref{prop:dersp}]
To see the above, we notice that \ref{eq:indi} coincides with \ref{eq:recder} in the case of $i=0$ (the undefined term ${d^{-1}B_{l,k-1}^{\boldsymbol \xi}}/{dx^{-1}}$ can be neglected since it is multiplied by $0$, so one can define it, for example, equal to zero). 
For the first derivative, i.e. $i=1$, we have
\begin{multline*}
\frac{dB_{l,k}^{\boldsymbol \xi}}{dx}(x)
=
\frac{1}{\xi_{l+k+1}-\xi_l}B_{l,k-1}^{\boldsymbol \xi}(x)+\frac{ x- {\xi_{l}}
  }{
  {\xi_{l+k+1}}-{\xi_{l}} 
  } 
\frac{dB_{l,k-1}^{\boldsymbol \xi}}{dx}(x)
+\\
+
\frac{1}{\xi_{l+1}-\xi_{l+k+1}}B_{l+1,k-1}^{\boldsymbol \xi}(x)+\frac{ {\xi_{l+k+1}-x}
  }{
  {\xi_{l+k+1}}-{\xi_{l+1}} 
  } 
\frac{dB_{l+1,k-1}^{\boldsymbol \xi}}{dx}(x).
\end{multline*}
We note that if $k=1$, then ${dB_{l+1,k-1}^{\boldsymbol \xi}}/{dx}\equiv 0$.
Then a simple induction argument leads to  \ref{eq:recder}. 
\end{proof}

Let consider the zero$^{\rm th}$ and first order $B$-splines with the boundary conditions  $B_{0,l}^{\boldsymbol \xi}$,  $B_{1,r}^{\boldsymbol \xi}$, where $l=0,\dots, n$, $r=0,\dots, n-1$. 
Then the corresponding $(n+2)\times 1$ and $(n+2)\times 2$ matrices are
\begin{equation*}
\mathbf S^{(0,l)}=\begin{bmatrix} 0 \\ \vdots \\ 0 \\ 1 \makebox[0in]{~~~\hspace{14mm}$\leftarrow l+1$} \\ 0 \\ \vdots\\  0
\end{bmatrix} \hspace{1.2cm} ,~~l\le n, \hspace{3mm} 
\mathbf S^{(1,r)}=\begin{bmatrix} 
0 & 0  \\
 \vdots & \vdots \\
  0 &  \frac{1}{\xi_{r+1}-\xi_{r}}\\ 
  1 & \frac{-1}{\xi_{r+2}-\xi_{r+1}} \makebox[0in]{~~~\hspace{28mm}$\leftarrow r+2$, $r < n$.} \\ 
  0  &  0\\
   \vdots & 
  \vdots \\  0 & 0
\end{bmatrix}
\end{equation*} 

Using the recurrent relation \ref{eq:recder} between the derivatives of $B$-splines, one can generalize the recurrence between matrix representation of the $B$-splines to the arbitrary order of splines. 
Using one sided representation of the splines, we have 
\begin{multline}
\label{eq:matrec}
\mathbf S_{\cdot j}^{(k,l)}=
\frac{1}{\xi_{l+k}-\xi_l}\left({j}\cdot \mathbf S_{\cdot j-1}^{(k-1,l)} - \boldsymbol \Lambda_l  \mathbf S_{\cdot j}^{(k-1,l)} \right)
+\\
+
\frac{1}{\xi_{l+1}-\xi_{l+k+1}}\left({j}\cdot \mathbf S_{\cdot j-1}^{(k-1,l+1)} - \boldsymbol \Lambda_{l+k+1}  \mathbf S_{\cdot j}^{(k-1,l+1)} \right),
\end{multline} 
where $l=0,\dots, n-k$, $j=0,\dots , k$ and the diagonal $(n+1)\times (n+1)$ matrices $\boldsymbol \Lambda_l$'s have $(\xi_0-\xi_l , \dots , \xi_n-\xi_l)$ on the diagonal.
Here we assume that if $j=k$, then $\mathbf S_{\cdot j}^{(k-1,l)}$ is a column made of zeros as the $k^{\rm th}$ derivatives of the $(k-1)^{\rm th}$ order spline is always zero. 

The above algebraic relation is implemented in the package.
To construct the $B$-spline basis of order {\tt k}, one uses 
\begin{verbatim}
so = splinet(xi, k); Bsplines=so$bs,
\end{verbatim}
where {\tt splinet()} generates three different types of basis and organizing them as a list. 
The element in the list labeled $\tt bs$ is always the {\tt Splinets}-object corresponding to the $B$-spline basis.

The string-flag {\tt  type} in a {\tt Splinets}-object indicates if the object is a basis, with {\tt  type="bs"} indicating that it is a $B$-splines basis, i.e. in the above example {\tt Bsplines@type="bs"}. 
An example of the result is shown in Fig.~\ref{fig:RecSp}~{\it (Right-top)}. 
There the equidistant case is presented and it should be noted that the algorithm used accounts for the additional efficiency that this case yields in the computations. 


\begin{figure}[t!]
\begin{center}
  \includegraphics[width=0.8\textwidth]{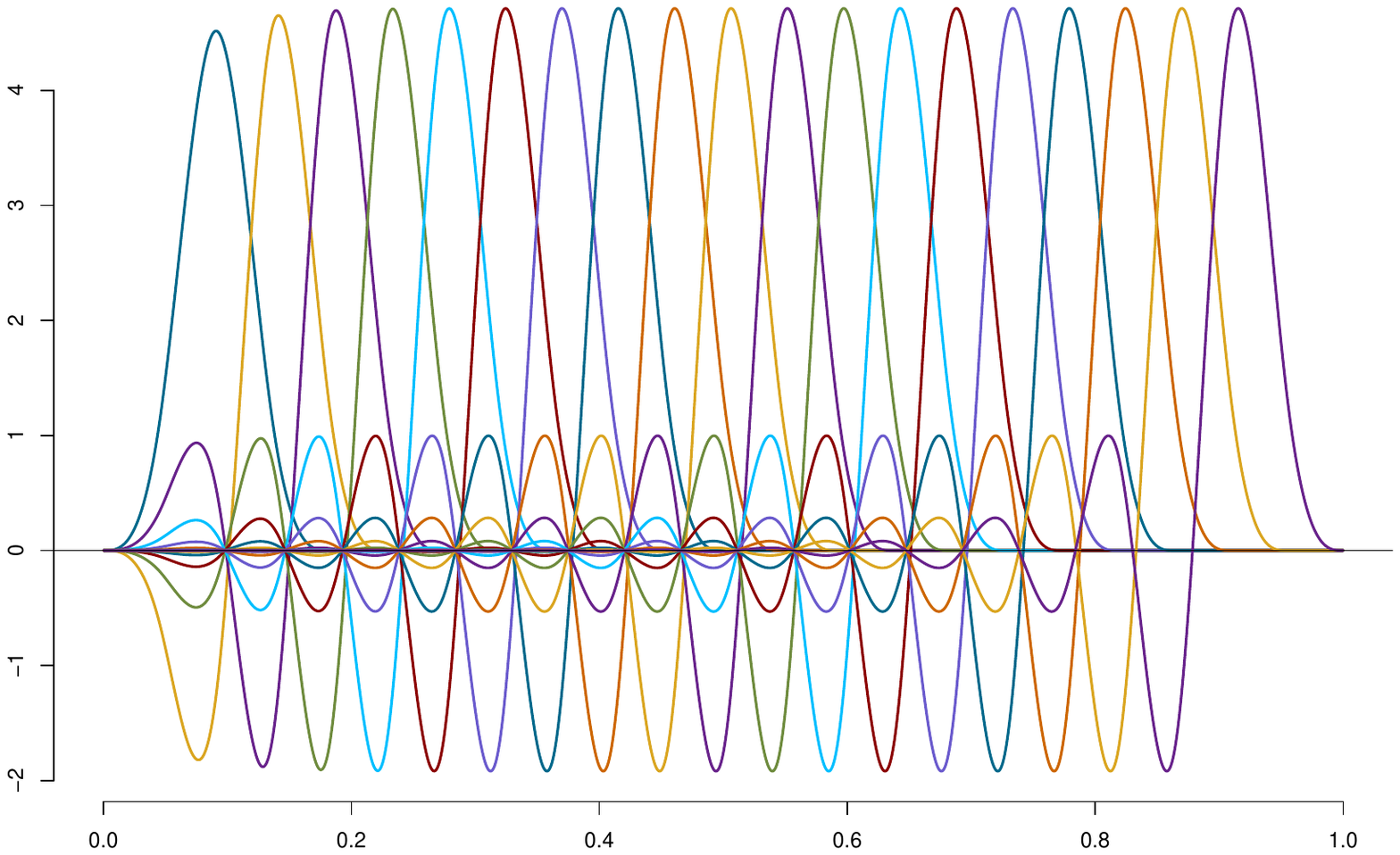}\vspace{-8mm}
  
      \includegraphics[width=0.8\textwidth]{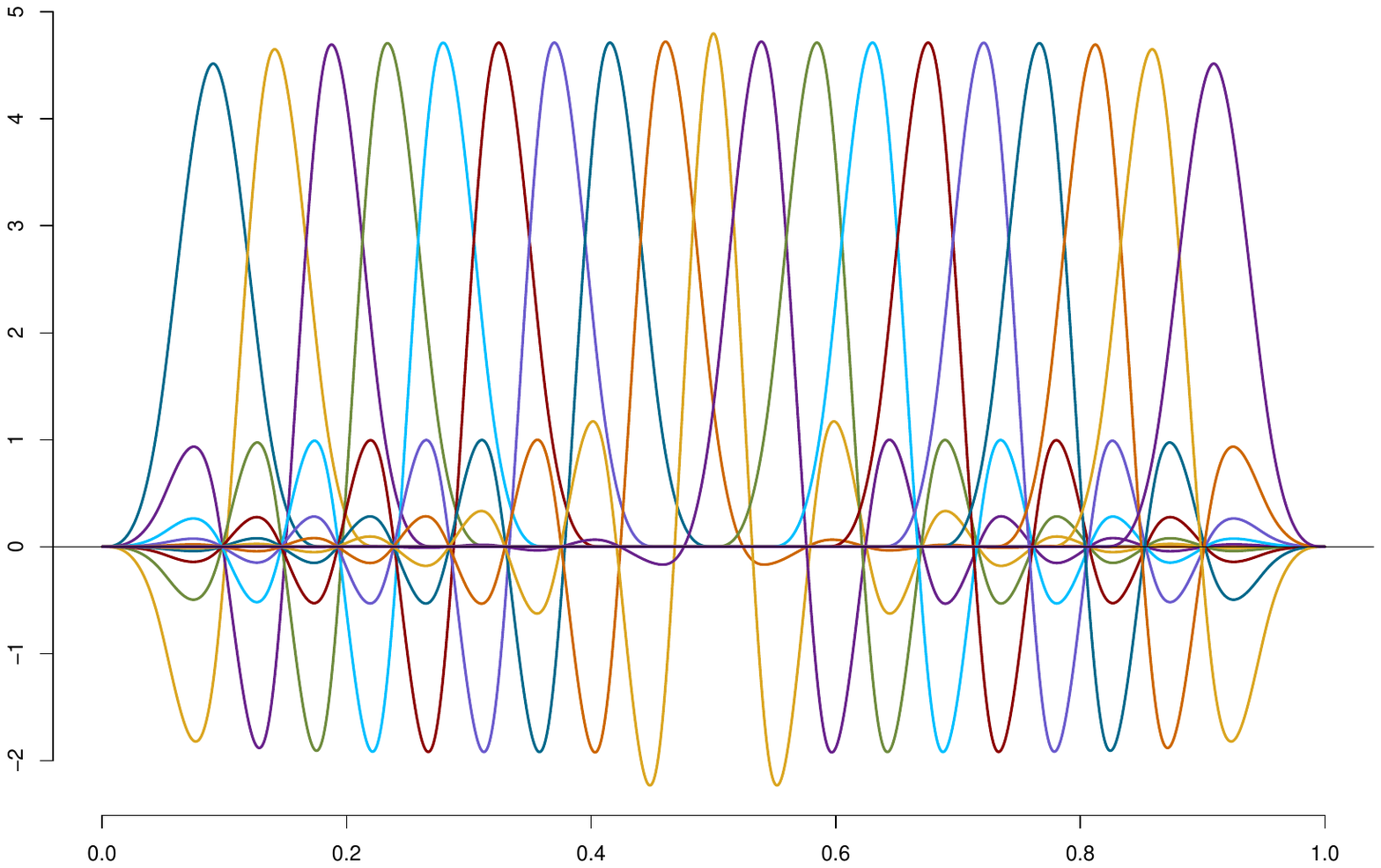}\vspace{-8mm}
  
\includegraphics[width=0.8\textwidth]{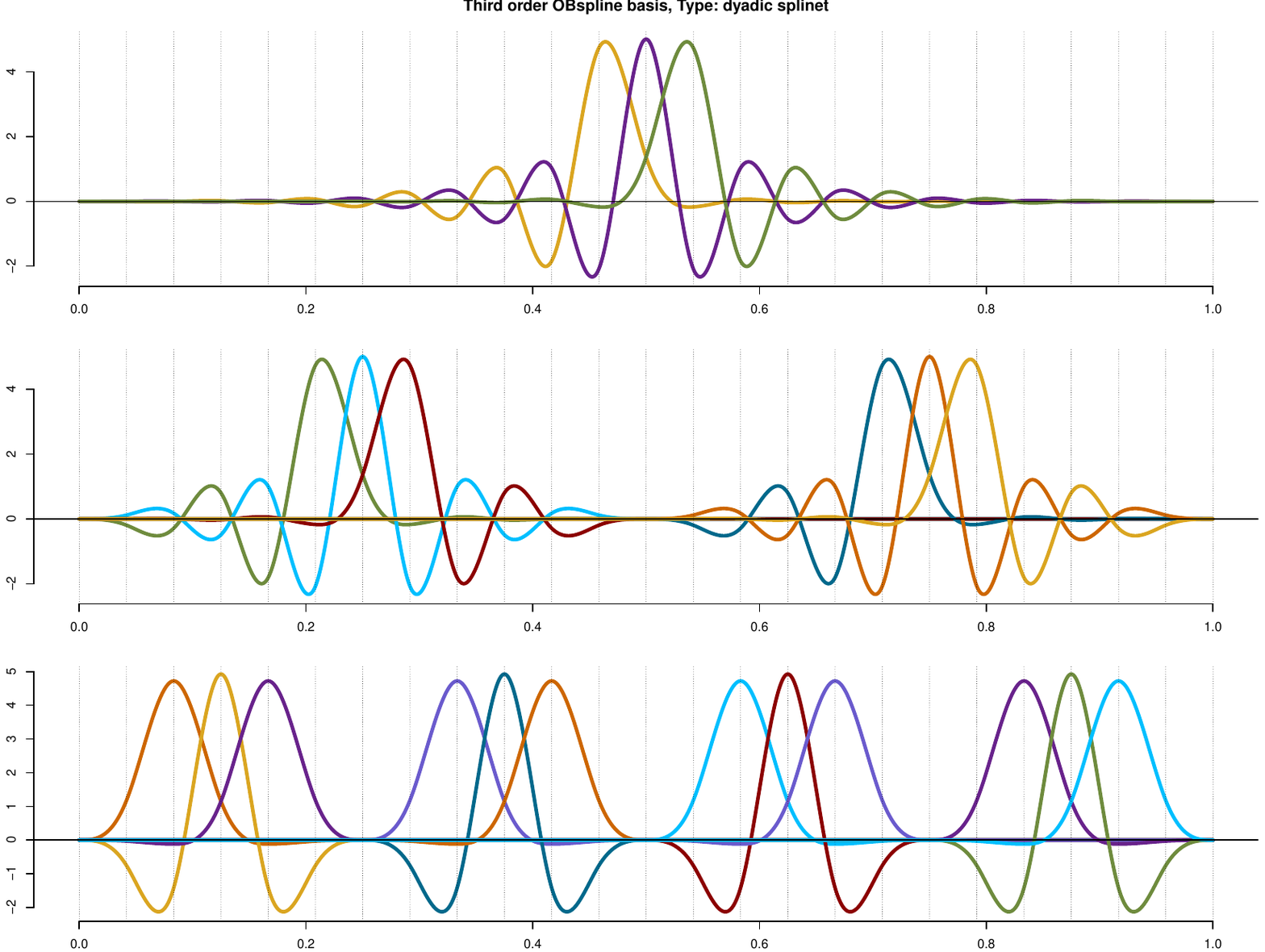} 
\end{center}   
    \caption{{\it Top-Middle:}\, The one-sided  {\it (top)} and symmetrized Gram-Schmidt orthogonalization {\it (middle)} for the $B$-splines from Fig.~\ref{fig:RecSp}; 
  {\it Bottom:}\, The splinet on a dyadic net.}
  \label{fig:Orth}
\end{figure}

 \subsection{Orthogonal spline bases}
The orthonormalized bases implemented in this package are obtained by one of the following three orthogonalization procedures applied to $B$-splines. 
The first one is simply the Gram-Schmidt orthogonalization performed on the $B$-splines ordered by their locations, the second one is a symmetric (with respect to the knot locations) version of the Gram-Schmidt, and, finally, the dyadic orthogonalization into a {\em splinet} which is our preferred method. All the methods have been discussed in detail in \cite{Liu2019SplinetsE}. 
In the object representation of collections of splines, i.e. in the {\tt Splinets}-class, the field {\tt type} specifies which of the orthonormal basis one deals with. 
The function  {\tt splinet()} is generating the proper basis with the default form 
\begin{verbatim}
so=splinet(xi); Bsplines=so$bs; Splinet=so$spnt
\end{verbatim}
and returning a list of two {\tt Splinets} objects, {\tt so\$bs} and {\tt so\$spnt} build over the ordered knots {\tt xi}.
The first object represents the basis made of the standard cubic $B$-splines and thus not orthogonal. 
The second one represents the recommended orthonormal basis referred to as a cubic splinet. 
This is also represented in the field {\tt type} which can be either {\tt dspnt} or {\tt spnt}, i.e. depending if the system is fully dyadic or not, respectively. 
The explanation of the dyadic structure of a splinet is given below. 
In Fig.~\ref{fig:RecSp}~{\it (Left-bottom)}, the splinet obtained from the equally spaced $B$-splines is presented. 

It is important to point out that the main computational engines of the orthogonalization processes are implemented to work on a generic gramian matrix $\mathbf H$ and returning $\mathbf P$ such that 
\begin{equation}
\label{eq:diag}
\mathbf I = \mathbf P^\top \mathbf H \mathbf P.
\end{equation}
These algorithms do not explicitly reference splines and can be applied in other contexts in which such diagonalization may deem important. 
They are available inside the package as the auxiliary functions but there are not explicitly referenced in the documentation. \vspace{1mm}

\noindent{\it One-sided orthogonalization.}
The function {\tt splinet} can also generate other than the splinet types of orthogonal spline bases.
One of them is  the one-sided $OS$-splines, {\tt type="gsob"},  which are obtained by the classical Gram-Schmidt orthogonalization. 
If the $B$-splines have been obtained earlier in {\tt so\$bs}, then the following command will produce the desired outcome
\begin{verbatim}
so=splinet(xi, Bsplines=so$bs, type='gsob'); GramSchmidt=so$os
\end{verbatim}

The one-sided orthogonalization has two flaws, firstly, it has a fairly large the total support set and, secondly, it is asymmetric even for the equally spaced case, as seen in Fig.~\ref{fig:Orth}~{\it (Top)}. 
The matrix $\mathbf P$ in (\ref{eq:diag}) has the triangular form and thus 50\% of the entries are non-zero.  

\vspace{1mm}

\noindent{\it Symmetric approach.}
The symmetric orthogonalization, {\tt type="twob"},  alleviates some of the flaws of the one-sided case. 
It utilizes a symmetrized version of the Gram-Schmidt orthogonalization and originally was proposed in \cite{Redd}. 
The two-sided orthogonalization addresses both the flaws of the one-sided basis: it has a smaller total support set, and it is symmetric with respect to the center, as seen in Fig.~\ref{fig:Orth}~{\it (Middle)}. 
The matrix $\mathbf P$ in (\ref{eq:diag}) has two equally sized rectangular blocks and thus its more sparse than the original by having only 25\% of non-zero terms.  
\vspace{1mm}

While the symmetrized orthogonalizaton can be viewed as an improvement over the one-sided orthogonalization, one can do even better as far as the total support size and the sparsity of the matrix $\mathbf P$ are concerned. 
Another orthogonalization procedure which utilizes the previous two approaches in a  `telescopic/dyadic' manner is far more optimal. 
This is the orthonormalization of  the $B$-splines that is promoted in the package and in this work. 
The next subsection is devoted to description of the functionality of the splinets in their {\tt Splinets}-implementation.

\subsection{ Splinets -- optimal orthogonalization.}
The dyadic algorithm for orthogonalization of the $B$-splines produces $OB$-splines that are residing only over a slightly bigger support (factor of $\log n$, where $n$ is the number of knots) than the total support size of the original $B$-splines. 
The obtained splines are symmetric with respect to the center but extend the symmetry into a dyadic structure based on the support sets. 
Finally, the sparsity of the matrix $P$ is dramatically improved over the two previous orthogonalization methods as seen in Fig.~\ref{fig:Sparsity}~{\it (Right)}, where the sel-similar fractal structure of the non-zero matrix entries  is presented. 
For a given set of knots in {\tt xi}, the splinet is evaluated simply by {\tt so = splinet(xi)} which returns the  {\tt Splinets}-object  {\tt so\$os} containing the splinet. 
The graphical presentation of the splinet on the dyadic grid as shown in Fig.~\ref{fig:Orth}~{\it Bottom} can be simply obtained by {\tt plot(so\$os)}. 
If one prefers to have a graph without dyadic structure it can be obtained by {\tt plot(so\$os,type="simple")} as seen in Fig.~\ref{fig:RecSp}~{\it (Right-Bottom)}.

The algorithm is most naturally described for the dyadic structure of knots that assumes that for some positive integer $N$ we have $k2^N-1 = n$, where $k$ is the order of the splines and $n$ is the number of the internal knots, i.e. we do not count the endpoints.
The resulting $OB$-splines are then located on a dyadic net with the $N$ levels featuring increasing support set sizes and  a decreasing number of basis element.  
The $OB$-splines are also grouped into $k$-tuples of the neighboring splines. 
All these features are best seen in Fig.~\ref{fig:Orth}, {\it Bottom}. 
We observe that each of the $OB$-spline in the splinet inherits its location from the corresponding $B$-spline.
The location is naturally represented by the middle knot in their support if the number of the knots in the support is odd, or by the average of the two middle knots if that number is even. 
This locations can be used to present any spline basis on a dyadic-net graph. 
In fact, the algorithms are implemented in such a way that any set of knots leads to a splinet that can be represented on a dyadic net which may be incomplete if the number of knots does not satisfy the dyadic case restriction. 
 Fig.~\ref{fig:Sparsity}~{\it (Center/Bottom-Left)} shows an example of such situation. 
 
As mentioned, all fundamental algorithms carrying computational burden of orthonormalization are operating on matrices.
Their computational efficiency can be improved if these generic algorithms are implemented in a lower-level language, such as the C-language, and compiled to the machine code.
This approach has not been yet taken but it will be considered in future versions of the package. 
 
 The special case of equally spaced knots correspond to  $\mathbf H$ being a Toeplitz matrix.
 In this case certain parts of the algorithms can be significantly accelerated  since the evaluation has to be performed only for one $k$-tuple instead of $2^{N-l}-1$ of them at each support level $l$ in the dyadic net. 
 In the package, this has been implemented in the fully dyadic equidistance knots case but not for the non-dyadic case.
  To utilize this efficiency it may be advantageous for a large number of equally spaced knots to choose them so that their number satisfies the fully dyadic condition because the gained efficiency may compensate additional computations resulting from the larger number of  knots. 
 Alternatively, one can decompose a non-dyadic structure into smaller dyadic structures and perform the orthogonalization on them. 
 The latter approach is illustrated in the next section to demonstrate functionality of the {\tt Splinets}-package. 

\begin{figure}[t!]
 \parbox[b]{0.5\textwidth}{
  \includegraphics[width=0.45\textwidth]{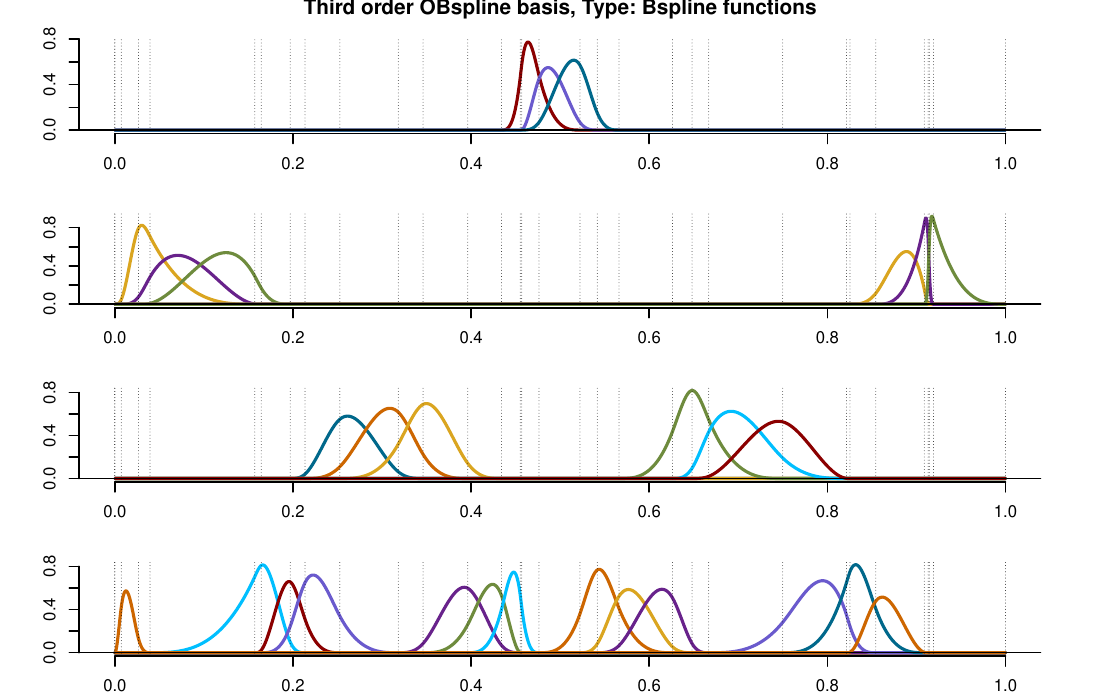} 
 \includegraphics[width=0.45\textwidth]{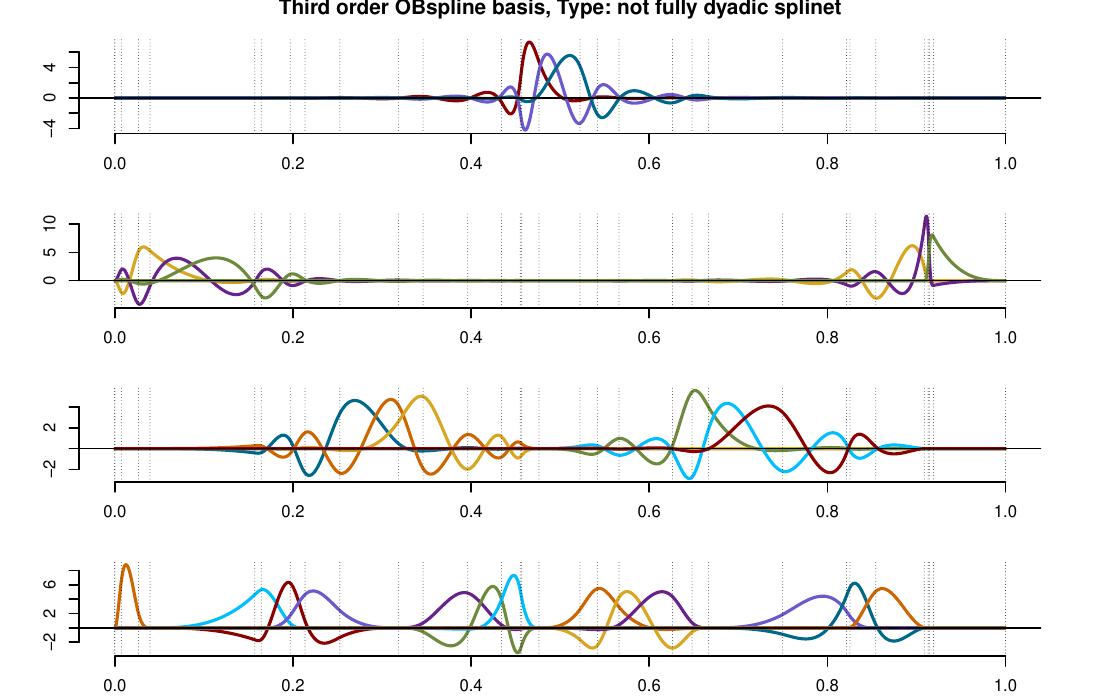} 
 }
\hspace{-7mm}  \includegraphics[width=0.55\textwidth]{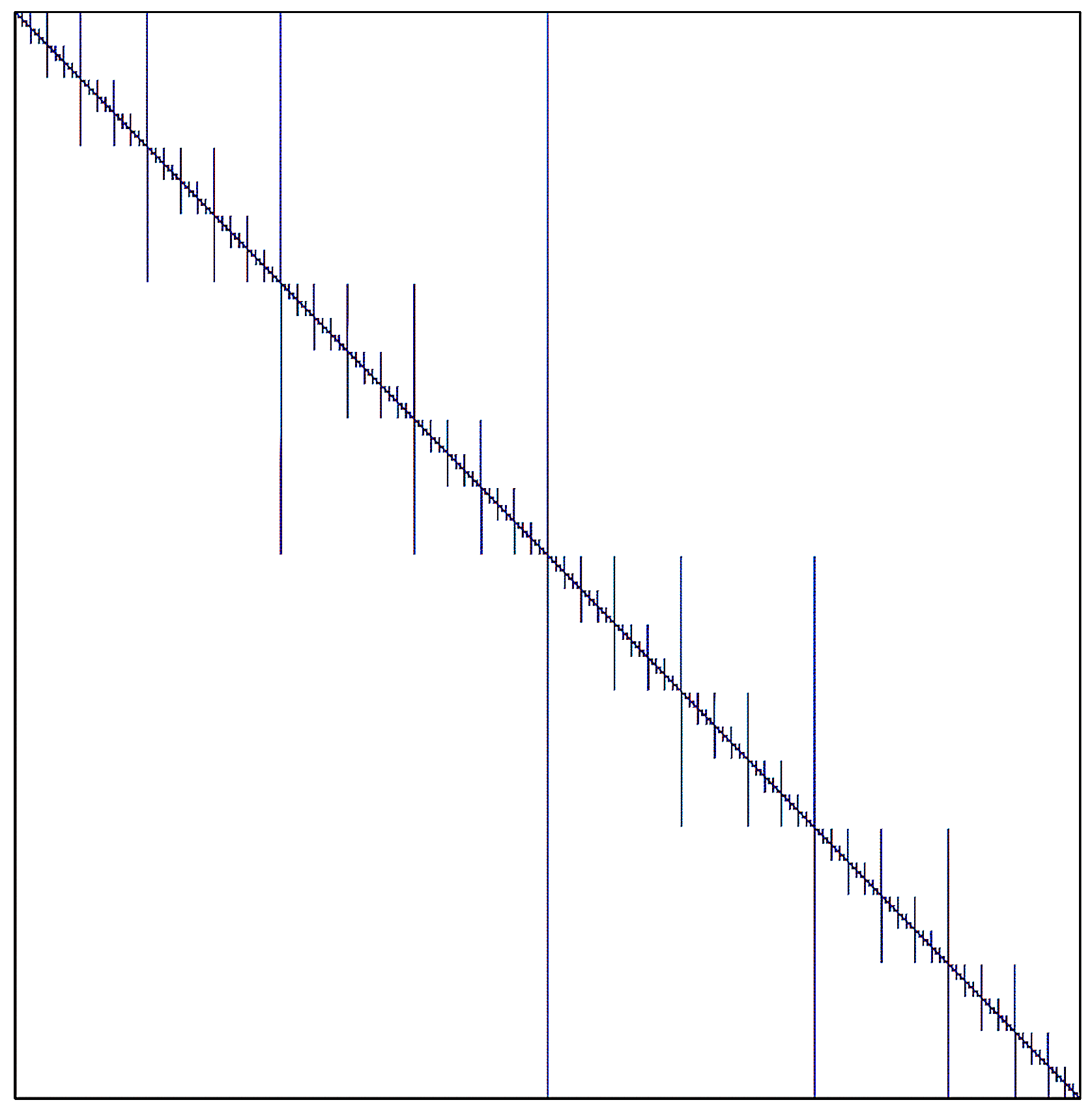} 
     
    \caption{ 
    {\it Left:}\, A not-fully dyadic and not-equaly spaced case. Both the $B$-splines (top) and the  splinet (bottom) are presented on a dyadic net. 
    {\it Right:}\, A  self-similar structure of the non-zero entries of $\mathbf P$ corresponding to the orthonormalization of the $B$-splines leading to a splinet. 
    The case of the third order $B$-splines with equally spaced knots leading to  a 1533 dimensional space.}
  \label{fig:Sparsity}
\end{figure}

\section{Orthogonal projection to a space of splines}
The spaces of splines are finite dimensional spaces of the Hilbert space of all square integrable functions and thus any such a function can be projected in the orthogonal fashion into the linear space of splines spanned over a particular set of knots. 
Any functional data analysis typically begins with projecting the data to a functional finite dimensional subspace -- the projection becomes a fundamental operation for carrying out statistical data analysis. 
One can also perform a projection to smooth the data and there a plethora of methods to target this goal. 
In the package, we have implemented orthogonal projection in the function {\tt project()}. 
Since actual functional data can be represented in a variety of ways, the projection itself depends on the input format and, more specifically, the way the inner product of the input with a spline is defined.  
We consider two type of inputs: {\tt Splinets}-objects and columns of pairs representing arguments and values of a discretization of functional data.

Independently of the input, the output of the function {\tt project()} is a list, say {\tt onsp}, made of the three components:
\begin{description}
\item{\tt onsp\$coeff} -- the matrix of coefficients of the decomposition in the selected basis,
\item{\tt onsp\$basis} -- the {\tt Splinets}-object representing the selected basis,
\item{\tt onsp\$sp} -- the {\tt Splinets}-object representing the projection of the input in the projection spline space.

\end{description}
Additional information, such as the knots, the order, the type of the basis, can be retrieved easily from the second component. 
Many of the algebraic operations on the splines are more conveniently performed on the matrix of coefficients of their spline basis representations, than on the {\tt Splinets}-objects themselves. 
The coefficient matrix {\tt onsp\$coeff} can be utilized for such computations and the corresponding linear combination of {\tt onsp\$basis} can be used whenever the functional form of the result is needed.

\noindent{\it Basis decomposition.} The simplest projection obtained through {\tt project()} is not, strictly speaking, a projection but rather decomposition of  a {\tt Splinets}-object to coefficients in the given basis. 
If {\tt sp} is a {\tt Splinets}-object, then 
\begin{verbatim}
bdsp=project(sp); bdsp2=project(sp,type='bs'); 
bdsp3=project(sp,type='gsob'); bdsp4=project(sp,type='twob'); 
\end{verbatim}
have as its main output the matrices of coefficients $a_{ji}$,  such that the $j^{\rm th}$ input-spline {\tt sp} has the form
$$
\sum_{i=1}^{n-k+1} a_{ji} OB_{i},
$$
where $j$ indexes the input splines, $n$ is the number of the internal knots, $k$ is the smoothness order, and $OB_{i}$ is the selected basis of splines controlled by the input {\tt type} (the default is the splinet built on the same knots as the input spline). 
The possible choices of the bases are the splinet, the one- and two-sided orthonormal bases, or the $B$-splines, all  built on the same knots as the input spline.

\noindent{\it Projecting splines.} 
The projection of {\tt Splinets}-objects over a given set of knots to the space of splines over a different set of knots is obtained as the orthogonal projection.
Namely, if $S$ is  the input {\tt Splinets}-object then the output is denoted as $\mathbf P S$ where $\mathbf P$ is the orthogonal projection to the space spanned by the spline space build by the second set of knots
\begin{equation}\label{projection}
\mathbf P S =\sum_{i=1}^{n-k+1} a_{ji} OB_{i}, \,\, (S-\mathbf P S) \bot \,\mathbf P S.
\end{equation}
This is an extension of the previous case since the output functions may belong to a different space than the input functions.  
The result is obtained by embedding both the input splines and the projection space to the space of splines that contains both and evaluating the inner products between functions in this space. 
The space of splines that contains both is build over the union of the two set of knots and uses the package function {\tt refine()} for that purpose.
 The following code will lead to the result if {\tt knots} are different from {\tt sp@knots}
\begin{verbatim}
bdsp=project(sp,knots); bdsp2=project(sp,knots,type='bs'); 
bdsp3=project(sp,knots,type='gsob'); bdsp4=project(sp,knots,type='twob'); 
\end{verbatim}
The results are represented in the spline space build over {\tt knots} and thus the {\tt Splinets}-object in the output representing the projection spline satisfies {\tt bdsp\$bs@knots=knots}. 

\noindent{\it Projecting discretized functional data.}
The function {\tt project()} works also when the input is a discretization of some continuous argument functional data.
In this case, the input is a matrix  having in the first column a set of arguments and in the remaining ones the corresponding values of a sample of functional data. 
The input data are considered to be a piecewise constant functions with the value over two subsequent arguments equal to the value in the input corresponding
to the left-hand-side argument.
In this way, the discretized data can be viewed as functions and their inner products with any spline are well defined. 
In the package, this specific inner product can be obtained by utilizing the indefinite integral implemented in the function {\tt integra()}. 
Consequently, the projection $\mathbf P S$ of the functional data in $S$ satisfies (\ref{projection}), if $S$ represents the data as a  piece-wise constant function . 
In Figure~\ref{fig:Projectiion}, we illustrate the output from {\tt project()} for a sample of ten functional data.

\begin{figure}[t!]
\begin{center}
\raisebox{-4mm}{  \includegraphics[width=0.4\textwidth]{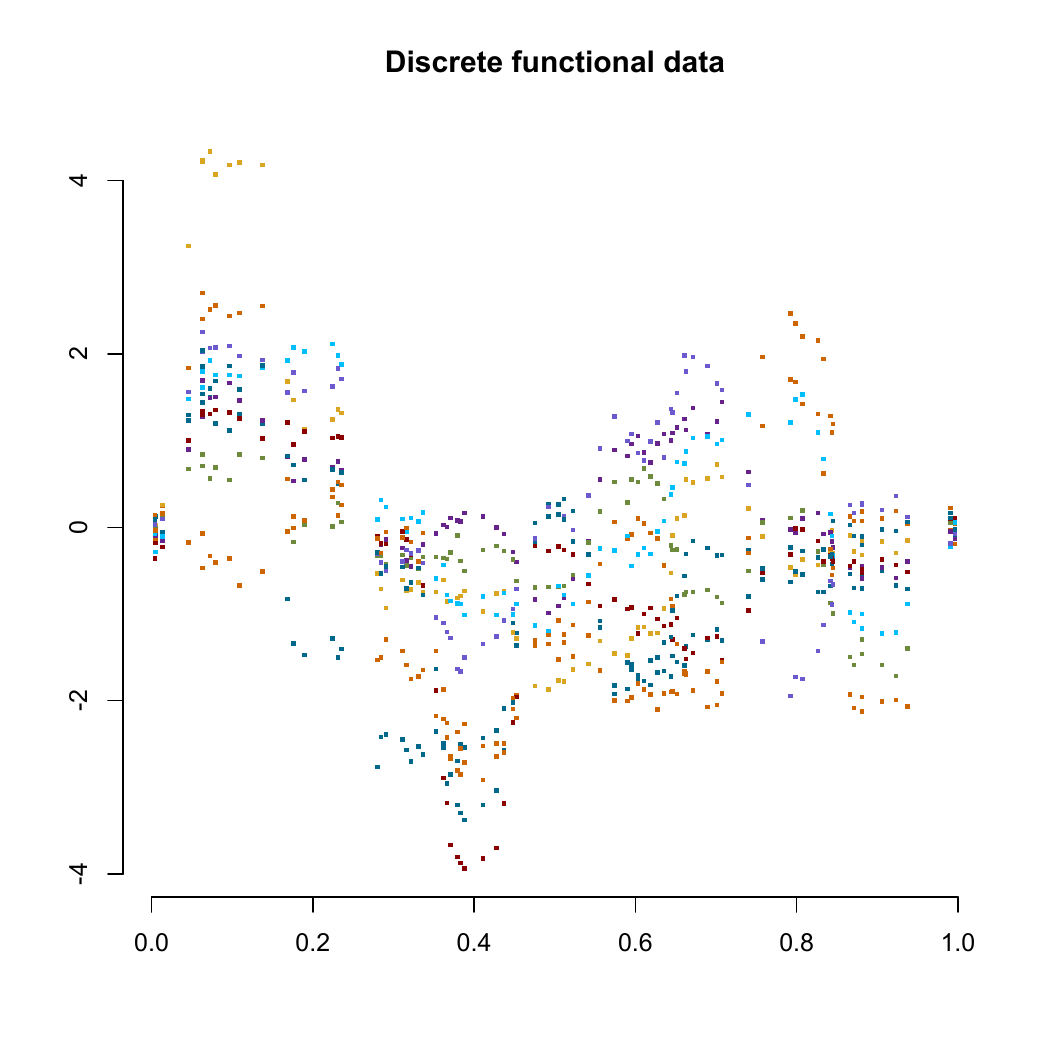} }
 \includegraphics[width=0.55\textwidth,height=0.35\textwidth]{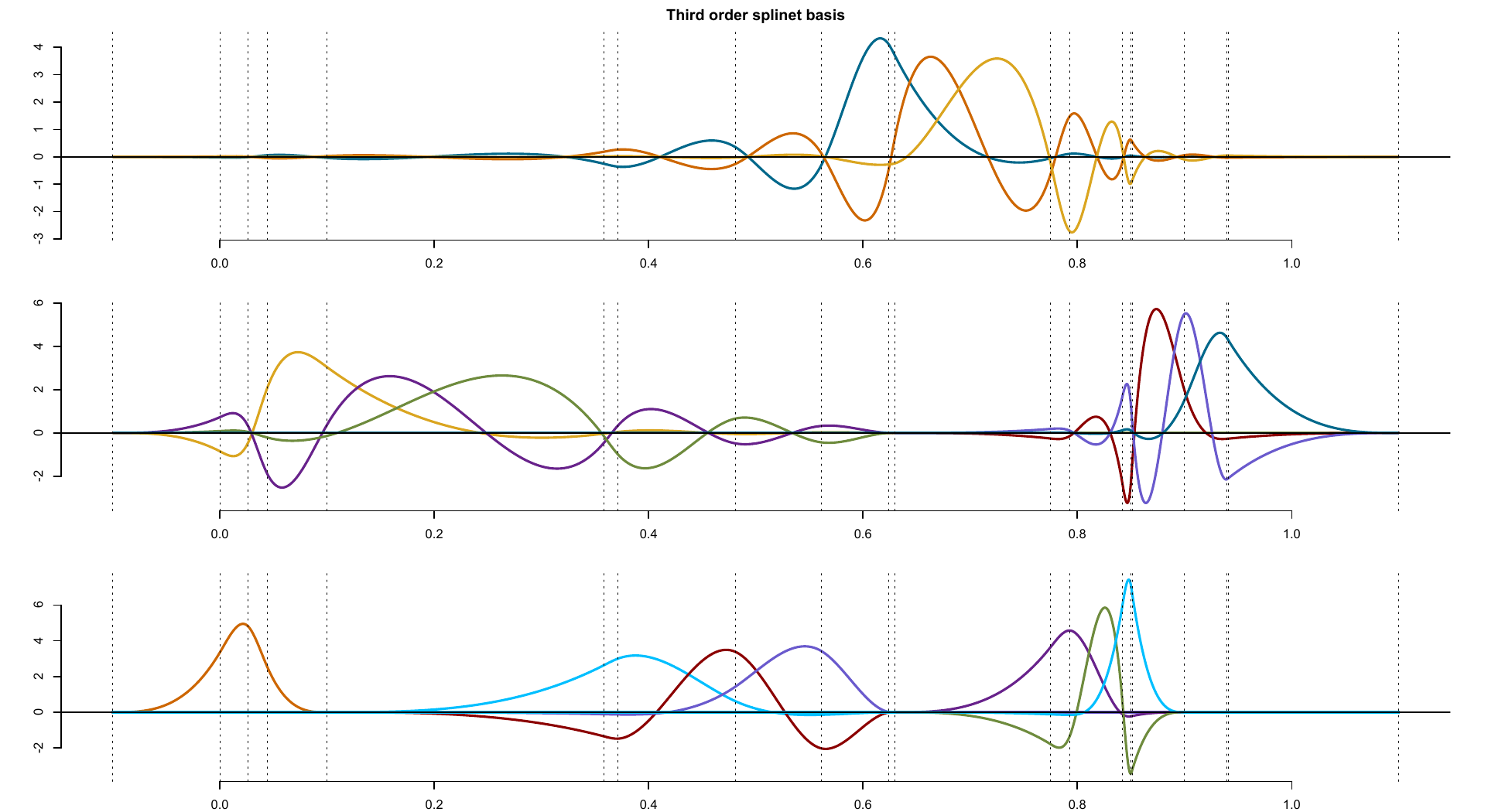} \\
\raisebox{-4mm}{     \includegraphics[width=0.4\textwidth]{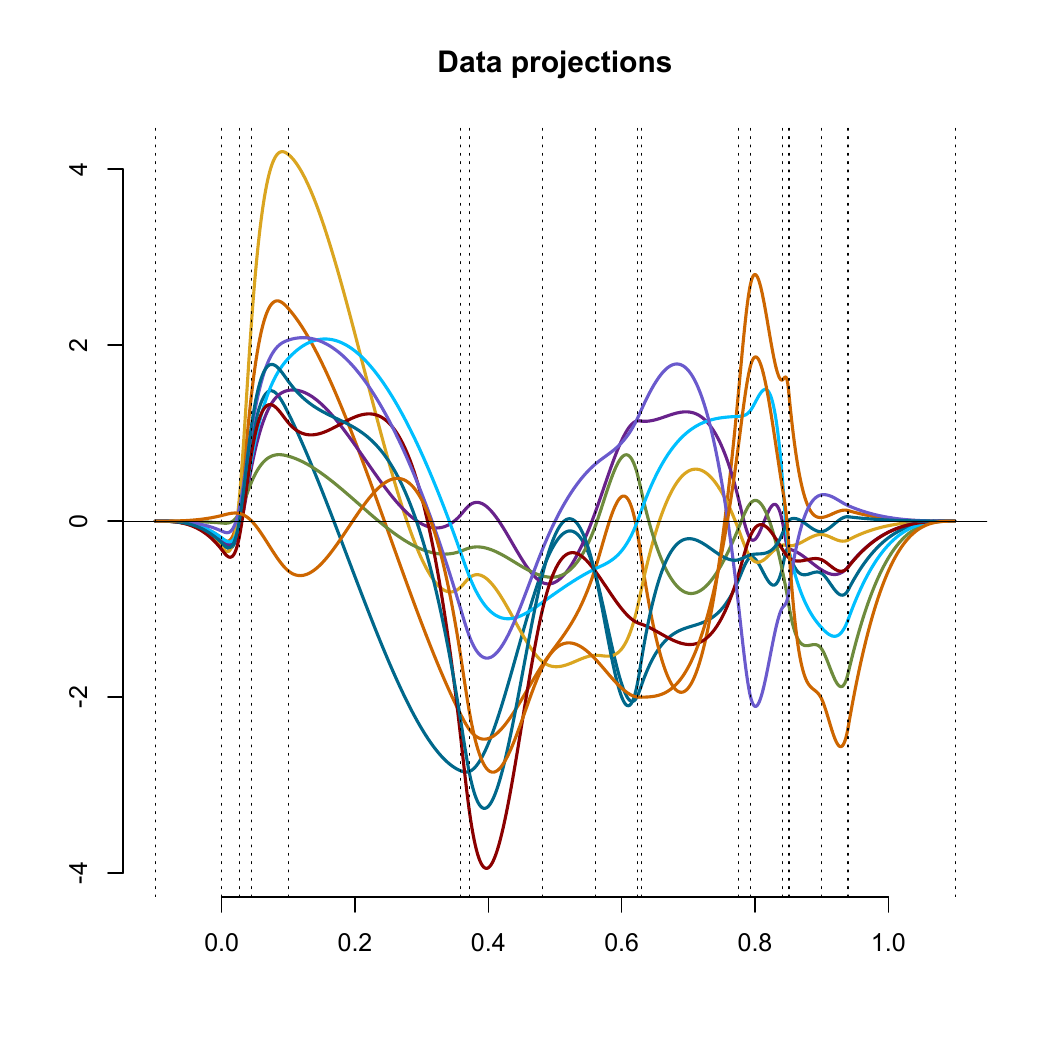} }
      \includegraphics[width=0.55\textwidth,height=0.35\textwidth]{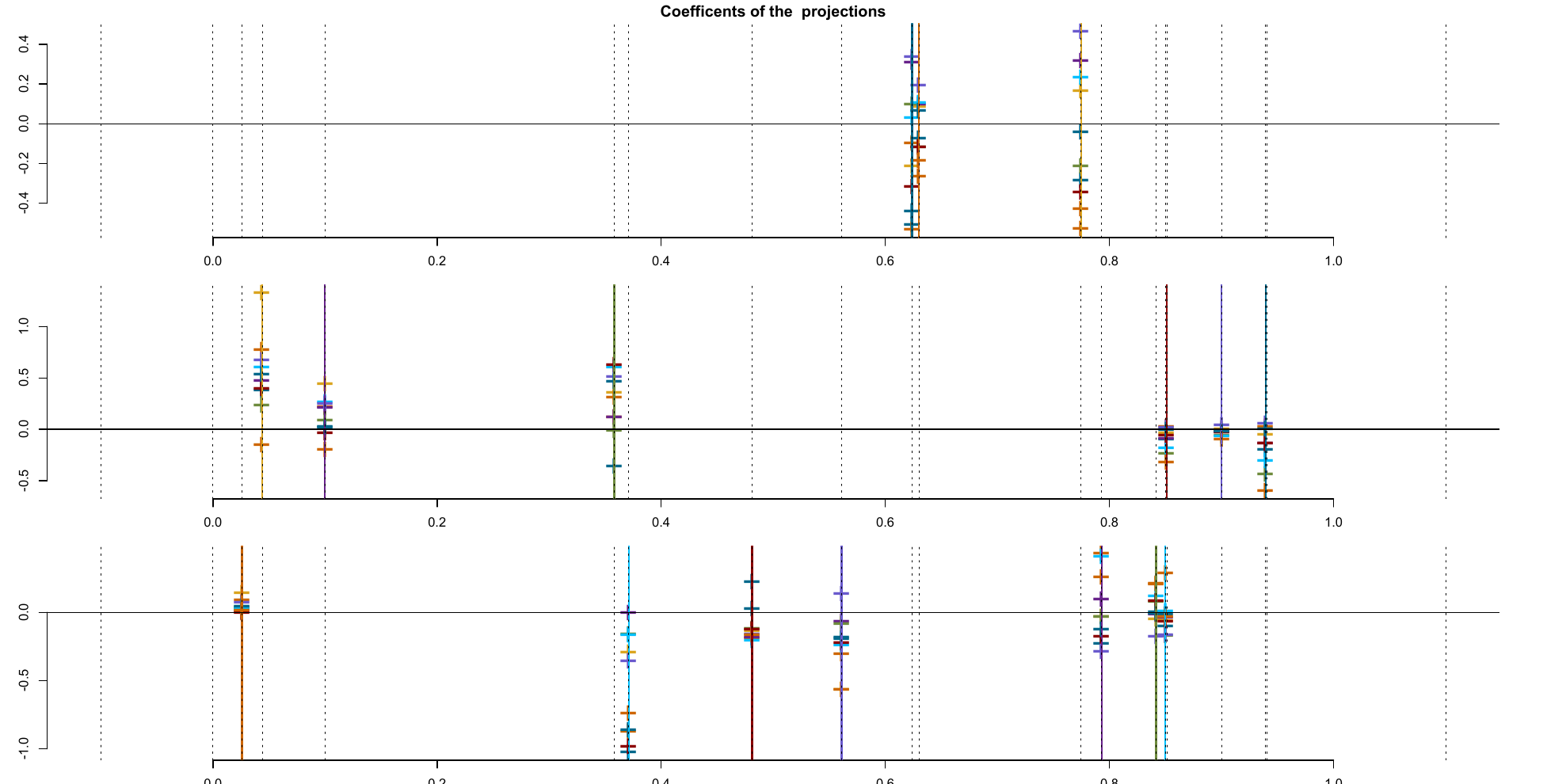} \vspace{-0.75cm}

 \end{center} 
    \caption{ 
    {\it Left-Top:}\, The original data $S_j$ $j=1,\dots, 10$. {\it Right-Top:}\,  The splinet $\{OB_i\}_{i=1}^{n-k+1}$ of the space to which the data are projected.  
    {\it Left-Bottom}\, The projection splines $\mathbf P S_j$. 
    {\it Right-Bottom}\, The projection coefficients $(a_{ji})_{j,i=1}^{10,n-k+1}$ for the data corresponding to the basis elements in the splinet. 
    The vertical lines represent the knot locations in the projection space. 
    }
    \label{fig:Projectiion}
\end{figure}

\section{Example -- an alternative orthonormal basis in the non-dyadic case}
When the number of internal and equally spaced knots satisfies the relation  $n=k2^N-1$ for some natural $N$, the splinet build upon them is computed faster than for irregularly spaced knots.
The reason is that for each support level in the splinet the matrices of the derivatives for the $k$-tuples are identical and need to be computed only once per level, compare Fig.~\ref{fig:Orth}~{\it Bottom} and Fig.~\ref{fig:Sparsity}~{\it Left-Bottom}.
However, for the non-dyadic case the symmetries break and the efficiency is not implemented in the current algorithm since it is difficult to track where these efficiencies occur. 
One can use an alternative approach to build bases in the non-dyadic case by extracting the maximal dyadic sub-splinets and orthogonalizing the remaining $B$-splines with respect to them. 
In our first example of utilizing the package, we illustrate how such an orthogonalization can be performed on a non-dyadic $B$-splines of  the first order ($k=1$).
We consider the equally spaced knots case although the identical approach applies for arbitrarily spaced knots, except for these the above mentioned efficiency in building the dyadic sub-splinets does not take place.

We take $n=43$ equally spaced internal knots, which means that this is a not fully dyadic case: $2^5-1<43<2^6-1$.
The following code builds the normalized $B$-splines (and also the splinet implemented in the package) 
\begin{verbatim}
k=1; n = 43; xi = seq(0, 1, length.out = n+2)
spts=splinet(xi,order = k, norm = T); Bsp=spts$bs; OBsp0=spts$os
\end{verbatim}
The largest dyadic sub-splinet that can be build from a subset of subsequent knots is for $n_1=2^5-1=31$ internal knots. 
Removing these knots plus the left-hand-side (LHS) end from the original knots leaves us with $43-32=11$ internal knots. 
We can repeat the process on the remaining knots with the dyadic sub-splinet build over $n_2=2^3-1=7$ internal knots leaving us with $11-7-1=3$ internal knots.
Over these final three knots we can built the dyadic splinet since $n_3= 2^2 - 1=3$.
This three splinets are adjacent and cover the whole range of the knots.
We group  all 45 knots as follows to build the dyadic splinet over each group
$$
(1:33),\,(33:42),\, (42:45).
$$
This is done in the following code in which the matrices of derivatives are only evaluated for the largest sub-splinet and the remaining 
subsplinets inherit the respective matrices.

\begin{lstlisting}
n1=31; n2=7; n3=3 
Bsp1=subsample(Bsp,1:n1); Bsp1@knots=xi[1:(n1+2)]; Bsp1@type='bs'  

sp1=splinet(Bsplines = Bsp1) 

OBsp=Bsp

OBsp@der[1:n1]=sp1$os@der; OBsp@supp[1:n1]=sp1$os@supp; OBsp@type='sp'

OBsp@der[(n1+2):(n1+n2+1)]=OBsp@der[1:n2] 
for(j in 1:n2 ) OBsp@supp[[(n1+1+j)]]=n1+1+OBsp@supp[[j]]

OBsp@der[(n1+n2+3):(n1+n2+n3+2)]=OBsp@der[1:n3]
for(j in 1:n3) OBsp@supp[[(n1+n2+2+j)]]=n1+n2+2+OBsp@supp[[j]]
\end{lstlisting}

This leaves us with only two $B$-splines labeled by 32 and 40, each overlapping two adjacent knot groups, that have not been yet accounted in the orthogonalization.
They are orgthogonalized with respect to the dyadic subsplinets (from the right to the left). 
We present the code to demonstrate the features of the package.
\begin{lstlisting}[firstnumber=15]
nt=seq2net(n1,k);  nt2=seq2net(n2,k); nt3=seq2net(n3,k) 

for(i in 1:length(nt2))
for(j in 1:length(nt2[[i]]))nt2[[i]][[j]]=nt2[[i]][[j]]+n1+1

for(i in 1:length(nt3))
for(j in 1:length(nt3[[i]]))nt3[[i]][[j]]=nt3[[i]][[j]]+n1+n2+2

S1=subsample(OBsp,n1+1); S2=subsample(OBsp,n1+n2+2)

S1L=subsample(OBsp,nt[[1]][[1]]); S1R=subsample(OBsp,nt2[[length(nt2)]][[1]]) 
for(i in 2:length(nt)) S1L=gather(S1L,subsample(OBsp,nt[[i]][[length(nt[[i]])]]))
for(i in (length(nt2)-1):1) S1R=gather(S1R,subsample(OBsp,nt2[[i]][[1]]))

S2R=subsample(OBsp,c(nt3[[2]][[1]],nt3[[1]][[1]]))
S2R=subsample(OBsp,c(nt3[[2]][[1]],nt3[[1]][[1]]))
indL=vector(); for(i in 1:length(nt2))indL=c(indL,nt2[[i]][[length(nt2[[i]])]])
S2L=subsample(OBsp,indL)

S2LR=gather(S2L,S2R)

A=matrix(c(1,-InnPr[3:5],-InnPr[5:4]),ncol=6)
sc=1/sqrt(1-A[1,2:dim(A)[2]]%*%t(A[1,2:dim(A)[2]])); A=sc[1,1]*A

OS2=lincomb(gather(S2,S2LR),A)
OBsp@der[n1+n2+2]=OS2@der; OBsp@supp[n1+n2+2]=OS2@supp

S1LR=gather(S1L,S1R); S1LR=gather(S1LR,OS2)

InP=gramian(S1,Sp2 = OS2)

A=matrix(c(1,-InnPr,-InnPr[5:3],InP),ncol=10)
sc=1/sqrt(1-A[1,2:dim(A)[2]]%*%t(A[1,2:dim(A)[2]])); A=sc[1,1]*A

OS1=lincomb(gather(S1,S1LR),A)
OBsp@der[n1+1]=OS1@der; OBsp@supp[n1+1]=OS1@supp

\end{lstlisting}

We note that the support level for each of these two splines is one above the largest level of the dyadic subsplinets that it overlaps as shown in Figure~\ref{fig:Non-dyadic}~{\it (Left-Top)}. 
The above code yields two different orthogonalization of the $B$-splines, the implemented in the package given in {\tt OBsp0} and the alternative one given in {\tt OBsp}. 
The two bases differ as seen in Figure~\ref{fig:Non-dyadic}.
The presented here alternative way of orthogonalizing the $B$-splines in a non-fully dyadic case will be considered in the full generality in the future releases of the package.

\begin{figure}[t!]

    \includegraphics[width=0.495\textwidth, height=0.42\textwidth]{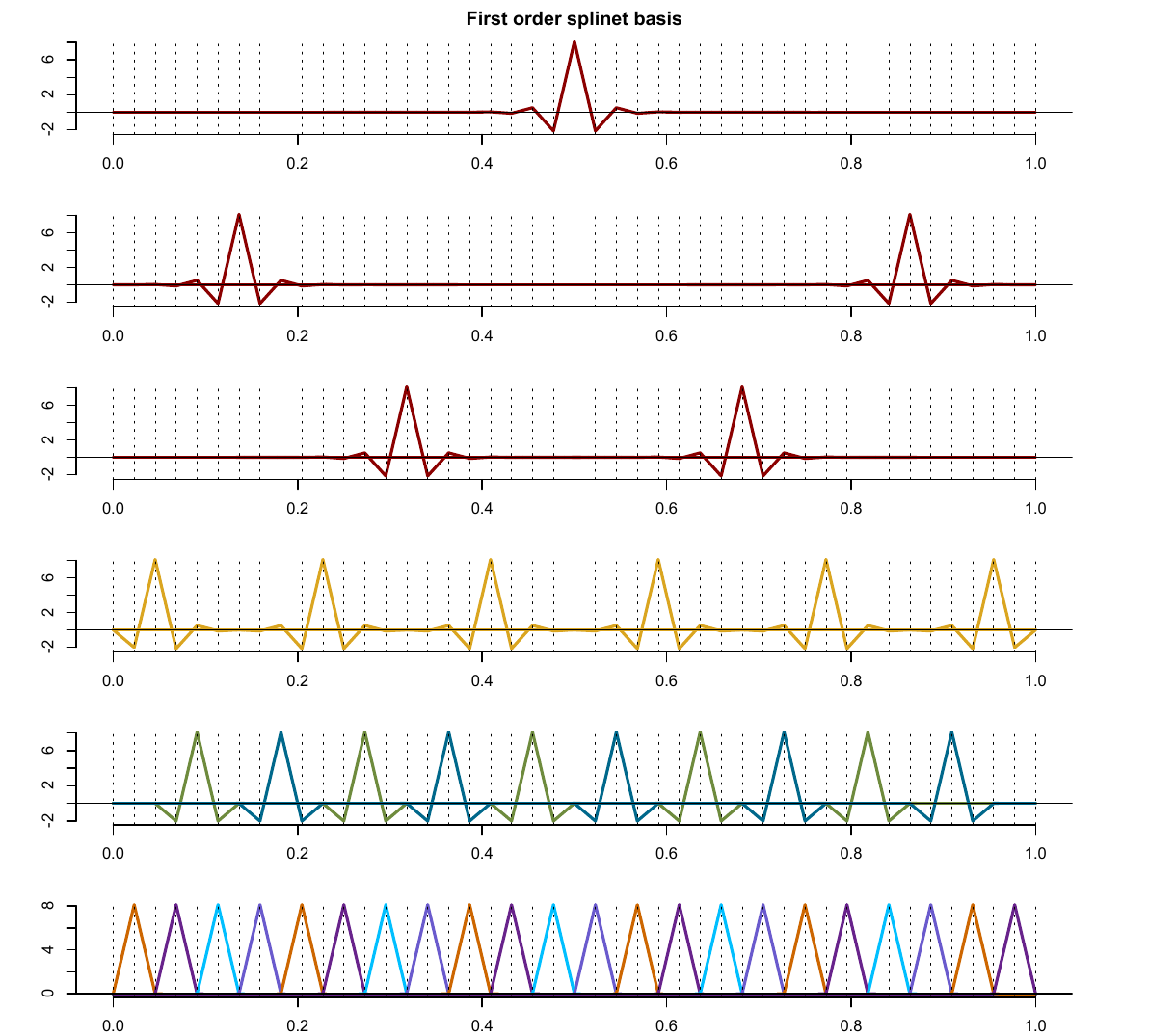}
   \includegraphics[width=0.495\textwidth, height=0.42\textwidth]{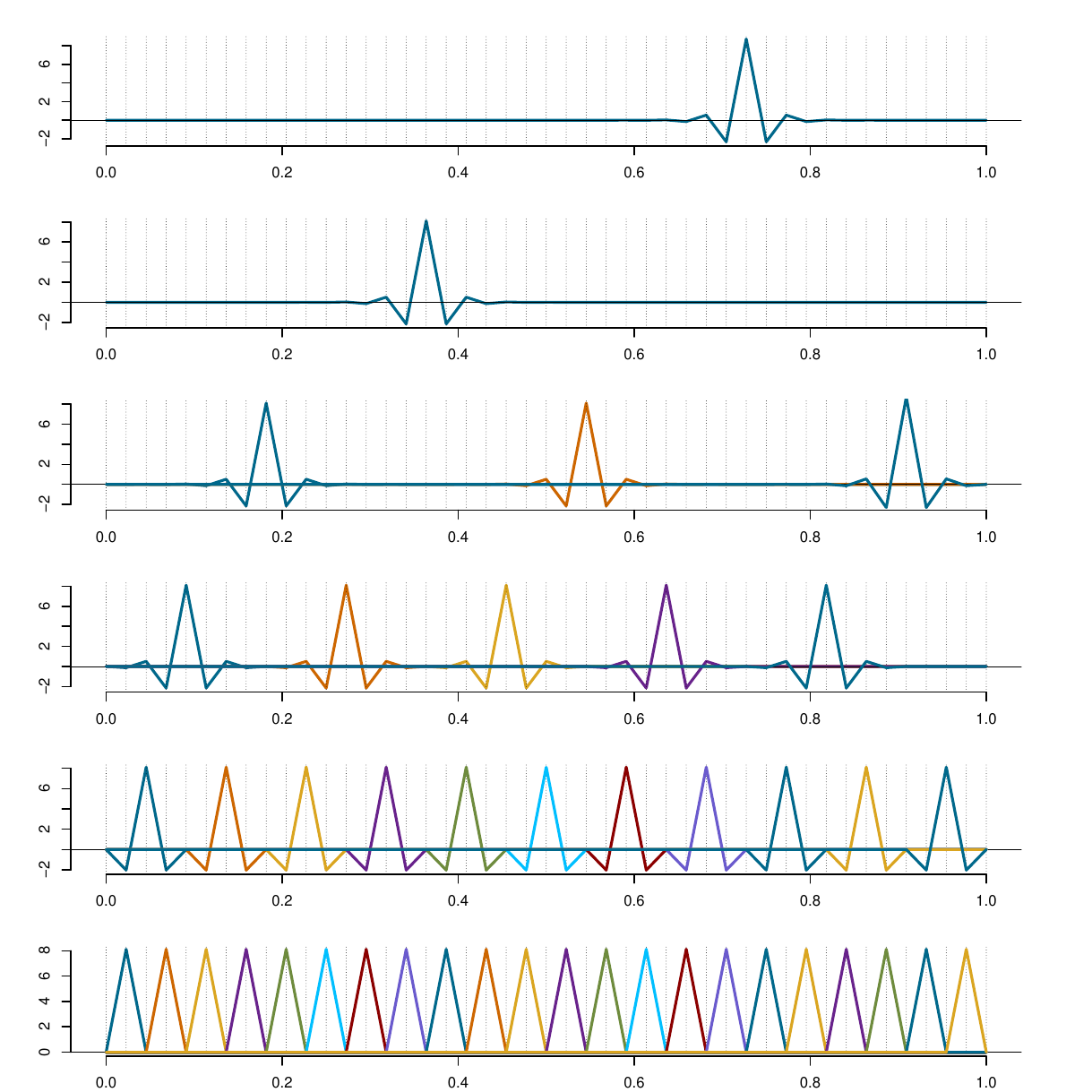} 

    \caption{{\it Top:}\, The splinets for the non-dyadic first order case: the implementation in the package {(\it left)} and the one build in the example {(\it right}). 
   }
  \label{fig:Non-dyadic}
\end{figure}

\section{Example -- simple functional data analysis with {\tt Splinets}}
We conclude this work with an example of functional data analysis utilizing the tools implemented in the package. 
The functional data are simulations from a rather complicated stochastic model, the so-called Slepian model that was developed in \cite{Podgorski:2015aa}. 
A computationally intense model involving a high-dimensional Gibbs sampler was developed for responses of a truck to a non-Gaussian loads at a transient (a high impact event on the road profile).
A functional data sample of size 1000 is presented in Figure~\ref{fig:TruckTire}~{\it (Top-left)}. 
The data are discretized over 4095 equally spaced arguments.
Here, instead of the complicated model, we want to fit a linear functional data model based on the Karhunen-Lo$\rm \grave{e}$ve decomposition of the process that is residual to the mean function.

\begin{figure}[t!]

    \includegraphics[width=0.495\textwidth, height=0.42\textwidth]{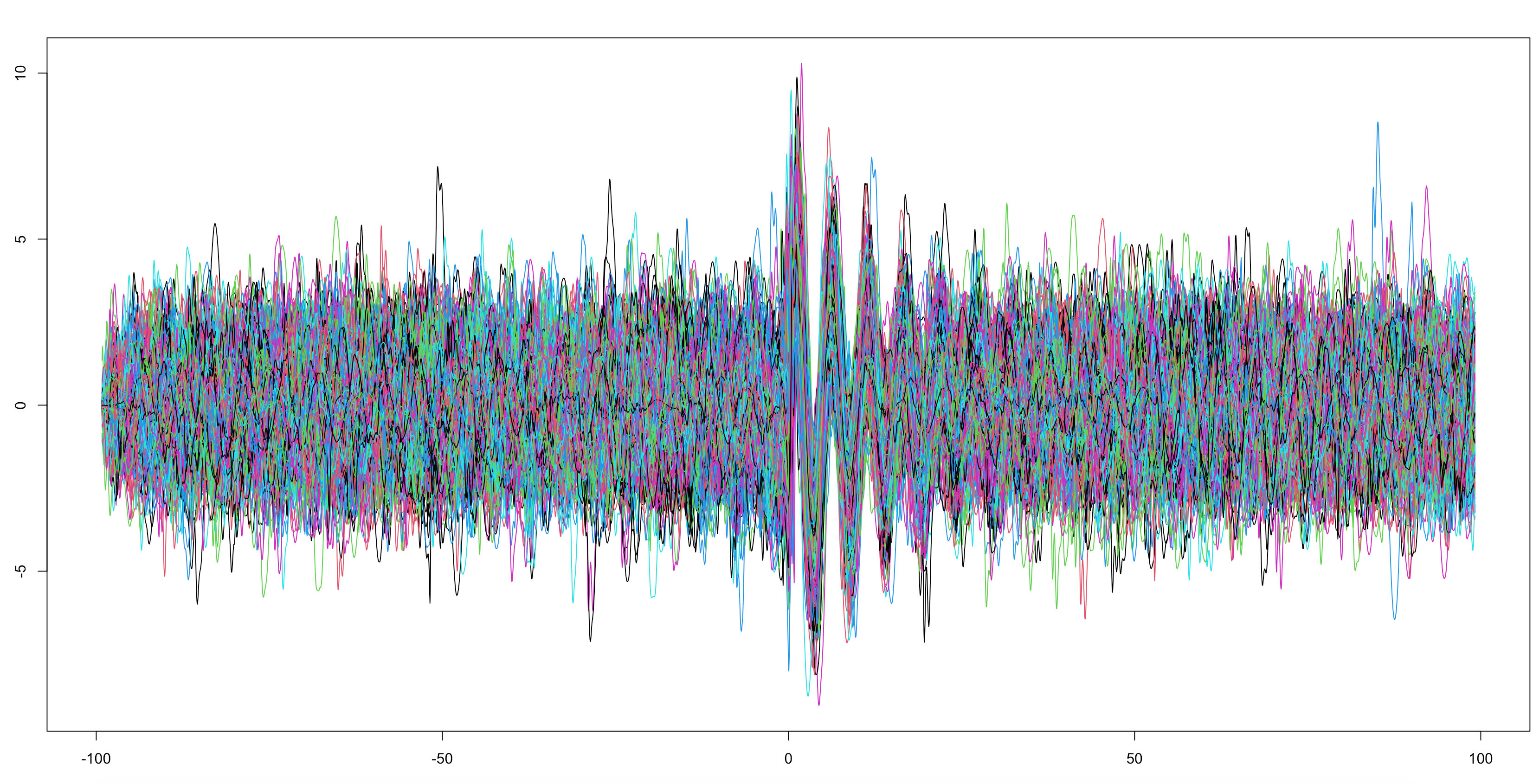}  \includegraphics[width=0.5\textwidth]{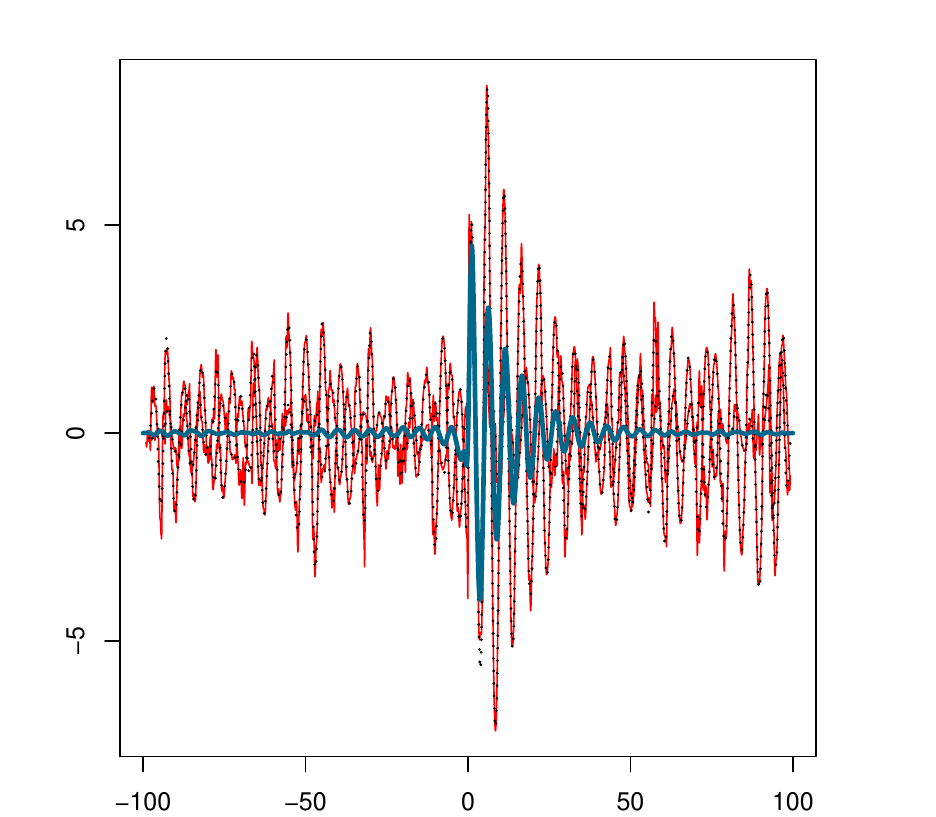}\vspace{-3mm}\\
     \raisebox{6mm}{\includegraphics[width=0.5\textwidth]{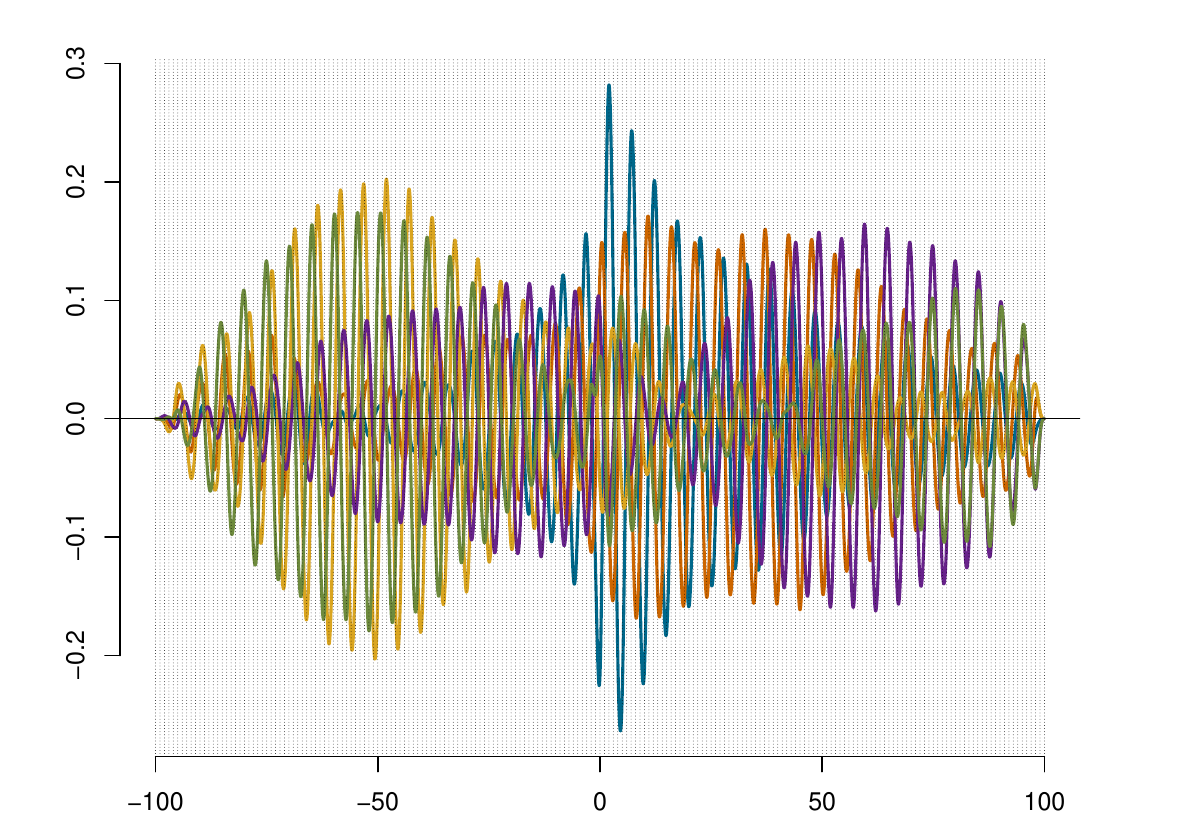}} \includegraphics[width=0.495\textwidth, height=0.42\textwidth]{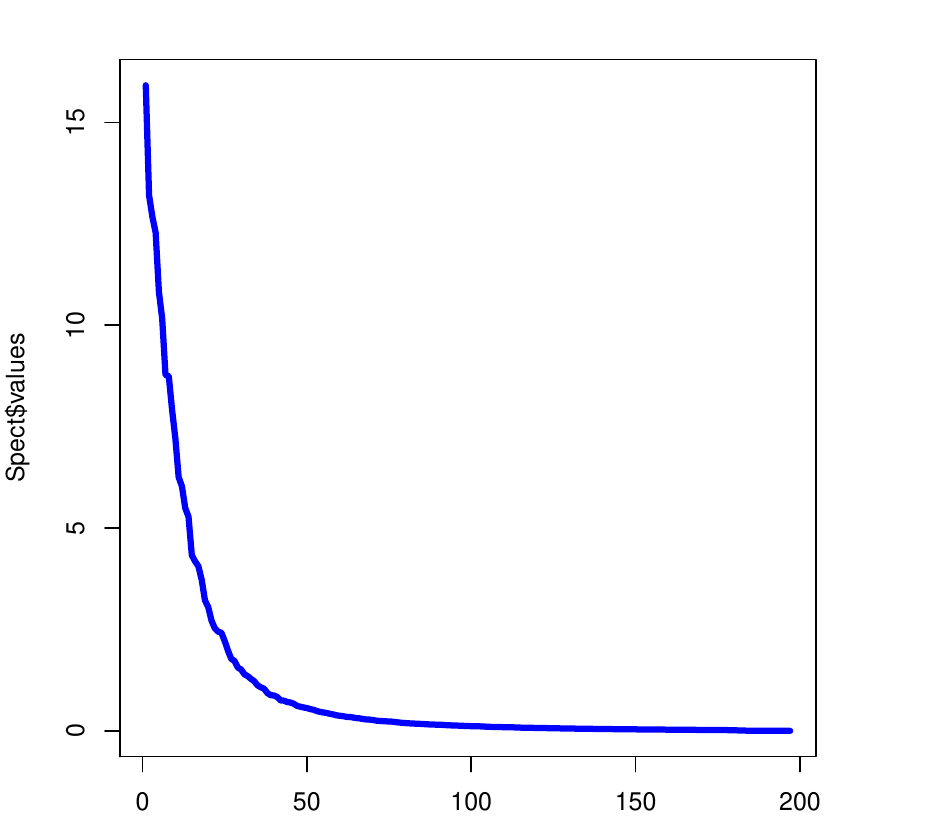} \vspace{-6mm}\\
     \includegraphics[width=0.495\textwidth, height=0.42\textwidth]{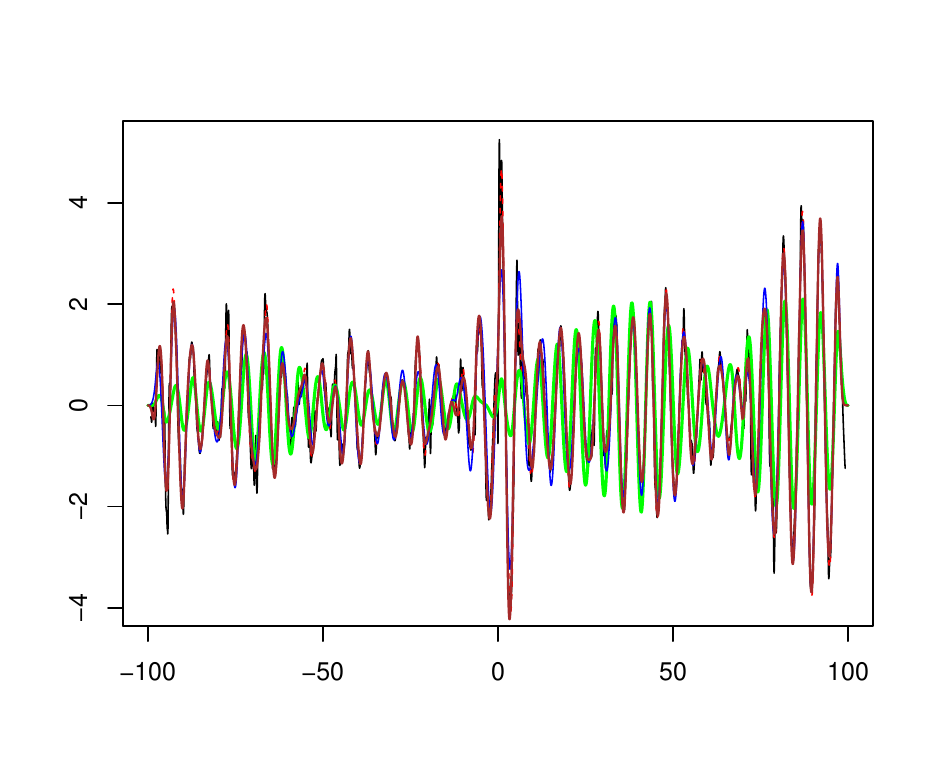} \includegraphics[width=0.495\textwidth, height=0.42\textwidth]{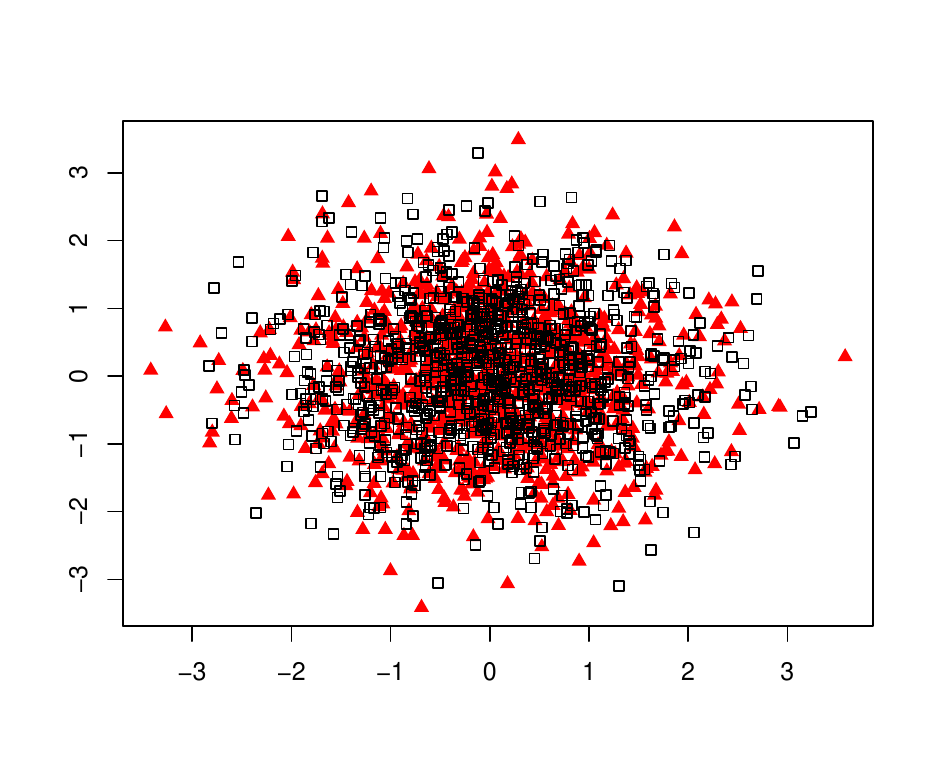}   \vspace{-5mm} 
    \caption{{\it Top:}\, 1000 truck responses to a transient in a road profile {\it (left)}, the mean truck response (thick line) and a single functional data (solid line) and its spline representation (dotted lines)  {\it (right)}.
   {\it Middle:}  The five eigenfunctions {\it (left)} corresponding to the first five largest eigenvalues and the ordered eigenvalues {\it right}. 
   {\it Bottom-left:} A single functional datum (thin-solid line), its reconstruction using full space of splines (dashed-line), and three reconstruction using the 10, 50, and 100 dimensional subspaces
   spanned by the eigenfunctions with the corresponding largest eigenvalues. {\it Bottom-right:} The scatter plot of standardized coefficients in the  Karhunen-Lo$\rm \grave{e}$ve decomposition for the first two eigenvalues (triangles) compared with a scatter plot of a standard normal bivariate distribution (squares).
   }
  \label{fig:TruckTire}
\end{figure}

In the first step, we represent data as splines using {\tt project()} to project data into splines $X_i(t)$, $i=1,\dots, 1000$ of the third order with 201 equally spaced knots.
The first five functional data  together with the mean spline $\mu(t)$ are presented in Figure~\ref{fig:TruckTire}~{\it (Top-right)} and the code applied to the {\it matrix} representation of the data in {\tt Truck} that performs these evaluation is 
\begin{lstlisting}
knots=seq(-100,100, by=1)
TruckProj=project(Truck,knots)
MeanTruck=matrix(colMeans(TruckProj$coeff),ncol=dim(TruckProj$coeff)[2])
MeanTruckSp=lincomb(TruckProj$basis,MeanTruck)
\end{lstlisting}
In the next step, we perform the spectral decomposition of $X_i(t)$ by estimating the eigenvalues $\lambda_i$'s and the corresponding eigenfunctions $f_i(t)$.
 This allows for representing the data according to the Karhunen-Lo$\rm \grave{e}$ve decomposition
$$
X(t)=\mu(t) +\sum_{i=1}^\infty \sqrt{\lambda_i} Z_i f_i(t),
$$
where $Z_i$ are uncorrelated random variables with the mean zero and variance one. 
The following code performs estimation of the spectral decomposition and reconstruct a single functional data point using 10, 50, and 100 eigenfunctions, while in Figure~\ref{fig:TruckTire}~{\it (Middle-left)} we present the estimated eigenvalues in the decreasing order and the first five eigenfuctions are presented in Figure~\ref{fig:TruckTire}~{\it (Middle-right)}.

\begin{lstlisting}[firstnumber=5]
Sigma=cov(TruckProj$coeff)
Spect=eigen(Sigma,symmetric = T)
EigenFunct=lincomb(TruckProj$basis,t(Spect$vec))

C=TruckProj$coeff%*%Spect$vec
D1R10=lincomb(EgnFncts10,C[1,1:10,drop=F])
D1R50=lincomb(EgnFncts50,C[1,1:50,drop=F])
D1R100=lincomb(EgnFncts100,C[1,1:100,drop=F])
\end{lstlisting}

Finally, the distributional properties of $Z_i$'s variables can be investigated by considering the columns of $1000 \times 197$ matrix {\tt C} standardized by the means and the square roots of the eigenvalues as computed in the code
\begin{lstlisting}[firstnumber=13]
Z1=(C[,1]-matrix(MeanTruck%*%Spect$vec[,1],nrow=1000))/sqrt(Spect$values[1])
Z2=(C[,2]-matrix(MeanTruck%*%Spect$vec[,2],nrow=1000))/sqrt(Spect$values[2])
\end{lstlisting}

A scatter plot of {\tt Z1} and {\tt Z2} is shown in Figure~\ref{fig:TruckTire}~{\it (Bottom-right)} together with a scatter plot of simulated values from the bivariate standard normal variable. 
We observe that coefficients obtained from the data resemble the ones for the standard normal although the former appears to be slightly more leptokurtic.

\section{Acknowledgement} 
The author would like to thank Xijia Liu who is the co-author of the {\tt Splinets}  package for helpful discussions on some parts of the paper. 

\address{Krzysztof Podg\'orski\\
  Department of Statistics\\
   Lund University\\
   Sweden\\
  \email{Krzysztof.Podgorski@stat.lu.se}}

\end{article}

\end{document}